\ifdefined\pdfminorversion\pdfminorversion=7\fi
\documentclass[12pt,a4paper]{article}

\usepackage{amsmath,amssymb,amsthm}
\usepackage{mathtools}
\usepackage{booktabs}
\usepackage{multirow}
\usepackage{enumitem}
\usepackage{graphicx}
\graphicspath{{output/}}
\usepackage{float}
\usepackage[round]{natbib}
\usepackage{setspace}
\usepackage{xcolor}
\usepackage[normalem]{ulem}
\usepackage{placeins}
\definecolor{deepblue}{rgb}{0,0,0.35}

\usepackage{restud_20260812}
\usepackage[colorlinks=true,linkcolor=deepblue,citecolor=deepblue,urlcolor=deepblue]{hyperref}

\newcommand{\E}{\mathbb{E}} 
\DeclareMathOperator{\Var}{Var}
\DeclareMathOperator{\Cov}{Cov}
\newcommand{\LATE}{\mathrm{LATE}}
\newcommand{\Wald}{\mathrm{Wald}}
\newcommand{\MTE}{\mathrm{MTE}}
\newcommand{\ind}{\mathbf{1}}
\newcommand{\bZ}{\boldsymbol{Z}}
\newcommand{\bz}{\boldsymbol{z}}
\newcommand{\bg}{\boldsymbol{g}}
\newcommand{\EGMM}{\mathrm{EGMM}}
\newcommand{\CUE}{\mathrm{CUE}}
\newcommand{\TwoS}{\mathrm{TSGMM}}

\newcommand{\norm}[1]{\left\lVert #1 \right\rVert}
\newcommand{\inner}[2]{\left\langle #1,#2\right\rangle}
\newcommand{\PLATE}{\mathcal{P}_{\LATE}}
\newcommand{\TLATE}{\mathcal{G}}
\newcommand{\TLATEbar}{\overline{\mathcal{G}}}
\newcommand{\Feas}{\mathcal{F}}
\newcommand{\Riesz}{\boldsymbol\Xi}
\newcommand{\eif}{\phi^{\mathrm{eff}}}
\newcommand{\diag}{\mathrm{diag}}
\newcommand{\supp}{\mathrm{supp}}
\newcommand{\bWald}{\boldsymbol{\Wald}}

\newcommand{\bomega}{\boldsymbol{\omega}}
\newcommand{\blambda}{\boldsymbol{\lambda}}
\newcommand{\bpsi}{\boldsymbol{\psi}}
\newcommand{\bgamma}{\boldsymbol{\gamma}}
\newcommand{\bphi}{\boldsymbol{\phi}}
\newcommand{\bSigmaZ}{\boldsymbol{\Sigma}_Z}
\newcommand{\bOmega}{\boldsymbol{\Omega}}
\newcommand{\bGamma}{\boldsymbol{\Gamma}}
\newcommand{\bA}{\mathbf{A}}
\newcommand{\bW}{\boldsymbol{W}}
\newcommand{\bG}{\mathbf{G}}
\newcommand{\bM}{\mathbf{M}}

\newcommand{\bUpsilon}{\boldsymbol{\Upsilon}}
\newcommand{\bSigma}{\boldsymbol{\Sigma}}
\newcommand{\bnu}{\boldsymbol{\nu}}

\newcommand{\bkappa}{\boldsymbol{\kappa}}
\newcommand{\bc}{\mathbf{c}}

\newcommand{\bK}{\mathbf{K}}
\newcommand{\bu}{\mathbf{u}}
\newcommand{\bv}{\mathbf{v}}

\newcommand{\bw}{\mathbf{w}}
\newcommand{\bd}{\mathbf{d}}
\DeclareMathOperator*{\argmin}{arg\,min}

\newtheorem{assumption}{Assumption}
\newtheorem{lemma}{Lemma}
\newtheorem{proposition}[lemma]{Proposition}
\newtheorem{corollary}[lemma]{Corollary}
\newtheorem{theorem}[lemma]{Theorem}
\newtheorem{definition}{Definition}

\newtheorem{condition}{Condition}
\theoremstyle{remark}
\newtheorem{remark}{Remark}

\begin{document}

\bibliographystyle{econsoc}

{
\onehalfspacing
\title{\textbf{When Is GMM Actually LATE? Weighting Matrices and Causal Interpretation in Overidentified IV}\footnote{We thank Obleo Demandre, Vadim Marmer, Paul Schrimpf, and Kyungchul Song for helpful comments and suggestions. This paper previously circulated under the title ``Representativeness and Efficiency in Overidentified IV.''}}
\author{
Chun Pang Chow\\
Vancouver School of Economics\\
University of British Columbia\\
alexccp@student.ubc.ca \and
Hiroyuki Kasahara\\
Vancouver School of Economics\\
University of British Columbia\\
hkasahara@mail.ubc.ca
}
\date{August 12, 2026}
\maketitle

\begin{abstract}
Under heterogeneous treatment effects, the weighting matrix of overidentified IV-GMM selects the estimand, not just its precision. We characterize the selection exactly: for any parameter-free weighting-matrix map, the GMM estimand is a sum-to-one combination of instrument-specific Wald estimands, with closed-form weights and an exact non-negativity condition; efficient weighting adds a heterogeneity penalty. Continuously updated GMM exits this class through a variance-score remainder. Under positive regression dependence each Wald estimand is a convex combination of compliance-type LATEs, and under maintained validity a $J$-rejection indicates unequal Wald estimands rather than invalid instruments. We propose Representativeness Targeting (RT), which estimates a researcher-specified convex combination of the Wald estimands without imposing a common coefficient across moments; RT weights compliance types nonnegatively, attains the local asymptotic minimax bound for its target, and extends to unreachable policy targets via projection with identification-gap bounds.  
In Tennessee STAR, we find the $J$-test rejects the Wald-estimand equality while the heterogeneity penalty pulls the efficient-GMM estimate
substantially below 2SLS; in a patent-leniency design, RT delivers a policy-relevant surrogate that standard GMM weightings miss.
\end{abstract}


}

\newpage

\section{Introduction} 

A researcher with several valid instruments for one binary treatment must
choose how to combine them: two-stage least squares (2SLS), two-step GMM
(TSGMM), iterated efficient GMM (EGMM), or continuously updated GMM (CUE). In
the constant-coefficient linear model this choice is a matter of precision
alone: every weighting matrix delivers the same estimand, and textbook advice
is to weight for efficiency. Under heterogeneous treatment effects, however,
that separation between estimand and precision breaks down. In our reanalysis of the Tennessee STAR experiment, the same data yield an estimated small-class effect of 0.224 standard deviations (math scores standardized within grade) under 2SLS, 0.205 under TSGMM, 0.172 under EGMM, and 0.176 under CUE. Under constant effects the four would agree in the
limit; under heterogeneity they answer four different questions. In
multi-instrument environments, the GMM weighting matrix determines whose
treatment effects the estimand represents.

This paper asks three questions about that fact. First, which combination of
the instrument-specific Wald estimands does a given weighting matrix deliver?
Second, when are the implied weights on the underlying treatment effects
nonnegative, so that the estimand remains causally interpretable? Third, can
the researcher choose the weighting, or drop the common-coefficient
restriction altogether, so that the estimand is one specified in advance and
estimated efficiently? We answer all three: the estimand is a sum-to-one
combination of the Wald estimands, with closed-form weights that the weighting
matrix controls; nonnegativity of the implied weights reduces to a checkable
condition; and Representativeness Targeting, an estimator built directly on
the instrument-specific moments, returns the choice of estimand to the
researcher with full efficiency for the declared target. Answering the three questions gives
concrete causal content to the principle that under misspecification the
choice of estimator is a choice of estimand \citep{HallInoue2003,
AndrewsChen2025}, and it extends to the GMM class a question that
\citet{ImbensAngrist1994}  classically answered, and that
\citet{BlandholBonneyMogstadTorgovitsky2022} and \citet{Sloczynski2026}
recently refined for 2SLS with covariates: when is the estimand actually a local
average treatment effect (LATE).

The analysis rests on compliance types, the multi-instrument generalization of the compliance taxonomy of \citet{AngristImbensRubin1996}. We observe an outcome $Y$, a binary treatment $D$, and $L$ binary instruments $Z_1, \dots, Z_L$, indexed by $\ell$; each compliance type, indexed by $t$, records how an individual's treatment would respond to every configuration of the instruments. Applying \citet{ImbensAngrist1994} one instrument at a time gives $L$ Wald estimands: instrument $\ell$'s Wald estimand, written $\Wald_\ell$, is the difference in mean outcomes between the instrument's two values divided by the corresponding difference in mean treatment take-up. Under instrument validity, relevance, monotonicity, and the realized-treatment (consistency) conditions, each Wald estimand is a weighted sum of compliance type-specific average treatment effects, with weights $\alpha_{\ell,t}$ that sum to one across types, $\alpha_{\ell,t}$ being the weight that instrument $\ell$'s Wald estimand places on compliance type $t$ (Proposition \ref{prop:wald_decomp}). When instruments are dependent, these weights can be negative. Positive regression dependence (PRD; \citealp{Lehmann1966}), a positive-association restriction on the instruments' joint distribution, rules this out: it holds under mutually independent instruments or, more generally, whenever the instruments are nondecreasing functions of a common scalar variable, as in examiner and judge leniency designs where the $\ell$-th instrument indicates leniency exceeding the $\ell$-th of an increasing set of cutoffs, and it makes each Wald estimand a convex combination of type-specific effects (Proposition \ref{prop:prd}). The question is how GMM combines these $L$ causally interpretable building blocks.

The first answer is a decomposition. Let $P$ denote the population distribution of the data, let $\gamma_\ell$ be the covariance between the treatment and instrument $\ell$ (the instrument's first-stage strength), and let $\bgamma = (\gamma_1, \dots, \gamma_L)'$ stack these covariances. A weighting-matrix map $\bW(P)$ assigns a positive-definite weighting matrix to each population distribution; the matrix may depend on the distribution but not on the coefficient being estimated, and we suppress the argument $P$ when it is clear. For any such map, the GMM estimand is a linear combination of the $L$ Wald estimands whose weights sum to one, which we call a Wald-share combination:
$\beta^{\mathrm{GMM}}_{\bW} = \sum_{\ell=1}^{L} \lambda_\ell\,\Wald_\ell$ with
$\lambda_\ell = \gamma_\ell\,[\bW\bgamma]_\ell / \bgamma'\bW\bgamma$, where $[\bW\bgamma]_\ell$ is the $\ell$-th entry of the vector $\bW\bgamma$ (Proposition \ref{prop:weighted_iv}). The share $\lambda_\ell$ is negative exactly when $\gamma_\ell$ and $[\bW\bgamma]_\ell$ have opposite signs, which off-diagonal weighting can produce even when every instrument has a positive first stage. The algebra parallels the statistical decompositions of \citet{Andrews2019}, who characterizes 2SLS and two-step GMM estimands as functions of single-instrument estimands while remaining deliberately agnostic about their causal content; here the content is causal, because under PRD each $\Wald_\ell$ is itself a convex combination of compliance type-specific effects, so the composite weight on type $t$ is $\psi_t = \sum_\ell \lambda_\ell\,\alpha_{\ell,t}$. The decomposition also delimits its own scope. Because the CUE criterion re-weights as the coefficient being estimated moves \citep{HansenHeatonYaron1996}, its estimand splits into a Wald-share component plus an extra term, a variance-score remainder, and so is not a Wald-share combination at all. That failure is the GMM-class counterpart of the pathologies of limited-information maximum likelihood, whose estimand can leave the range of the per-instrument estimands, documented by \citet{Kolesar2013} and \citet{Andrews2019}.

The decomposition matters because heterogeneity makes the moment system misspecified. GMM imposes a single coefficient $\beta$ on all $L$ moment conditions, requiring the covariance between each instrument and the residual $Y - \beta D$ to be zero. When Wald estimands differ across instruments, no single $\beta$ satisfies all $L$ conditions, which is misspecification in the sense of \citet{HallInoue2003}, and the weighting matrix then dictates which sum-to-one combination the estimand targets \citep{AndrewsChen2025}. With variance-adaptive maps, weighting matrices built from the second moments of those residuals, the dependence becomes systematic. Efficient GMM inverts the residual second-moment matrix evaluated at its own probability limit under misspecification (its pseudo-true value), the coefficient at which further iteration locally settles (Definition \ref{def:egmm}). This creates a heterogeneity penalty. In the diagonal specialization, where the relevant second-moment matrices are diagonal, an instrument's EGMM share relative to its 2SLS share is inversely proportional to the second moment of that instrument's residual, scaled by the instrument's variance and evaluated at EGMM's own limit (Corollary \ref{cor:egmm_diag}). Holding fixed the coefficient at which the residuals are evaluated and the other components of that moment, it rises with the dispersion of treatment effects among the instrument's compliers, so EGMM downweights instruments whose complier groups have dispersed effects. The residual variance, an object with no causal meaning, thereby shapes whose effects the estimand represents. Negative shares are a separate, off-diagonal mechanism: under diagonal weighting the shares stay nonnegative, and it is the sign condition above that can push them below zero.

The second answer separates three layers of weights: the within-instrument weights $\alpha_{\ell,t}$ (inner), the cross-instrument shares $\lambda_\ell$ (outer), and the composite weights $\psi_t$. Under PRD the inner weights are nonnegative; the outer shares can be negative under the exact sign condition above; and causal interpretability turns on the composite weights, which negative outer shares threaten.   In our patent-examiner application, the estimated TSGMM and EGMM shares $\hat\lambda_\ell$ (hats denote sample estimates) are negative on the fourth and fifth of the six threshold instruments, counted from the lowest cutoff, with the highest-cutoff instrument positive. 

Heterogeneity also changes what the overidentification test tests. Under maintained instrument validity, rejection by the Hansen $J$-statistic, the standard test of the overidentifying restrictions, indicates unequal instrument-specific Wald estimands, not invalid instruments. Against local alternatives, the noncentrality of the limiting chi-square distribution is one and the same quadratic form in the local movements of the Wald estimands, each scaled by its instrument's first-stage covariance, for TSGMM, EGMM, and CUE, and it vanishes exactly when those movements coincide; under fixed heterogeneity, each statistic divided by the sample size converges to a positive estimator-specific constant, with CUE's constant no larger than EGMM's (Proposition \ref{prop:jtest}). The interpretive message is in \citet{Lee2018}, who shows that multiple LATEs misspecify the 2SLS moments and warns that the $J$-test rejects without implying invalidity,  and in \citet{AndrewsChen2025}, whose $J$-interval measures the estimates attainable across weightings at a given standard error. The local chi-square calculation is standard GMM theory; what the proposition adds is the Wald-movement reading of the noncentrality, its uniformity across the three estimators, and the estimator-specific fixed-alternative limits. Validity itself remains separately testable under heterogeneity \citep{Kitagawa2015}. In STAR, the $J$-test rejects Wald-estimand equality, and conditional on the randomized design and the maintained exclusion and no-interference conditions we read this as heterogeneity, not invalidity.

The third answer is the paper's proposal. Representativeness Targeting (RT) abandons the common-coefficient architecture: rather than fitting a single coefficient to all $L$ moments, RT estimates each instrument-specific Wald estimand separately and averages the estimates with researcher-specified weights $\omega_\ell(P)$, one per instrument, nonnegative and summing to one (Definition \ref{def:rt}). Under PRD, the RT estimand is a convex combination of type-specific LATEs for every admissible choice of weights (Proposition \ref{prop:rt_types}); the researcher chooses whose effects the estimand represents, and the estimator follows that choice rather than the residual-variance criterion. RT is not a departure from GMM so much as a disciplined selection within it: when every weight is strictly positive, the diagonal weighting matrix with $\ell$-th diagonal entry $\omega_\ell(P)/\gamma_\ell(P)^2$ implements the same estimand inside the GMM class (Corollary \ref{cor:witness}).

Choosing an estimand raises the question of what efficiency means for it. The answer separates weighting maps that agree at the true distribution but respond differently to perturbations of it. The estimator's first-order sensitivity to the data, its influence function, has a baseline Wald component plus a weight-drift component, the term that tracks how the weights themselves move with the distribution; the RT plug-in estimator, which substitutes sample analogues for the population quantities, attains the local asymptotic minimax efficiency bound for its target under the regularity, tangent-compatibility, and local-risk conditions collected in Appendix \ref{app:regularity} \citep[cf.][]{Chamberlain1987} (Proposition \ref{prop:eif-decomp}), and with a single instrument the bound collapses to \citet{Frolich2007}'s LATE bound. More is true. Call a GMM weighting map regular for a chosen target when its limiting behavior is stable under local perturbations of the data-generating process; regularity holds if and only if the resulting GMM estimand matches the target at the true distribution and their first-order responses agree to local changes that respect the LATE assumptions (Proposition \ref{thm:tangent-matching}). Matching the value is not enough: a variance-adaptive map is regular for a declared target only when its first-order response also matches, a condition its construction does nothing to enforce; CUE sits outside the Wald-share class and outside the scope of the classification altogether.

Some policy targets are not reachable: a policy-relevant treatment effect \citep{HeckmanVytlacil2005} specifies desired weights over compliance types, and the weights attainable by combining the $L$ instruments form at most an $(L-1)$-dimensional subset of all nonnegative type weights summing to one, a set whose dimension is the number of compliance types minus one. RT extends by projection: projecting the desired weights onto the attainable set, in a metric the researcher states, defines a point-identified surrogate; the plug-in is regular and efficient for it under the stated regularity conditions, including local stability of which constraints bind at the projection (Proposition \ref{prop:lam-projected}), and the gap between surrogate and target is set-identified with bounds under explicit restrictions on the type effects (Propositions \ref{prop:id-gap-distfree}, \ref{prop:id-gap-mte}).

The Tennessee STAR experiment provides the diagnostic benchmark. School-specific Wald estimates range from $-0.93$ to $+0.90$ standard deviations across the $L=80$ schools (Figure~\ref{fig:m1_late_landscape}). Reflecting treatment effect heterogeneity, the $J$-test
rejects Wald-estimand equality, and the four estimates separate as the
closed-form weights predict, with the heterogeneity penalty pulling EGMM 23
percent below 2SLS. The application shows the penalty operating in a
randomized design, where the maintained validity conditions are most credible
and the heterogeneity reading of the $J$-rejection is cleanest.

\begin{figure}[htbp]
\centering
\includegraphics[width=0.65\textwidth]{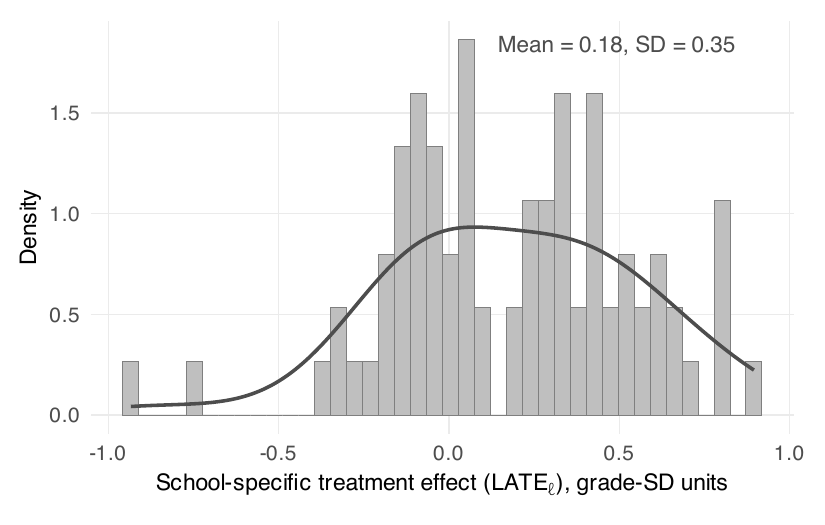}
\caption{Distribution of the $L = 80$ school-specific Wald estimands ($\LATE_\ell$) on grade-standardized math scores; histogram with kernel-density overlay (mean $0.18$, SD $0.35$, grade-SD units).}
\label{fig:m1_late_landscape}
\end{figure}

The patent-examiner leniency design \citep{FarreMensaHegdeLjungqvist2020} shows the estimand drift and the remedy. With cumulative-threshold instruments from a scalar leniency ladder, the estimated variance-adaptive weightings concentrate on the lowest threshold and assign negative shares to the fourth and fifth of the six cumulative thresholds (for CUE, through the Wald-share component of its estimand); the resulting estimates fall below every threshold-specific Wald estimate, so they match no convex weighting of the instrument-specific estimands.  The RT projection plug-in instead delivers the surrogate that best approximates the policy-relevant weighting, with identification-gap bounds and confidence regions.

Our characterization results sit in the literature that decomposes IV estimands into subgroup effects. For projection estimators, \citet{ImbensAngrist1994} and \citet{AngristImbens1995} give the convex combination, \citet{Kolesar2013} the two-step-IV class, estimators that first collapse the instruments
into a single constructed instrument, together with the pathologies of the contrasting minimum-distance class, \citet{MogstadTorgovitskyWalters2021} the complier-group weights for 2SLS with multiple instruments, and \citet{Goff2024} vector monotonicity conditions for multiple instruments. Closest to our decomposition, \citet{GoldsmithPinkhamSorkinSwift2020} give the per-instrument (Rotemberg) weight decomposition that Proposition \ref{prop:weighted_iv} generalizes to the fixed-map weighting-matrix class. \citet{BlandholBonneyMogstadTorgovitsky2022} show that with covariates the 2SLS estimand is weakly causal only under rich covariates, and their supplement treats a vector of excluded instruments only under 2SLS weighting; \citet{Sloczynski2026} shows that the interacted specification restores convexity while the standard one can sign-flip conditional LATEs. \citet{BorusyakHull2024} show that design-based specifications produce convex ex-ante unit-level weights within a single estimand, a guarantee that does not extend to the cross-instrument shares GMM places on multiple Wald estimands. Koles\'ar's unbiased jackknife IV estimator (UJIVE) repairs the many-instrument bias of the fixed two-step-IV target; RT instead makes the target the researcher's choice. The pedigree of negative weights on treatment effects is older still: \citet{HeckmanVytlacil2005} and \citet{HeckmanUrzuaVytlacil2006} document negative weights on marginal treatment effects for IV estimands, and our sign-reversal analysis is the instrument-specific Wald-share counterpart. That literature fixes projection weighting or, in Koles\'ar's minimum-distance analysis, varies a weight matrix on the reduced-form proportionality restriction; this paper makes the moment-system weighting map $\bW(P)$ itself the object of causal study.

A second strand studies moment-condition models under misspecification. \citet{HallInoue2003} and \citet{HansenLee2021} develop the pseudo-true (misspecification-limit) asymptotics, \citet{Lee2018} applies them to 2SLS with multiple LATEs, \citet{Andrews2019} maps the geometry of IV estimands while leaving their causal interpretation to future work, and \citet{AndrewsChen2025} synthesize the field's premise, that the estimator choice is an estimand choice, and recommend that researchers say what their estimators estimate. The present paper takes the researcher's stated target as the primitive: it supplies the causal content of the estimands, the diagnostic content of the $J$-statistic, and an estimator whose estimand is the stated target.

A third strand chooses estimands rather than estimators. \citet{Frolich2007} gives the efficiency bound our single-instrument case recovers; \citet{HiranoImbensRidder2003} study efficient estimation of weighted effects under unconfoundedness; \citet{AngristFernandezVal2013} reweight covariate-specific LATEs toward other populations; \citet{MogstadSantosTorgovitsky2018} bound fixed policy targets; \citet{PoirierSloczynski2025} quantify ex post how representative a weighted estimand is. RT's contribution is specific to overidentified IV: researcher-declared weights over the instrument-specific Wald estimands, inference that carries the weight-drift term, the classification of which GMM maps implement the target, with the construction of a weighting matrix that reproduces RT inside GMM and tangent matching, and, when the target is unreachable, a projection onto the attainable set in a stated metric with bounds on the gap.

Section~2 sets up compliance types, the Wald decomposition, and PRD. Section~3 develops the GMM weight characterization, the heterogeneity penalty, and the $J$-test diagnostic. Section~4 introduces RT, its efficiency theory, and the regularity classification. Section~5 develops projection targeting. Section~6 reports the two applications, and Section~7 concludes. Additional Assumptions and Proofs are in the Appendix and the Online Appendix.


\section{Framework}\label{sec:framework}

\subsection{Setup and compliance types}\label{sec:compliance-types}

The data consist of a scalar outcome $Y_i$, a binary treatment decision $D_i \in \{0,1\}$, and $L \geq 2$ binary instruments $Z_{1i}, \ldots, Z_{Li}$, collected into the vector $\bZ_i = (Z_{1i}, \ldots, Z_{Li})'$. Each individual has potential outcomes $\{Y_i(d) : d \in \{0,1\}\}$ and a potential treatment function $D_i(\cdot) : \{0,1\}^L \to \{0,1\}$, called the individual's \emph{compliance type}, which records the individual's treatment decision at each instrument configuration.\footnote{The compliance type generalizes the complier/always-taker/never-taker classification of \citet{AngristImbensRubin1996} to $L$ instruments and is a special case of the treatment response type in \citet{BaiHuangMoonSantosShaikhVytlacil2024}, who develop inference for treatment effects conditional on generalized principal strata under weaker monotonicity conditions and for nonbinary treatments.} 

\begin{assumption}[Joint independence]\label{ass:joint_indep}
$(Y_i(0), Y_i(1), D_i(\cdot)) \perp\!\!\!\perp \bZ_i$.
\end{assumption}

\begin{assumption}[Monotonicity]\label{ass:monotonicity}
For every $\ell \in \{1, \ldots, L\}$ and every $\bz_{-\ell} \in \{0,1\}^{L-1}$, $D_i(1, \bz_{-\ell}) \geq D_i(0, \bz_{-\ell})$.\footnote{Assumption~\ref{ass:monotonicity} is \citet{MogstadTorgovitskyWalters2021}'s actual monotonicity (their Assumption~AM). The stronger uniform-ordering form of \citet{Vytlacil2002}, equivalent to the latent-index representation, is invoked locally as Assumption~\ref{ass:latent_index} in Section~\ref{sec:mte_penalty}.}
\end{assumption}

\begin{assumption}[Realized treatment and exclusion]\label{ass:realized_exclusion}
The observed treatment and outcome satisfy $D_i = D_i(\bZ_i)$ and $Y_i = Y_i(D_i)$.
\end{assumption}

The data law $P$ governs the full vector $(Y_i(0),Y_i(1),D_i(\cdot),\bZ_i)$, with the observation $O_i\equiv (Y_i,D_i,\bZ_i)$ generated by the observed marginal. The expectation operator $\E$ and probability measure $\mathbb{P}$ leave the underlying data law $P$ implicit. Let $P_0$ denote the true law, $\hat P_n$ the empirical law of the observed data, and $P_{n,h}$ a contiguous local alternative in local asymptotics. 
Let $\supp(\bZ) \subseteq \{0,1\}^L$ be the support of $\bZ_i$, and let $\mathcal{T}$ denote the ambient compliance-type space: the set of functions $t : \{0,1\}^L \to \{0,1\}$ respecting Assumption~\ref{ass:monotonicity} and allowed by the maintained model. Type probabilities $\theta_t(P)\equiv \mathbb{P}(D_i(\cdot)=t)$ may be zero, and $\mathcal{T}_+(P)\equiv \{t\in\mathcal{T}:\theta_t(P)>0\}$ denotes the positive-support types at $P$. The type-specific local average treatment effect is $$\LATE_t(P) \equiv  \E[Y_i(1) - Y_i(0) \mid D_i(\cdot)=t]$$ for $t\in\mathcal{T}_+(P)$.
The positive-support \textit{complier group} for instrument $\ell$ is
\[
\mathcal{C}_\ell(P)\equiv
\{t\in\mathcal{T}_+(P):t(1,\bz_{-\ell})>t(0,\bz_{-\ell})\text{ for some }\bz_{-\ell}\},
\]
the positive-mass types whose treatment responds to instrument $\ell$.

For each instrument $\ell$, write $p^Z_\ell(P) \equiv \mathbb{P}[Z_{\ell i}=1]$ for the instrument probability, $\pi_\ell(P) \equiv  \E[D_i \mid Z_{\ell i} = 1] - \E[D_i \mid Z_{\ell i} = 0]$ for the first-stage coefficient, and $\rho_\ell(P) \equiv  \E[Y_i \mid Z_{\ell i} = 1] - \E[Y_i \mid Z_{\ell i} = 0]$ for the reduced-form coefficient. Let $p(z; P) \equiv \mathbb{P}(D_i = 1 \mid \bZ_i = z)$ denote the joint propensity score on $\supp(\bZ)$, abbreviated to $p(z)$ when context fixes the law. Define the covariance matrix of $\bZ$ as $\bSigmaZ(P)\equiv\Var_P(\bZ_i)$, with $\bSigmaZ(P_0)$ denoting its baseline value.

\begin{assumption}[Relevance and instrument distribution]\label{ass:relevance}
For each $\ell = 1, \ldots, L$, $p^Z_\ell(P_0) \in (0,1)$ and $\pi_\ell(P_0) > 0$, and  $\bSigmaZ(P_0)$ is positive definite.
\end{assumption}

Under Assumption~\ref{ass:relevance}, the Wald estimand for instrument $\ell$ is given by the functional $$\Wald_\ell(P) \equiv \frac{\rho_\ell(P)}{\pi_\ell(P)}=\frac{\E[Y_i \mid Z_{\ell i} = 1] - \E[Y_i \mid Z_{\ell i} = 0]}{\E[D_i \mid Z_{\ell i} = 1] - \E[D_i \mid Z_{\ell i} = 0]},$$ 
and the first-stage covariance vector $\bgamma(P) = (\gamma_1(P), \ldots, \gamma_L(P))'$ with $\gamma_\ell(P) \equiv  \Cov(D_i, Z_{\ell i})$ indexes the GMM-implied weights. Standard regularity conditions (i.i.d.\ sampling, SUTVA, exclusion, finite moments) and the additional rank, tangent-compatibility, and local-support conditions used in the LAM results are collected in Appendix~\ref{app:regularity}.

\subsection{Moment conditions}\label{sec:moments}

Using centered instruments, $Z_{\ell i} - \E[Z_\ell]$, is analytically equivalent to including a constant in the instrument matrix and concentrating out the intercept.\footnote{Concentrating out the intercept induces a transformed weighting matrix for the remaining moments via a Schur complement. We formulate our analysis directly within the centered moment system using this effective weighting matrix.} Solving the individual moment condition $\E[(Y_i - \beta_\ell D_i)(Z_{\ell i} - \E[Z_\ell])] = 0$ for $\beta_\ell$ yields the instrument-specific estimand $\Wald_\ell(P)$. 
When all Wald estimands coincide, all $L$ moment conditions hold simultaneously. When local average treatment effects vary across compliance types and the instruments weight those types differently, the Wald estimands $\Wald_\ell(P)$ differ across $\ell$, and no single choice of $\beta$ simultaneously satisfies $\E[(Y_i - \beta D_i)(Z_{\ell i} -\E[Z_\ell])] = 0$ for all $\ell$. In this scenario, the moment conditions are misspecified in the sense of \citet{HallInoue2003}.\footnote{\citet{AndrewsBarahonaGentzkow2025} develop a general framework for structural estimation under misspecification, of which this IV setting is a special case.} The choice of the weighting matrix $\bW$ then dictates which sum-to-one linear combination of the $\Wald_\ell(P)$ estimands the GMM estimand targets (Proposition~\ref{prop:weighted_iv}); different weighting matrices yield distinct estimands under misspecification \citep{HallInoue2003, AndrewsChen2025}.

\subsection{The Wald decomposition}\label{sec:wald_decomp}

\begin{proposition}[Wald decomposition]\label{prop:wald_decomp}
Under Assumptions~\ref{ass:joint_indep}, \ref{ass:monotonicity}, \ref{ass:realized_exclusion}, and~\ref{ass:relevance}, at every $P$ with $p^Z_\ell(P) \in (0,1)$ and $\pi_\ell(P) > 0$ for each $\ell$, the Wald estimand for instrument $\ell$ admits the decomposition
\begin{equation}\label{eq:wald_type}
  \Wald_\ell (P)= \sum_{t \in \mathcal{T}_+(P)} \LATE_t(P) \cdot \alpha_{\ell,t}(P),
\end{equation}
where $\sum_{t \in \mathcal{T}_+(P)} \alpha_{\ell,t}(P) = 1$ and
\[
  \alpha_{\ell,t}(P) = \frac{\theta_t(P) \cdot \varphi_{\ell,t}(P)}{\pi_\ell(P)}, \qquad \varphi_{\ell,t}(P)= \sum_{\bz_{-\ell}} \bigl[t(1, \bz_{-\ell})\, q_\ell(\bz_{-\ell}) - t(0, \bz_{-\ell})\, q^0_\ell(\bz_{-\ell})\bigr],
\]
with $t(z_\ell, \bz_{-\ell})$ the treatment decision of compliance type $t$ at instrument values $(z_\ell, \bz_{-\ell})$, $q_\ell(\bz_{-\ell}) = \mathbb{P}(\bZ_{-\ell} = \bz_{-\ell} \mid Z_\ell = 1)$, and $q^0_\ell(\bz_{-\ell}) = \mathbb{P}(\bZ_{-\ell} = \bz_{-\ell} \mid Z_\ell = 0)$.
\end{proposition}

\noindent The compliance type weight $\alpha_{\ell,t}(P)$ on type $t$ is proportional to $\theta_t(P) \cdot \varphi_{\ell,t}(P)$: the probability of being type $t$, multiplied by how much instrument $\ell$ shifts treatment for that type.\footnote{\citet{MogstadTorgovitskyWalters2021} decompose the combined 2SLS estimand into complier-group-specific treatment effects. Proposition~\ref{prop:wald_decomp} operates at a finer level, decomposing each instrument-specific Wald estimand separately into compliance-\emph{type}-specific contributions.} This is a structural decomposition under the maintained compliance-type model. The weights  $\alpha_{\ell,t}(P)$ may become  data functionals under the additional identifying structures used below, such as the ordered-threshold structure. The type contribution $\varphi_{\ell,t}(P)$ compares type $t$'s expected treatment under $Z_\ell = 1$ versus $Z_\ell = 0$, averaging over the other instruments. The weights $\alpha_{\ell,t}(P)$ can be negative when instrument dependence makes the indirect component sufficiently negative; Lemma~\ref{lem:di} isolates the source by splitting $\varphi_{\ell,t}(P)$ into direct and indirect components.
\begin{lemma}[Direct-indirect decomposition]\label{lem:di}
Under Assumption~\ref{ass:monotonicity}, $\varphi_{\ell,t}(P)= \varphi_{\ell,t}^D(P) + \varphi_{\ell,t}^I(P)$, where
\begin{align}
  \varphi_{\ell,t}^D(P) &= \sum_{\bz_{-\ell}} \bigl[t(1, \bz_{-\ell}) - t(0, \bz_{-\ell})\bigr]\, q_\ell(\bz_{-\ell}) \;\geq\; 0, \notag \\[3pt]
  \varphi_{\ell,t}^I(P) &= \sum_{\bz_{-\ell}} t(0, \bz_{-\ell})\, \bigl[q_\ell(\bz_{-\ell}) - q^0_\ell(\bz_{-\ell})\bigr]. \notag
\end{align}
\end{lemma}

\noindent The direct component $\varphi_{\ell,t}^D(P) \geq 0$ by monotonicity: switching $Z_\ell$ from 0 to 1 can only increase treatment for type $t$, holding $\bZ_{-\ell}$ fixed. The indirect component $\varphi_{\ell,t}^I(P)$ depends on how conditioning on $Z_\ell = 1$ shifts the distribution of $\bZ_{-\ell}$. Under independent instruments, $q_\ell = q^0_\ell$, the indirect component vanishes, and all numerators $\theta_t\varphi_{\ell,t}$ are non-negative; under Assumption~\ref{ass:relevance}, the normalized Wald weights $\alpha_{\ell,t}(P)$ are non-negative. If conditioning on $Z_\ell = 1$ shifts probability mass in $\bZ_{-\ell}$ toward configurations where type $t$ is less likely to take treatment at $Z_\ell=0$, the indirect component can be negative enough to overwhelm the direct component, and the Wald weights can be negative. Whether the weights are non-negative is therefore determined entirely by the instrument dependence structure and the positive first stage.

\subsection{Non-negative weights under positive regression dependence}\label{sec:prd_weights}

\begin{assumption}[Positive regression dependence (PRD)]\label{ass:prd}
For each $\ell = 1, \ldots, L$, the random vector $\bZ_{-\ell}$ is positively regression dependent on $Z_\ell$: $\E[f(\bZ_{-\ell}) \mid Z_\ell = 1] \geq \E[f(\bZ_{-\ell}) \mid Z_\ell = 0]$ for every nondecreasing $f : \{0,1\}^{L-1} \to \mathbb{R}$ \citep{Lehmann1966}.
\end{assumption}

\begin{proposition}[PRD implies non-negative inner Wald weights]\label{prop:prd}
Under Assumptions~\ref{ass:joint_indep}, \ref{ass:monotonicity}, \ref{ass:relevance}, and~\ref{ass:prd}, at every $P$ with $p^Z_\ell(P) \in (0,1)$ and $\pi_\ell(P) > 0$ for each $\ell$, $\alpha_{\ell,t}(P) \geq 0$ for every positive-support compliance type $t \in \mathcal{T}_+(P)$ and every instrument $\ell$.
\end{proposition}

PRD is a sufficient condition for positive weights in overidentified IV settings,\footnote{For $L = 2$, PRD reduces to non-negative covariance, recovering the sufficient condition in \citet{MogstadTorgovitskyWalters2021}; for $L \geq 3$, pairwise non-negative covariance no longer suffices. \citet{GoldsmithPinkhamSorkinSwift2020} require a same-sign first-stage condition and an exclusion restriction that jointly give each just-identified IV estimate a convex-combination interpretation in the shift-share setting; \citet{HahnKuersteinerSantosWilligrod2024} require share orthogonality or non-negatively correlated shocks. In a single-instrument, covariate-cell setting, \citet{PoirierSloczynski2025} develop a sharp-bound measure of implicit representativeness and relate non-negative-weight conditions to \citet{BlandholBonneyMogstadTorgovitsky2022}'s weakly-causal property; Proposition~\ref{prop:prd} is the multi-instrument analog. In applications with finite support, the FOSD inequalities have sample analogs, but they are diagnostics rather than a proof of the population restriction.} imposed on the joint distribution of the instruments. The following designs imply PRD:
\begin{enumerate}[label=(\roman*),itemsep=3pt]
  \item \emph{Independence Design:} if the instruments are randomly assigned independently of one another, the indirect component vanishes and the weights are non-negative.
  \item \emph{Cumulative/Common Factor Design:} when instruments derive from a common source $G_i$\footnote{The source $G_i$ is typically multi-valued and need not be monotone across all pairs of levels; distinct value comparisons may shift treatment in different directions for different individuals. The cumulative-threshold binarization replaces this with a scalar-ordering condition: if treatment is weakly increasing in $G_i$ for every individual, each separate-threshold first stage is nonnegative (Appendix~\ref{app:cumulative_monotone}).} such that $Z_{\ell i} = g_\ell(G_i)$ with each $g_\ell$ nondecreasing, as in examiner leniency designs with cumulative thresholds ($Z_{\ell i} = \ind\{G_i \geq \ell+1\}$), turning one instrument ``on'' makes the others more likely to be active.\footnote{The instruments are associated \citep{EsaryProschanWalkup1967}; association implies PRD.} The positive dependence shifts the indirect component upward and preserves non-negativity.
\end{enumerate}

PRD is a population-level design restriction, equivalent to first-order stochastic dominance of $\bZ_{-\ell} \mid Z_\ell = 1$ over $\bZ_{-\ell} \mid Z_\ell = 0$ for each $\ell$. PRD ensures non-negativity of the \emph{inner} compliance-type weights $\alpha_{\ell,t}(P)$ within each Wald estimand in~\eqref{eq:wald_type}; the \emph{outer} instrument weights depend on the weighting matrix map and can be negative under GMM, as Section~\ref{sec:egmm} shows.


\section{The GMM-implied weighting}\label{sec:egmm}

\subsection{GMM as affine combination of Wald estimands}

Applied practice often specifies not a weighting matrix directly but a \emph{weighting matrix map}: a rule that selects the weighting matrix as a functional of the data law $P$. Let $\PLATE$ denote the class of distributions compatible with the LATE model (Assumptions~\ref{ass:joint_indep}--\ref{ass:relevance}), and let $\mathcal{S}^{++}_L$ denote the cone of $L\times L$ positive definite matrices.

\begin{definition}\label{def:gmm-weight-map}
A {GMM weighting matrix map} is a functional $\bW(\cdot): \PLATE \to \mathcal{S}^{++}_L$ that depends on $P$ only through a finite-dimensional vector of moments $m(P) \in \mathbb{R}^d$, i.e., $\bW(P) = \tilde{\bW}(m(P))$ for some $\tilde{\bW}: U \to \mathcal{S}^{++}_L$ continuous on a neighborhood $U$ of $m(P_0)$ in $\mathbb{R}^d$.
\end{definition}

 All examples considered in this paper naturally satisfy this finite-dimensional requirement; for instance, the 2SLS map uses $m(P) = \mathrm{vech}(\bSigmaZ(P))$.

The GMM estimator with a weighting matrix map $\bW(\cdot)$ solves
\[
  \hat\beta^{\mathrm{GMM}}_{\bW} = \arg\min_\beta \, \bg_n(\beta)' \bW(\hat P_n) \, \bg_n(\beta),
\]
where $\bg_n(\beta) \equiv (g_{n,1}(\beta), \ldots, g_{n,L}(\beta))'$ with $g_{n,\ell}(\beta) \equiv  n^{-1} \sum_{i=1}^n (Y_i - \beta D_i)(Z_{\ell i} - \bar z_\ell)$ and $\bar z_\ell \equiv n^{-1}\sum_{i=1}^n Z_{\ell i}$.
Since the model is linear in $\beta$, $\bg_n(\beta) = \bg_n(0) - \beta \hat\bgamma$, where $\hat\bgamma = (\hat{\gamma}_1, \ldots, \hat{\gamma}_L)'$ with $\hat{\gamma}_\ell = n^{-1} \sum_{i=1}^n D_i (Z_{\ell i} - \bar z_\ell)$, the sample analog of $\gamma_\ell(P)$ from Assumption~\ref{ass:sampling}(e). The first-order condition yields
\begin{equation*}
  \hat\beta^{\mathrm{GMM}}_{\bW}  = \frac{\hat\bgamma' \bW(\hat P_n) \, \bg_n(0)}{\hat\bgamma' \bW(\hat P_n) \, \hat\bgamma}=\sum_{\ell=1}^L \lambda_\ell(\bW(\hat P_n );\hat P_n) \, \widehat\Wald_\ell,
\end{equation*}
where $\widehat\Wald_\ell\equiv \Wald_\ell(\hat P_n)$.  Let $\blambda(\bW(P); P)\equiv (\lambda_1(\bW(P); P),...,\lambda_L(\bW(P); P))'$, $\bWald(P)\equiv (\Wald_1(P),...,\Wald_L(P))'$, and  $\widehat\bWald\equiv (\widehat\Wald_1,...,\widehat\Wald_L)'$. The map $\blambda_{\bW}(P) \equiv \blambda(\bW(P); P)$, with components $\lambda_{\bW,\ell}(P)$, is the \emph{implicit weight map} of $\bW(\cdot)$.

\begin{proposition}[GMM as affine combination of Wald estimands]\label{prop:weighted_iv}
Let $\bW(\cdot)$ be a GMM weighting matrix map, and assume the data $\{O_i\}_{i=1}^n$ are i.i.d.\ from $P_0$ satisfying Assumption~\ref{ass:relevance} and Assumption~\ref{ass:sampling}(a),(d),(e).
\begin{enumerate}[label=(\alph*),leftmargin=*]
  \item \emph{Population formula.} For every $P$ in a neighborhood of $P_0$ and $\gamma_\ell(P) \neq 0$ for every $\ell$, the GMM estimand admits the affine decomposition
  \begin{equation}\label{eq:weighted_iv}
    \beta^{\mathrm{GMM}}_{\bW}(P) \;=\; \blambda(\bW(P); P)'\,\bWald(P) \;=\; \sum_{\ell=1}^L \lambda_\ell(\bW(P); P)\,\Wald_\ell(P),
  \end{equation}
  with implicit instrument weights
  \begin{equation}
    \lambda_\ell(\bW(P); P) \;=\; \frac{\gamma_\ell(P)\,[\bW(P)\,\bgamma(P)]_\ell}{\bgamma(P)'\,\bW(P)\,\bgamma(P)}, \qquad \sum_{\ell=1}^L \lambda_\ell(\bW(P); P) \;=\; 1.
  \end{equation}
  \item \emph{Sample formula and consistency.} With probability tending to one, the GMM estimator equals the sample-Wald affine combination
  \[
    \hat\beta^{\mathrm{GMM}}_{\bW} \;=\; \blambda(\bW(\hat P_n); \hat P_n)'\,\widehat\bWald \;=\; \sum_{\ell=1}^L \lambda_\ell(\bW(\hat P_n); \hat P_n)\,\widehat\Wald_\ell,
  \]
  with $\sum_{\ell=1}^L \lambda_\ell(\bW(\hat P_n); \hat P_n) = 1$ also w.p.a.1, and $\hat\beta^{\mathrm{GMM}}_{\bW} \xrightarrow{p} \beta^{\mathrm{GMM}}_{\bW}(P_0)$.
\end{enumerate}
\end{proposition}

\noindent By the Wald decomposition (Proposition~\ref{prop:wald_decomp}), the GMM estimand admits a sum-to-one decomposition over compliance types:\footnote{Proposition~\ref{prop:weighted_iv} generalizes the Rotemberg decomposition of \citet{GoldsmithPinkhamSorkinSwift2020} from a single estimator to the  fixed-map weighting-matrix GMM class.}
$$\beta^{\mathrm{GMM}}_{\bW}(P) = \sum_{t \in \mathcal{T}_+(P)}\left(\sum_{\ell=1}^L \lambda_\ell(\bW(P); P)\, \alpha_{\ell,t}(P)\right) \LATE_t(P).$$
The composite weights have a two-level structure: under PRD, the inner weights $\alpha_{\ell,t}(P)$ are non-negative (Proposition~\ref{prop:prd}), but the outer weights $\lambda_\ell(\bW(P); P)$ can be negative by Proposition~\ref{prop:weighted_iv}(a): $\lambda_\ell(\bW(P); P) < 0$ iff $\gamma_\ell(P)$ and $[\bW(P)\bgamma(P)]_\ell$ have opposite signs; off-diagonal entries in $\bW(P)$ can generate this sign reversal when their contribution to $[\bW(P)\bgamma(P)]_\ell$ dominates the diagonal contribution. Composite weights can therefore be negative even in PRD designs, with negativity originating in the off-diagonal structure of $\bW(P)$. Section~\ref{sec:egmm_weights} works out the $L=2$ case under 2SLS, where a negative weight arises when strong inter-instrument correlation overwhelms the direct first-stage contribution.

\subsection{2SLS, TSGMM, EGMM, CUE, and the heterogeneity penalty}\label{sec:2sls}\label{sec:egmm_weights}

The \emph{2SLS weighting matrix map} sets $\bW^{2\mathrm{SLS}}(P) \equiv \bSigmaZ(P)^{-1}$. Substituting into~\eqref{eq:weighted_iv}:
\[
\lambda^{2\mathrm{SLS}}_\ell(P) \;=\; \frac{\gamma_\ell(P)\,[\bSigmaZ(P)^{-1}\bgamma(P)]_\ell}{\bgamma(P)'\bSigmaZ(P)^{-1}\bgamma(P)}.
\]
The weight $\lambda^{2\mathrm{SLS}}_\ell$ is the partial first-stage contribution of instrument $\ell$ after partialling out the other instruments; with correlated instruments ($\bSigmaZ$ non-diagonal), this partial contribution can be negative even when every instrument has a positive marginal first stage.

The product $\bSigmaZ(P)^{-1}\bgamma(P)$ represents the population OLS coefficient from regressing $D_i$ on $\bZ_i$. Component $\ell$, $[\bSigmaZ(P)^{-1}\bgamma(P)]_\ell$, is therefore the \emph{partial first-stage contribution} of $Z_\ell$ after partialling out the remaining instruments. The 2SLS weight $\lambda_\ell^{2\mathrm{SLS}}(P) \propto \gamma_\ell(P)\,[\bSigmaZ(P)^{-1}\bgamma(P)]_\ell$ therefore isolates this partial contribution rather than the marginal first stage $\gamma_\ell(P)$.

For $L = 2$, write $\bSigmaZ(P) = \bigl(\begin{smallmatrix}\sigma_1^2 & \sigma_{12} \\ \sigma_{12} & \sigma_2^2\end{smallmatrix}\bigr)$ with $\sigma_\ell^2 = \Var(Z_\ell)$ and $\sigma_{12} = \Cov(Z_1, Z_2)$. The cofactor expansion gives $[\bSigmaZ(P)^{-1}\bgamma(P)]_1 = (\sigma_2^2\,\gamma_1(P) - \sigma_{12}\,\gamma_2(P))/\det(\bSigmaZ(P))$, so
\begin{equation}
\lambda_1^{2\mathrm{SLS}}(P) \;\propto\; \gamma_1(P)\bigl(\gamma_1(P)\,\sigma_2^2 \;-\; \gamma_2(P)\,\sigma_{12}\bigr).
\end{equation}
A strong positive correlation $\sigma_{12} > 0$ contributes a negative penalty $-\gamma_2(P)\,\sigma_{12}$ that can overwhelm the direct first-stage contribution $\gamma_1(P)\,\sigma_2^2$, forcing $\lambda_1^{2\mathrm{SLS}}(P) < 0$ even when both instruments have positive marginal first stages. When $Z_1$ is highly correlated with $Z_2$, its incremental contribution flips sign, and 2SLS attaches a negative weight to $\Wald_1(P)$ in the affine decomposition of Proposition~\ref{prop:weighted_iv}(a).

TSGMM, EGMM, and CUE refine 2SLS by incorporating residual second moments into the weighting matrix. Under heterogeneous treatment effects, this variance adaptation can change the estimand, yielding three distinct estimands depending on how the second-moments matrix is evaluated.

Denote the second-moment matrix by $\bOmega(\beta;P)$$\equiv \E[\bg(\beta;P)\bg(\beta;P)']$, where $\bg(\beta;P)$$\equiv\bigl(\tilde V_i(\beta)(Z_{1i}-\E[Z_{1i}]),\ldots,\tilde V_i(\beta)(Z_{Li}-\E[Z_{Li}])\bigr)'$ and $\tilde V_i(\beta) \equiv (Y_i - \beta D_i) - \E_P[Y_i - \beta D_i]$ is the demeaned residual; demeaning matches the centering induced by the estimated instrument means in the sample moment, and the empirical implementations, which residualize $Y_i$ and $D_i$ on fixed effects before forming moments, impose it automatically. The second-moments matrix depends on $\beta$ through the residuals.

\begin{definition}\label{def:2step}
The \emph{TSGMM estimator} of \citet{Hansen1982} uses a first-step estimator obtained with the identity weight matrix,
which, under Assumption~\ref{ass:relevance}, admits the Wald-share form $\beta^I(P)  \;\equiv\;  \sum_{\ell=1}^L \bigl(\gamma_\ell(P)^2/\|\bgamma(P)\|^2\bigr)\,\Wald_\ell(P).$ The second-step weighting matrix is $\bW^{\TwoS}(P) \equiv \bOmega(\beta^I(P);P)^{-1}$, and the TSGMM estimand is the fixed-weight GMM target at $\bW^{\TwoS}$:
\[
  \beta^{\TwoS}(P) \;\equiv\; \sum_{\ell=1}^L \lambda_\ell\bigl(\bW^{\TwoS}(P);P\bigr)\,\Wald_\ell(P).
\]
\end{definition}

\begin{definition}\label{def:egmm}
Define the population EGMM update map as the variance-weighted average of the $L$ Wald estimands:
\[
  T(\beta;P) \equiv \sum_{\ell=1}^L \lambda_\ell\bigl(\bOmega(\beta;P)^{-1};P\bigr)\,\Wald_\ell(P).
\]
An \emph{EGMM branch} is a locally attracting fixed point where the estimate and its implicit weights align:
\[
  \beta^{\EGMM}(P) = T(\beta^{\EGMM}(P);P).
\]
The associated EGMM weighting matrix map is $\bW^{\EGMM}(P) \equiv  \bOmega( \beta^{\EGMM}(P); P)^{-1}$. When multiple fixed points exist, one is taken as selected throughout, and ``EGMM'' refers exclusively to this selected branch.
\end{definition}

\begin{definition}\label{def:cue}

The \emph{continuously updating estimator} (CUE) of \citet{HansenHeatonYaron1996} jointly evaluates the moment system and weighting matrix at the same candidate:
\[
\begin{aligned}
  \beta^{\CUE}(P) &\;\equiv\; \argmin_{\beta\in\mathcal B} \,  \E[\bg(\beta;P)]'\,\bOmega(\beta;P)^{-1}\, \E[\bg(\beta;P)],\\
  \bW^{\CUE}(P) &\;\equiv\; \bOmega\bigl(\beta^{\CUE}(P);P\bigr)^{-1},
\end{aligned}
\]
where $\mathcal B\subset\mathbb R$ is a fixed compact interval.

\end{definition}

Definitions~\ref{def:2step}--\ref{def:cue} yield three pseudo-true values that are distinct under heterogeneous treatment effects. TSGMM solves $\bgamma(P)'\bW^{\TwoS}(P)\,\E[\bg(\beta;P)]=0$ in closed form at the plug-in weight $\bW^{\TwoS}(P)=\bOmega(\beta^{I}(P);P)^{-1}$. EGMM solves the fixed-point condition $\bgamma(P)'\bOmega(\beta;P)^{-1}\E[\bg(\beta;P)]=0$, with $\bW^{\EGMM}(P)\equiv\bOmega(\beta^{\EGMM}(P);P)^{-1}$. CUE solves the same equation with a variance-score correction,
\begin{equation}\label{eq:cue_foc}
  \bgamma(P)'\bOmega(\beta;P)^{-1}\E[\bg(\beta;P)] \;=\; -\tfrac{1}{2}R(\beta;P),
\end{equation}
where $R(\beta;P)\equiv\E[\bg(\beta;P)]'\bOmega(\beta;P)^{-1}\bOmega_\beta(\beta;P)\bOmega(\beta;P)^{-1}\E[\bg(\beta;P)]$ and $\bOmega_\beta\equiv\partial\bOmega/\partial\beta$ \citep[the ``extra term'' of][]{HansenHeatonYaron1996}. Rearranging~\eqref{eq:cue_foc} with $[\E[\bg(0;P)]]_\ell=\gamma_\ell(P)\Wald_\ell(P)$ and $\bW^{\CUE}(P)\equiv\bOmega(\beta^{\CUE}(P);P)^{-1}$ gives
\begin{equation}\label{eq:cue_decomp}
  \beta^{\CUE}(P) \;=\; \sum_{\ell=1}^L \lambda_\ell\bigl(\bW^{\CUE}(P);P\bigr)\,\Wald_\ell(P) \;+\; \frac{R\bigl(\beta^{\CUE}(P);P\bigr)}{2\,\bgamma(P)'\bW^{\CUE}(P)\bgamma(P)}.
\end{equation}
TSGMM and EGMM are thus exact Wald-share combinations at $\bW^{\TwoS}(P)$ and $\bW^{\EGMM}(P)$, differing only through their anchor, $\beta^{I}(P)$ versus $\beta^{\EGMM}(P)$; the non-Wald remainder in~\eqref{eq:cue_decomp} removes any interpretation of the CUE target as a weighted average of the instrument-specific Wald estimands. All three coincide under correct specification ($\E[\bg(\beta_0;P)]=0$ for some $\beta_0$).

The 2SLS input $\bSigmaZ(P)$ reflects only the instrument covariance, while the TSGMM, EGMM, and CUE inputs evaluate $\bOmega(\beta;P)$ at $\beta^{I}(P)$, $\beta^{\EGMM}(P)$, and $\beta^{\CUE}(P)$, respectively, absorbing the residual variance from fitting a single $\beta$ to $L$ distinct Wald estimands. In the diagonal specialization, increasing an instrument's normalized residual second moment lowers its EGMM weight (Corollary~\ref{cor:egmm_diag}, Appendix~\ref{app:diagonal}); the same reallocation operates for TSGMM and CUE through their $\lambda_\ell$, with CUE additionally absorbing the variance-score remainder of \eqref{eq:cue_decomp}, signed by $\partial_\beta\log[\bOmega(\beta;P)]_{\ell\ell}$. Appendix~\ref{app:general_omega} gives the local channel under general $\bOmega$. This variance-mediated reallocation is the heterogeneity penalty: the residual variance structure shapes the estimand alongside the researcher's substantive target. Since instruments load differently on compliance types, the penalty passes through to the composite type weights $\psi_t^{\mathrm{est}}(P) = \sum_{\ell=1}^L \lambda_\ell(\bW^{\mathrm{est}}(P);P)\,\alpha_{\ell,t}(P)$, $\mathrm{est} \in \{\TwoS, \EGMM, \CUE\}$, which may be negative for types weighted heavily by penalized instruments.\footnote{Inner weights $\alpha_{\ell,t}(P)$ are non-negative under PRD; outer weights $\lambda_\ell$ can be negative under correlated instruments. \citet{MogstadTorgovitskyWalters2021} show 2SLS complier-group weights can turn negative under partial monotonicity; \citet{GoldsmithPinkhamSorkinSwift2020} document negative Rotemberg weights. The heterogeneity penalty is a residual-variance mechanism through which Hansen-efficient GMM can amplify this channel.}

\subsection{The $J$-test and Wald-estimand equality}\label{sec:jtest_wald}

Under maintained instrument validity, the Hansen $J$-test of overidentifying restrictions is a joint test of Wald-estimand equality and is available for each of the three Hansen-efficient estimators of Definitions~\ref{def:2step}--\ref{def:cue}. The asymptotic null distribution is $\chi^2_{L-1}$ for all three; under contiguous local alternatives the limit is noncentral $\chi^2_{L-1}(\delta^2)$ with the same noncentrality parameter $\delta^2$ for all three; under fixed Wald-heterogeneity the three statistics diverge at rate $n$ to three distinct positive constants tied to their estimator-specific pseudo-true values. The diagnostic is conditional: rejection does not by itself separate treatment-effect heterogeneity from violations of exclusion, monotonicity, or the sampling restrictions.

Let $\hat\beta^I\equiv\beta^I(\hat P_n)$, $\hat\beta^{\TwoS}\equiv\beta^{\TwoS}(\hat P_n)$, $\hat\beta^{\EGMM}\equiv\beta^{\EGMM}(\hat P_n)$, and $\hat\beta^{\CUE}\equiv\beta^{\CUE}(\hat P_n)$ denote the sample analogs of the estimands in Definitions~\ref{def:2step}--\ref{def:cue}, and write $\widehat\bOmega(\beta)\equiv\bOmega(\beta;\hat P_n)$. Define the three $J$-statistics
\begin{align*}
  J_n^{\TwoS}  &\equiv n\,\bg_n(\hat\beta^{\TwoS})'\,\widehat\bOmega(\hat\beta^I)^{-1}\,\bg_n(\hat\beta^{\TwoS}),\\
  J_n^{\EGMM}  &\equiv n\,\bg_n(\hat\beta^{\EGMM})'\,\widehat\bOmega(\hat\beta^{\EGMM})^{-1}\,\bg_n(\hat\beta^{\EGMM}),\\
  J_n^{\CUE}   &\equiv n\,\widehat Q(\hat\beta^{\CUE})\;=\;\min_{\beta\in\mathcal B} n\,\widehat Q(\beta), \qquad \widehat Q(\beta)\equiv \bg_n(\beta)'\widehat\bOmega(\beta)^{-1}\bg_n(\beta).
\end{align*}
The TSGMM statistic freezes the weight at the first-step plug-in $\widehat\bOmega(\hat\beta^I)$; the EGMM statistic uses the self-consistent fixed-point weighting matrix $\widehat\bOmega(\hat\beta^{\EGMM})$; the CUE statistic is the minimum of the GMM criterion over $\mathcal B$.

\begin{proposition}[$J$-test and Wald-estimand equality]\label{prop:jtest}
Suppose Assumption~\ref{ass:relevance}, Assumption~\ref{ass:sampling}(a),(d),(e), and Assumption~\ref{ass:gmm-reg} hold. For parts (a) and (b), suppose further that a common value $\beta_0$ satisfies $\E[(Y_i-\beta_0 D_i)(Z_{\ell i}-p^Z_\ell(P_0))]=0$ for every $\ell$. The moment Jacobian $-\bgamma(P_0)$ has rank one and the GMM projection has overidentifying rank $L-1$.
\begin{enumerate}[label=\textnormal{(\alph*)},leftmargin=*]
\item \textnormal{(Null limit.)} Under $H_0:\Wald_1=\Wald_2=\cdots=\Wald_L$,
\[
  J_n^{\TwoS},\ J_n^{\EGMM},\ J_n^{\CUE} \;\Rightarrow\; \chi^2_{L-1}.
\]
\item \textnormal{(Local-alternative limit.)} If contiguous local alternatives $\{P_{n,h}\}\subset\PLATE$ satisfy
\[
\sqrt n\,\E_{P_{n,h}}\bigl[\bg(\beta_0;P_{n,h})\bigr] \to \dot{\bd} \in {\mathbb R}^L,
\qquad
\sqrt n\,\bg_n(\beta_0)\;\Rightarrow_{P_{n,h}}\;N\bigl(\dot{\bd},\,\bOmega(\beta_0;P_0)\bigr),
\]
then
\[
  J_n^{\TwoS},\ J_n^{\EGMM},\ J_n^{\CUE} \;\Rightarrow\; \chi^2_{L-1}(\delta^2),
  \qquad
  \delta^2 \;=\; \dot{\bd}'\,\bM_{\bOmega}\,\dot{\bd},
\]
with the common projector
\[
  \bM_{\bOmega} \;=\; \bOmega^{-1} \;-\; \bOmega^{-1}\bgamma\,\bigl(\bgamma'\bOmega^{-1}\bgamma\bigr)^{-1}\bgamma'\bOmega^{-1},
  \qquad (\bgamma,\bOmega) \;\text{evaluated at } (\beta_0,P_0).
\]
If the Wald estimands move locally, $\Wald_\ell(P_{n,h})=\beta_0+n^{-1/2}\alpha_\ell+o(n^{-1/2})$, then $\dot{\bd}=\diag\{\gamma_\ell(P_0)\}\boldsymbol\alpha$: the noncentrality vanishes exactly when all $\alpha_\ell$ coincide and is positive under local Wald heterogeneity.
\item \textnormal{(Fixed-alternative divergence.)} Under fixed alternatives for which no scalar $\beta$ zeros all $L$ population moments, $n^{-1}J_n^{\mathrm{est}}\xrightarrow{p} c^{\mathrm{est}}_*>0$ for $\mathrm{est}\in\{\TwoS,\EGMM,\CUE\}$, where
\[
  \begin{aligned}
    c^{\TwoS}_*  &= \E[\bg(\beta^{\TwoS};P)]'\,\bOmega(\beta^I;P)^{-1}\E[\bg(\beta^{\TwoS};P)],\\
    c^{\EGMM}_*  &= \E[\bg(\beta^{\EGMM};P)]'\,\bOmega(\beta^{\EGMM};P)^{-1}\E[\bg(\beta^{\EGMM};P)],\\
    c^{\CUE}_*   &= \min_{\beta\in\mathcal B} \E[\bg(\beta;P)]'\,\bOmega(\beta;P)^{-1}\E[\bg(\beta;P)] \;\le\; c^{\EGMM}_*.
  \end{aligned}
\]
Each statistic therefore diverges to $+\infty$ at rate $n$. The constants $c^{\TwoS}_*$ and $c^{\CUE}_*$ (and $c^{\TwoS}_*$ and $c^{\EGMM}_*$) are not in general comparable because $c^{\TwoS}_*$ uses the inner product $\bOmega(\beta^I;P)^{-1}$ while the others use $\bOmega$ evaluated at a different value.
\end{enumerate}
\end{proposition}

\begin{remark}\label{rem:jtest_yardsticks}
Proposition~\ref{prop:jtest} shows that the three $J$-statistics test the same null and share a first-order limit, but differ in their pseudo-true values under fixed heterogeneity. Since the CUE statistic globally minimizes $n\,\widehat Q$ over $\mathcal B$, $c^{\CUE}_*\le c^{\EGMM}_*$ and $J_n^{\CUE}\le J_n^{\EGMM}$ whenever $\widehat\bOmega(\beta)$ is invertible on $\mathcal B$, the CUE minimizer exists, and the selected EGMM fixed point lies in $\mathcal B$: $\hat\beta^{\EGMM}$ satisfies the plug-in normal equation alone, omitting $\widehat R(\hat\beta^{\EGMM})$ from \eqref{eq:cue_foc}. The frozen weight $\widehat\bOmega(\hat\beta^I)^{-1}$ leaves the TSGMM residual-maker Loewner-incomparable to the other two, so no deterministic ordering ties $J_n^{\TwoS}$ to $J_n^{\EGMM}$ or $J_n^{\CUE}$.
\end{remark}

\noindent Under maintained instrument validity, rejection by any of $J_n^{\TwoS}$, $J_n^{\EGMM}$, or $J_n^{\CUE}$ indicates unequal instrument-specific Wald estimands, not necessarily invalid instruments.\footnote{This conditional diagnostic reading aligns with \citet{GoldsmithPinkhamSorkinSwift2020} and \citet{MogstadTorgovitskyWalters2021}, while \citet{AndrewsChen2025} show the $J$-statistic characterizes the range of achievable estimates under local misspecification. The first-order equivalence of parts~(a) and~(b) of Proposition~\ref{prop:jtest} across the three estimators is the standard GMM noncentrality calculation; see \citet{NeweyMcFadden1994} for the parametric submodel projection underlying $\bM_{\bOmega}$.} This reflects treatment-effect heterogeneity interacting with varying instrument weights, causing different GMM weighting matrix maps to target distinct combinations of $\LATE_t(P)$, conditional on the maintained exclusion restriction.

\subsection{Diagonal specialization}\label{sec:diagonal}

\begin{assumption}[Independent instruments]\label{ass:diag_instruments}
$Z_{1i}, \ldots, Z_{Li}$ are mutually independent.
\end{assumption}

\begin{assumption}[Non-overlapping compliance]\label{ass:diag_nonoverlap}
$\mathcal{C}_\ell(P_0) \cap \mathcal{C}_k(P_0) = \emptyset$ for all $\ell \neq k$.
\end{assumption}

\noindent Under Assumptions~\ref{ass:joint_indep}--\ref{ass:relevance} and~\ref{ass:diag_instruments}, $\bSigmaZ = \diag(p^Z_\ell(1-p^Z_\ell))$, the indirect effect in Lemma~\ref{lem:di} vanishes, and all weights $\alpha_{\ell,t}(P)$ are non-negative. Combined with Assumption~\ref{ass:diag_nonoverlap}, each complier group $\mathcal{C}_\ell$ contains a single compliance type, and $\Wald_\ell = \E[Y_i(1)-Y_i(0)\mid i\in\mathcal{C}_\ell]$: each Wald estimand is the single-type LATE of \citet{ImbensAngrist1994}.

Throughout the diagonal specialization, write $s_\ell(P)\equiv p^Z_\ell(P)\{1-p^Z_\ell(P)\}$ and $\sigma^2_{\epsilon,\ell}(\beta;P)\equiv[\bOmega(\beta;P)]_{\ell\ell}/s_\ell(P)$, the normalized residual second moment for instrument $\ell$ at evaluation point $\beta$. For brevity, we write $s_\ell(P)$ and $\sigma^2_{\epsilon,\ell}(\beta;P)$ as $s_\ell$ and $\sigma^2_{\epsilon,\ell}(\beta)$ when no confusion arises.

\begin{corollary}[2SLS under independent instruments]\label{cor:2sls_diag}
Under Assumption~\ref{ass:relevance}, Assumption~\ref{ass:sampling}(a),(d),(e), and Assumption~\ref{ass:diag_instruments}, the 2SLS weights are
\[
    \lambda^{2\mathrm{SLS}}_\ell \;=\; \frac{\pi_\ell^2 \, s_\ell}{\sum_k \pi_k^2 \, s_k}.
\]
Each weight is proportional to the instrument's marginal first-stage variance $\pi_\ell^2 s_\ell$.
\end{corollary}

\begin{corollary}[EGMM under diagonal $\bOmega$]\label{cor:egmm_diag}
Under Assumptions~\ref{ass:joint_indep}, \ref{ass:monotonicity}, \ref{ass:relevance}, \ref{ass:diag_instruments}, and~\ref{ass:diag_nonoverlap}, Assumption~\ref{ass:sampling}, and the regularity conditions of Assumption~\ref{ass:gmm-reg}, $\bOmega$ is diagonal and the EGMM weights are
\[
    \lambda^{\EGMM}_\ell \;=\; \frac{\pi_\ell^2 \,  s_\ell\, / \, \sigma^2_{\epsilon,\ell}(\beta^{\EGMM})}{\sum_k \pi_k^2 \,  s_k \, / \, \sigma^2_{\epsilon,k}(\beta^{\EGMM})}. 
\]
\end{corollary}

\begin{corollary}[TSGMM under diagonal $\bOmega$]\label{cor:2step_diag}
Under the assumptions of Corollary~\ref{cor:egmm_diag}, the first-step Wald-share weights are
\begin{equation}
  \lambda^I_\ell \;=\; \frac{\gamma_\ell^2}{\|\bgamma\|^2} \;=\; \frac{\pi_\ell^2\,s_\ell^2}{\sum_{k=1}^L \pi_k^2\,s_k^2},
\end{equation}
and the second-step Wald-share weights are
\begin{equation}\label{eq:2step_diag_weights}
  \lambda^{\TwoS}_\ell \;=\; \frac{\pi_\ell^2\,s_\ell\,/\,\sigma^2_{\epsilon,\ell}(\beta^I)}{\sum_{k=1}^L \pi_k^2\,s_k\,/\,\sigma^2_{\epsilon,k}(\beta^I)},
\end{equation}
formally identical to $\lambda^{\EGMM}_\ell$ but evaluated at $\beta^I$ rather than $\beta^{\EGMM}$.
\end{corollary}

\begin{corollary}[CUE under diagonal $\bOmega$]\label{cor:cue_diag}
Under the assumptions of Corollary~\ref{cor:egmm_diag}, the CUE Wald-share weights are
\begin{equation}\label{eq:cue_diag_weights}
  \lambda^{\CUE}_\ell \;=\; \frac{\pi_\ell^2\,s_\ell\,/\,\sigma^2_{\epsilon,\ell}(\beta^{\CUE})}{\sum_{k=1}^L \pi_k^2\,s_k\,/\,\sigma^2_{\epsilon,k}(\beta^{\CUE})},
\end{equation}
formally identical to $\lambda^{\EGMM}_\ell$ but evaluated at $\beta^{\CUE}$ rather than $\beta^{\EGMM}$. The variance-score remainder in~\eqref{eq:cue_decomp} admits the closed form
\begin{equation}\label{eq:cue_diag_remainder}
  \frac{R\bigl(\beta^{\CUE};P\bigr)}{2\,\bgamma(P)'\bW^{\CUE}(P)\,\bgamma(P)}
  \;=\;
  \frac{1}{2}\sum_{\ell=1}^L \lambda^{\CUE}_\ell\,\bigl(\Wald_\ell(P) - \beta^{\CUE}\bigr)^2\,\frac{\partial}{\partial\beta}\log \sigma^2_{\epsilon,\ell}(\beta^{\CUE}).
\end{equation}
\end{corollary}

  \noindent The ratio $\lambda^{\EGMM}_\ell/\lambda^{\mathrm{2SLS}}_\ell \propto 1/\sigma^2_{\epsilon,\ell}(\beta^{\EGMM})$ gives the heterogeneity penalty in closed form. Holding the evaluation point and other residual components fixed, $\sigma^2_{\epsilon,\ell}$ increases with the within-complier treatment-effect variance $\sigma^2_{\tau,\ell} \equiv \Var(Y_i(1)-Y_i(0)\mid D_i(\cdot)\in\mathcal{C}_\ell(P))$, so under diagonality $\partial\lambda^{\EGMM}_\ell/\partial\sigma^2_{\tau,\ell}<0$ (Appendix~\ref{app:diagonal}): EGMM downweights instruments whose complier groups exhibit high treatment-effect dispersion, with 2SLS and EGMM weights coinciding iff $\sigma^2_{\epsilon,\ell}(\beta^{\EGMM})$ is constant across $\ell$. The TSGMM and CUE shares $\lambda^{\TwoS}_\ell$ and $\lambda^{\CUE}_\ell$ inherit the same mechanism at $\beta^{I}$ and $\beta^{\CUE}$. TSGMM carries no variance-score remainder, so the TSGMM-vs-EGMM gap stems entirely from evaluating $\sigma^2_{\epsilon,\ell}$ at $\beta^{I}$ versus $\beta^{\EGMM}$; CUE adds the non-Wald remainder \eqref{eq:cue_diag_remainder}, vanishing under homogeneity ($\Wald_\ell = \beta^{\CUE}$ for all $\ell$), each term otherwise signed by $\partial_\beta\log\sigma^2_{\epsilon,\ell}(\beta^{\CUE})$.

\subsection{MTE representation}\label{sec:mte_penalty} 
 
The compliance-type weights $\alpha_{\ell,t}(P)$ of Proposition~\ref{prop:wald_decomp} are discrete objects indexed by positive-support response types. A supplementary latent-index assumption places these discrete weights on a continuous latent-resistance scale.
\begin{assumption}[Latent-index representation]\label{ass:latent_index}
The propensity score $p(z)$ is coordinate-wise nondecreasing on $\supp(\bZ)$ under the instrument partial order. There exists a latent variable $U_i \sim \mathrm{Uniform}(0,1)$ with $(Y_i(0),Y_i(1),U_i)\perp\!\!\!\perp \bZ_i$ such that, for every $z\in\supp(\bZ)$,
\[
  D_i(z)=\ind\{p(z)\geq U_i\}.
\]
\end{assumption} 
\noindent 
\citet{Vytlacil2002} shows that Assumption~\ref{ass:latent_index} is equivalent to the uniform-ordering (index-assignment) strengthening of coordinate-wise monotonicity.\footnote{Under the maintained  Assumptions~\ref{ass:joint_indep}-\ref{ass:realized_exclusion},  Assumption \ref{ass:latent_index}  is equivalent to  
the pairwise uniform monotonicity condition that, for every pair $z,z' \in \supp(\bZ)$, either $D_i(z)\geq D_i(z')$ a.s.\ or $D_i(z)\leq D_i(z')$ a.s. holds.}

Under Assumption~\ref{ass:latent_index}, the space of realized-support compliance patterns reduces to a subset with threshold types on the realized propensity ordering. Let $0 \leq p_1 < p_2 < \cdots < p_K \leq 1$ be the distinct values of $p(z)$ on $\supp(\bZ)$, with boundary convention $p_0 = 0$ and $p_{K+1} = 1$. The threshold types are $t_k(z) \equiv \ind\{p(z) \geq p_{k+1}\}$ for $k = 0, 1, \ldots, K$, with type probabilities $\theta_{t_k} = p_{k+1} - p_k$. Type $t_0$ is the always-taker (active when $p_1>0$), and $t_K$ is the never-taker (active when $p_K<1$). The threshold indexing is used only under this supplementary structure.  

In common-source designs, this is a restriction on how potential treatment decisions respect the underlying scalar ordering, not a consequence of the instrument labels alone. For instance, in a cumulative-threshold design where the instrument support forms a chain, applying Assumption~\ref{ass:monotonicity} along that chain automatically delivers the required uniform ordering within each cell.\footnote{Non-overlapping-compliance designs do not require this global scalar-resistance ordering for the discrete diagonal results; the STAR application uses the non-overlap specialization of Section~\ref{sec:diagonal}. The Patent MTE and PRTE analysis maintains the nested threshold ordering and checks the local-support implications in Section~\ref{sec:patent}.}

  Assumption~\ref{ass:latent_index} maps each active threshold type to a contiguous interval of latent resistance, turning the type-level decomposition into a continuous integral over the marginal treatment effect curve.

\begin{lemma}[Threshold types as latent-resistance intervals]\label{lem:type_interval}
Under Assumption~\ref{ass:latent_index}, each active threshold type $t_k$ induced by the realized propensity ordering with $\theta_{t_k} > 0$ corresponds to a single contiguous interval $\mathcal{R}_{t_k} \subset [0,1]$, determined by the distinct ordered values of $p(z)$ across $z \in \supp(\bZ)$. The intervals $\{\mathcal{R}_{t_k}\}$ are non-overlapping and partition $(0,1]$ up to the choice of half-open endpoints, with always-takers at the bottom and never-takers at the top.
\end{lemma}
 
\begin{remark}[Monotonicity vs.\ latent-index compliance types]
For $L=1$, Assumption~\ref{ass:monotonicity} rules out defiers and leaves a single complier group (Table~\ref{tab:type_dictionaries}, left); Assumption~\ref{ass:latent_index} adds nothing. For $L=2$ with $0<p(0,0)<p(1,0)<p(0,1)<p(1,1)<1$, coordinate-wise monotonicity admits six response patterns (Table~\ref{tab:type_dictionaries}, right). Assumption~\ref{ass:latent_index} enforces strict nestedness on this ordering, retaining the ordered-threshold types but excluding the crossing pattern: anyone whose resistance $U$ is overcome by the weaker encouragement $p(1,0)$ must also comply under the stronger $p(0,1)$. Consequently, there are only five threshold types, with their latent-resistance intervals given by $\mathcal{R}_{t_0}=(0,p(0,0)]$, $\mathcal{R}_{t_1}=(p(0,0),p(1,0)]$, $\mathcal{R}_{t_2}=(p(1,0),p(0,1)]$, $\mathcal{R}_{t_3}=(p(0,1),p(1,1)]$, and $\mathcal{R}_{t_4}=(p(1,1),1]$.
\end{remark}

\begin{table}[t]
\centering
\small
\caption{Compliance-type response patterns. \textbf{Left:} The $L=1$ benchmark where Assumption~\ref{ass:monotonicity} rules out the defier. \textbf{Right:} $L=2$ with $p(0,0)<p(1,0)<p(0,1)<p(1,1)$. Assumption~\ref{ass:monotonicity} (coordinate-wise) admits all six patterns shown, including the crossing row below the rule; Assumption~\ref{ass:latent_index} additionally excludes the crossing pattern, leaving the five ordered-threshold types $t_0,\ldots,t_4$.}
\label{tab:type_dictionaries}
\vspace{0.2cm}

\begin{minipage}[t]{\textwidth}
    \centering
    \begin{tabular}[t]{lcc}
    \toprule
    Type & $Z=0$ & $Z=1$ \\
    \midrule
    Always-taker & 1 & 1 \\
    Complier & 0 & 1 \\
    Never-taker & 0 & 0 \\
    \midrule
    \emph{Defier (ruled out)} & \emph{1} & \emph{0} \\
    \bottomrule
    \end{tabular}
\end{minipage}%
\par\smallskip
\begin{minipage}[t]{\textwidth}
    \centering
    \begin{tabular}[t]{lcccc}
    \toprule
    Type & $(0,0)$ & $(1,0)$ & $(0,1)$ & $(1,1)$ \\
    \midrule
    Always-taker ($t_0$) & 1 & 1 & 1 & 1 \\
    $Z_1\!\cup\!Z_2$-complier ($t_1$) & 0 & 1 & 1 & 1 \\
    $Z_2$-only complier ($t_2$) & 0 & 0 & 1 & 1 \\
    $Z_1\!\cap\!Z_2$-complier ($t_3$) & 0 & 0 & 0 & 1 \\
    Never-taker ($t_4$) & 0 & 0 & 0 & 0 \\
    \midrule
    \emph{Crossing (ruled out by Assumption~\ref{ass:latent_index})} & \emph{0} & \emph{1} & \emph{0} & \emph{1} \\
    \bottomrule
    \end{tabular}
\end{minipage}
\end{table}

Under Assumption~\ref{ass:latent_index}, each Wald estimand admits the \citet{HeckmanVytlacil2005, HeckmanUrzuaVytlacil2006} marginal treatment effect (MTE) representation
\begin{equation}\label{eq:wald_mte}
\begin{aligned}
  \Wald_\ell(P)
  &= \int_0^1 \MTE(u; P) \, \alpha^{\MTE}_\ell(u; P) \, du,\\
  \alpha^{\MTE}_\ell(u; P)
  &=
  \frac{\mathbb{P}(p(\bZ_i) \geq u \mid Z_{\ell i} = 1)
        - \mathbb{P}(p(\bZ_i) \geq u \mid Z_{\ell i} = 0)}
       {\E[p(\bZ_i) \mid Z_{\ell i} = 1]
        - \E[p(\bZ_i) \mid Z_{\ell i} = 0]} ,
\end{aligned}
\end{equation}
where $\MTE(u; P) = \E[Y_i(1) - Y_i(0) \mid U_i = u]$ is the average treatment effect at latent resistance $u$, $\alpha^{\MTE}_\ell(\cdot; P)$ integrates to one, and $\int_{\mathcal{R}_t(P)} \alpha^{\MTE}_\ell(u; P)\,du = \alpha_{\ell,t}(P)$ recovers the discrete compliance-type weight of Proposition~\ref{prop:wald_decomp}. Because the $L$ instruments identify only $L$ linear combinations of treatment effects, the continuous MTE curve is generally not point-identified; Section~\ref{sec:feasibility_assessment} formalizes this identification gap and bounds it under MTE shape restrictions. The composite MTE-weight function $\psi^{\MTE}(u; \bomega; P) = \sum_{\ell=1}^L \omega_\ell\, \alpha^{\MTE}_\ell(u; P)$ integrates to one for any $\bomega$ summing to one; it converts a linear combination of Wald estimands into an integral over the MTE curve.

\begin{proposition}[MTE representation of the GMM estimand]\label{prop:mte_representation}
Under Assumptions~\ref{ass:joint_indep}, \ref{ass:monotonicity}, \ref{ass:realized_exclusion}, \ref{ass:relevance}, \ref{ass:latent_index} (which implies coordinate-wise monotonicity on $\supp(\bZ)$), and~\ref{ass:sampling}(d), for any GMM weighting matrix map $\bW(\cdot)$,
\[
  \beta^{\mathrm{GMM}}_{\bW(P)}(P) \;=\; \int_0^1 \MTE(u; P) \, \psi^{\MTE}\bigl(u; \blambda(\bW(P); P); P\bigr) \, du.
\]
\end{proposition}

\noindent The composite weight function $\psi^{\MTE}\bigl(u; \blambda^{\EGMM}(P); P\bigr)$ can contract on latent-resistance ranges emphasized by instruments with greater residual dispersion; Appendix~\ref{app:general_omega} gives the corresponding local channel under a general $\bOmega$, where responses of the off-diagonal entries of $\bOmega$ can also matter. 

\begin{proposition}[PRD implies non-negative MTE weights]\label{prop:prd_mte}
Under Assumptions~\ref{ass:joint_indep}, \ref{ass:monotonicity}, \ref{ass:realized_exclusion}, \ref{ass:relevance}, \ref{ass:prd}, and~\ref{ass:latent_index}, $\alpha^{\MTE}_\ell(u; P) \geq 0$ for all $\ell$ and $u \in [0,1]$.
\end{proposition}

\noindent With $\alpha^{\MTE}_\ell(u; P) \geq 0$, each Wald estimand is a convex weighted average of marginal treatment effects, and any convex combination $\psi^{\MTE}(u; \bomega; P) = \sum_\ell \omega_\ell\, \alpha^{\MTE}_\ell(u; P)$ with $\bomega \in \Delta^{L-1}$ is non-negative.


\section{Representativeness Targeting}\label{sec:rt}

Under heterogeneity, the implicit weight map $\blambda_{\bW}(\cdot)$ of a GMM weighting matrix map drifts with $P$ through $\bgamma(P)$ and $\bW(P)$, and the GMM estimator is the plug-in of this map (Proposition~\ref{prop:weighted_iv}(b)): choosing $\bW(\cdot)$ selects the weight map indirectly, and the estimand follows whatever compliance-type allocation the residual variance structure dictates. The direct alternative specifies target weights $\bomega(\cdot)$ as a map from the model to the instrument simplex, evaluated at the empirical distribution. The plug-in's implicit weights then equal the researcher's target weights at every realization, and the compliance-type allocation $\{\psi_t\bigl(\bomega; P\bigr)\}_{t\in \mathcal{T}_+(P)}$ is determined by the target, not by the variance structure.
   
	 \subsection{The estimator and its target class}

\begin{definition}[Representativeness Targeting]\label{def:rt}
An RT weight map is a simplex-valued functional $\bomega(\cdot) : \PLATE \to \Delta^{L-1}$ continuous at $P_0$; the collection of admissible weight maps is $\mathcal{W}(P_0)$. The Representativeness Targeting (RT) estimand and estimator for $\bomega(\cdot) \in \mathcal{W}(P_0)$ are, respectively,
\begin{equation*}
\beta^{\mathrm{RT}}_{\bomega}(P)\equiv \bomega(P)'\, \bWald(P) \;=\; \sum_{\ell=1}^{L} \omega_\ell(P)\, \Wald_\ell(P)\ \text{ and }\  \hat\beta^{\mathrm{RT}}_{\bomega} \equiv \bomega(\hat P_n)'\, \widehat\bWald.
\end{equation*} 
\end{definition}

\noindent The RT estimator evaluates the RT weight map at the empirical distribution and averages the sample Wald estimands with the resulting weights. Proposition~\ref{prop:rt_types} characterizes the resulting RT estimand as a convex combination of type-specific LATEs.

\begin{proposition}[Causal validity of RT]\label{prop:rt_types}
Under Assumptions~\ref{ass:joint_indep}, \ref{ass:monotonicity}, \ref{ass:realized_exclusion}, \ref{ass:relevance}, and~\ref{ass:prd}, for any $\bomega(\cdot) \in \mathcal{W}(P_0)$ and any $P$ in the neighborhood where the map is evaluated, the RT estimand  is a convex combination of type-specific treatment effects:
\begin{equation}\label{eq:rt_type_decomp}
  \beta^{\mathrm{RT}}_{\bomega}(P) \;=\; \sum_{t \in \mathcal{T}_+(P)} \psi_t\bigl(\bomega; P\bigr) \cdot \LATE_t(P), \qquad \psi_t(\bomega; P) \;=\; \sum_{\ell=1}^{L} \omega_\ell(P)\, \alpha_{\ell,t}(P),
\end{equation}
with $\psi_t(\bomega; P) \geq 0$ and $\sum_{t \in \mathcal{T}_+(P)} \psi_t(\bomega; P) = 1$.
\end{proposition}

\noindent Unlike 2SLS, TSGMM, EGMM, and CUE, whose composite type weights $\psi_t(\bW; P)$ can be negative under correlated instruments (Section~\ref{sec:egmm_weights}), RT with simplex-valued $\bomega(\cdot)$ implies $\psi_t(\bomega; P) \geq 0$ under PRD: the RT estimand is a convex weighted average of type-specific LATEs.\footnote{All results extend to settings with covariates $X_i$ under conditional versions of joint independence, monotonicity, relevance, and PRD, together with full first-stage saturation \citep{BlandholBonneyMogstadTorgovitsky2022}; Proposition~\ref{prop:covariates} in Appendix~\ref{app:covariates} formalizes the conditional and marginal RT estimands. For a single instrument ($L = 1$), $\omega \equiv 1$ and RT recovers the \citet{ImbensAngrist1994} LATE.}

\subsection{LAM bound and plug-in attainment}\label{sec:lam-attain}

The pathwise derivative of a single Wald estimand $\Wald_\ell(\cdot)$ is represented by the per-instrument linearization below.
\[
\phi_\ell(O_i; P_0) \;=\; \frac{[(Y_i - \E[Y_i]) - \Wald_\ell(P_0)(D_i - \E[D_i])](Z_{\ell i} - \E[Z_{\ell i}])}{\gamma_\ell(P_0)} \;\in\; L^2_0(P_0)
\]
derived via the quotient rule on $\Wald_\ell = \Cov(Y_i,Z_{\ell i})/\gamma_\ell$, the covariance form of $\rho_\ell/\pi_\ell$; the $L \times L$ \emph{Wald covariance matrix}
\[
\bGamma^{\Wald}(P_0) \;=\; \E\!\bigl[\bphi(O_i; P_0)\,\bphi(O_i; P_0)'\bigr] \;\in\; \mathbb{R}^{L \times L}
\]
collects their second moments.

\begin{definition}[LATE tangent space]\label{def:late-tangent}
Let $\mathcal{G}\subset L^2_0(P_0)$ denote the linear span of all scores arising from one-parameter LATE-compatible submodels $\{P_t\}\subset\PLATE$ through $P_0$ that perturb exactly one primitive:
\begin{enumerate}[label=\textnormal{(\roman*)},leftmargin=*,itemsep=2pt,topsep=4pt]
\item the instrument mass $\mu_z$, by a bounded signed measure of zero total mass;
\item a positive-support type probability $\theta_r$, by a bounded direction in the simplex on $\mathcal{T}_+(P_0)$;
\item a type-conditional outcome density $f_{r,d}$ for $r\in\mathcal{T}_+(P_0)$, by a bounded $L^2_0(F_{r,d})$ direction.
\end{enumerate}
\end{definition}

The submodels are \emph{quadratic mean differentiable} (QMD) \citep{vanderVaart1998} under Assumption~\ref{ass:late-regularity}. The \emph{LATE tangent space} is its $L^2(P_0)$-closure $\overline{\mathcal{G}}$.

\begin{proposition}[RT attainment]\label{prop:eif-decomp}
Assume Assumptions~\ref{ass:relevance} and~\ref{ass:sampling}, and suppose $\bomega(\cdot)$ is pathwise differentiable at $P_0$ with Riesz representers $\Riesz_{\omega,\ell}\in L^2_0(P_0)$. Define
\begin{equation}\label{eq:eif-general}
\eif_{\omega(\cdot)}(O_i;P_0)
=
\sum_{\ell=1}^{L}\omega_\ell(P_0)\phi_\ell(O_i;P_0)
+\sum_{\ell=1}^{L}\Wald_\ell(P_0)\Riesz_{\omega,\ell}(O_i;P_0),
\end{equation}
and write $V_\omega(P_0)\equiv\Var_{P_0}[\eif_{\omega(\cdot)}]$. By Lemma~\ref{lem:rt-pathwise} in the Supplementary Appendix, $\beta^{\mathrm{RT}}_{\bomega}(\cdot)$ is then pathwise differentiable at $P_0$ with observed-law Riesz representer $\eif_{\omega(\cdot)}$.
\begin{enumerate}[label=\emph{(\roman*)},leftmargin=*]
\item \emph{LAM lower bound.} Under the further Assumptions~\ref{ass:joint_indep}, \ref{ass:monotonicity}, \ref{ass:realized_exclusion}, and~\ref{ass:late-regularity}, and the tangent-compatibility condition $\eif_{\omega(\cdot)}\in\TLATEbar$, $\eif_{\omega(\cdot)}$ is the LATE-model efficient influence function and, along LATE-compatible contiguous alternatives $\{P_{n,h}\}$ with score $h$, every estimator sequence $\{\hat\beta_n\}$ satisfies
\begin{equation}\label{eq:lam-bound}
\sup_{\substack{I \subset \TLATE\\ \dim I < \infty}}
\lim_{C\to\infty}
\liminf_{n\to\infty}
\sup_{\substack{h\in I\\ \norm{h}_{L^2(P_0)}\le C}}
n\,\E\!\left[\{\hat\beta_n-\beta^{\mathrm{RT}}_{\bomega}(P_{n,h})\}^2\right]
\ge V_\omega(P_0).
\end{equation}

\item \emph{Plug-in attainment.} If $\bomega(\cdot)$ is also Hadamard differentiable along LATE-compatible charts, $\eif_{\omega(\cdot)}$ has a finite $(2+\delta)$ moment, the von Mises expansion
\[
\omega_\ell(\hat P_n)-\omega_\ell(P_0)
=n^{-1}\sum_{i=1}^n\Riesz_{\omega,\ell}(O_i;P_0)+o_p(n^{-1/2})
\qquad (\ell=1,\ldots,L)
\]
holds, and Assumption~\ref{ass:local-risk-ui} is in force, then the plug-in is asymptotically normal along every LATE-compatible contiguous sequence,
\begin{equation}\label{eq:plug-in-asylinear}
\sqrt n\,\{\hat\beta^{\mathrm{RT}}_{\bomega}-\beta^{\mathrm{RT}}_{\bomega}(P_{n,h})\}
\xrightarrow{d}_{P_{n,h}} N\bigl(0,V_\omega(P_0)\bigr),
\end{equation}
and attains the lower bound~\eqref{eq:lam-bound} with equality.
\end{enumerate}
\end{proposition}

\noindent The proposition adds the LATE efficiency claim.
 The tangent-compatibility condition asks that the candidate representer in~\eqref{eq:eif-general} belong to the closed observed-data tangent space generated by LATE-compatible structural perturbations of the instrument law, the positive-support type-probability simplex, and the type-conditional outcome distributions. Appendix~\ref{app:tc-primitive-supp} gives primitive sufficient conditions for this closure membership: when every instrument cell has positive probability and the three canonical cell scores (cell mass, conditional treatment probability, and conditional outcome mean) are themselves expressible as observed scores of LATE-compatible structural perturbations, every Wald representer $\phi_\ell$ lies in $\TLATEbar$. If in addition $\bomega(\cdot)$ is a continuously differentiable function of these cell primitives, then so is the weight-drift representer $\Riesz_{\omega,\ell}$, and the full expression~\eqref{eq:eif-general} then satisfies the tangent-closure condition (Theorem~\ref{thm:tc-primitive}).\footnote{For projected RT, the same conclusion holds at a locally stable active face when the projection primitives $(\alpha_{\ell,t},\bpsi^\star,\bkappa)$ are themselves so expressible (Corollary~\ref{cor:tc-projected}); for PRTE with endogenous policy weights, direct expressibility of the interval-weight derivative representers as observed structural scores is a primitive sufficient verification of the tangent-closure requirement that Assumption~\ref{ass:projection-tangent} imposes and that the PRTE analysis maintains alongside Assumption~\ref{ass:prte-weight-regularity}.}
 
 The decomposition in~\eqref{eq:eif-general} distinguishes this framework from standard textbook GMM. The first term is the fixed-weight Wald Riesz component with weights frozen at $\bomega(P_0)$, while the second term captures how the weight map $\bomega(\cdot)$ varies with $P$. For fixed weights ($\bomega(\cdot)\equiv\bomega_0$), $\Riesz_\omega\equiv\mathbf{0}$ and the representer reduces to $\sum_\ell \omega_{0,\ell}\phi_\ell$, with variance $\bomega_0'\bGamma^{\Wald}(P_0)\bomega_0$. For data-adaptive maps, $\Riesz_\omega$ is derived from the von Mises expansion of $\bomega(\hat P_n)$. The expanding radius ($\lim_{C\to\infty}$) in~\eqref{eq:lam-bound} is essential: fixed-radius local balls can be beaten by shrinkage. Both parts are proved in Appendix~\ref{app:proof_thm_pathwise}. The fixed-weight specialization matches the semiparametric efficiency bound for LATE at $L=1$ of \citet{Frolich2007}; for $L\geq2$ with a $P$-dependent RT weight map, the baseline quadratic form is augmented by the weight-drift cross-term $2\bomega(P_0)'\bUpsilon\bWald(P_0)$ and the pure weight-drift variance $\bWald(P_0)'\bSigma_{\Riesz}\bWald(P_0)$.

\begin{proposition}[RT variance estimation]\label{prop:var-estimation}
Under the plug-in-attainment conditions of Proposition~\ref{prop:eif-decomp} and Assumption~\ref{ass:riesz-consistency}, let $\hat\bGamma^{\Wald}$, $\hat\bUpsilon$, and $\hat\bSigma_{\Riesz}$ be sample-analog covariance matrices formed from the Wald influence functions $\hat\phi_\ell$ and the estimated weight-map Riesz representers $\hat\Riesz_{\omega,\ell}$ (explicit formulas in Appendix~\ref{app:proof_var_estimation}). The RT variance estimator
\begin{equation}\label{eq:Vhat-sample}
\hat V_n \;=\; \bomega(\hat P_n)'\, \hat\bGamma^{\Wald}\, \bomega(\hat P_n) \;+\; 2\, \bomega(\hat P_n)'\, \hat\bUpsilon\, \widehat\bWald \;+\; \widehat\bWald'\, \hat\bSigma_{\Riesz}\, \widehat\bWald
\end{equation}
is consistent: $\hat V_n \xrightarrow{p} V_\omega(P_0)$. The Wald interval $\hat\beta^{\mathrm{RT}}_{\bomega} \pm z_{1-\alpha/2}\sqrt{\hat V_n / n}$ has pointwise asymptotic coverage $1 - \alpha$ along every LATE-compatible contiguous submodel.
\end{proposition}

\noindent The first term in~\eqref{eq:Vhat-sample} is the baseline Wald variance; the second and third are the weight-estimation correction. For fixed weights ($\bomega(\cdot) \equiv \bomega_0$), the correction vanishes and $\hat V_n = \bomega_0'\hat\bGamma^{\Wald}\bomega_0$. The naive GMM sandwich omits the correction. Under the supplementary uniformity conditions of Appendix~\ref{app:uniform-coverage} (Assumptions~\ref{ass:late-regularity-uniform} and~\ref{ass:omega-uniform-hadamard}), this coverage holds uniformly over a restricted local class (Proposition~\ref{prop:uniform-coverage}).

\subsection{RT as a GMM weighting matrix map}\label{sec:witness}
 
Given RT's LAM attainment (Proposition~\ref{prop:eif-decomp}), can RT be implemented inside the GMM class? A GMM estimand reads its weighting matrix entirely through the $L$-vector $\bW\bgamma$; therefore, the answer is a single alignment condition.

\begin{proposition}[Alignment between RT and GMM]\label{cor:witness}
Any weighting matrix map in the family $\bW(P) = c(P)\,\bW^\dagger(P) + \bK(P)$ implements the target weights, satisfying $\blambda_{\bW}(P) = \bomega(P)$ for every $P$. Here, $c(P) > 0$ is a scaling factor, and $\bK(P)$ is a symmetric matrix satisfying $\bK(P)\bgamma(P) = \mathbf{0}$ such that the sum $c(P)\,\bW^\dagger(P) + \bK(P) \succ 0$. When $\omega_\ell(P)>0$ for every $\ell$, the canonical diagonal matrix $\bW^\dagger(P) = \diag\!\bigl(\omega_\ell(P)/\gamma_\ell(P)^2\bigr)$ is the unique diagonal member of this family up to scale.
\end{proposition}

\noindent This result reverses the standard GMM logic: rather than letting a residual-variance criterion dictate the estimand, the researcher fixes the target $\bomega$ first, and the canonical diagonal matrix map implements it when the target is interior. The extra term $\bK$ acts as an invisible gauge: because the implicit weights in~\eqref{eq:weighted_iv} depend on $\bW$ only through the product $\bW\bgamma$, any $\bK$ with $\bK\bgamma = \mathbf{0}$ leaves the final target completely unchanged. For the general implementing condition, see Appendix~\ref{app:witness-class-supp}.

\subsection{Regularity for the RT target}\label{sec:tangent-matching}

Proposition~\ref{cor:witness} exhibits the weighting matrix maps whose implicit weight map reproduces the target exactly. A general map need not: $\blambda_{\bW}(\cdot)$ may drift with $P$ in a direction that the target map $\bomega(\cdot)$ does not follow, and the resulting first-order discrepancy biases inference along LATE-compatible local alternatives.

\begin{proposition}[Regularity characterization]\label{thm:tangent-matching}
Let $\bW(\cdot)$ be a weighting matrix map with $\bW(\hat P_n) \xrightarrow{p} \bW(P_0)$, and let $\bomega(\cdot) \in \mathcal{W}(P_0)$ satisfy the plug-in-attainment conditions of Proposition~\ref{prop:eif-decomp}. Suppose the implicit weight map admits the von Mises expansion
\begin{equation}\label{eq:W-asy-linear-main}
\lambda_{\bW,\ell}(\hat P_n)-\lambda_{\bW,\ell}(P_0)
=n^{-1}\sum_{i=1}^n\psi^{\blambda}_{\bW,\ell}(O_i)+o_p(n^{-1/2}),\qquad \ell=1,\ldots,L,
\end{equation}
with $\psi^{\blambda}_{\bW,\ell} \in L^2_0(P_0)$. Define the ambient linearization $\phi^{\bW} \equiv \sum_{\ell=1}^L \lambda_{\bW,\ell}(P_0)\,\phi_\ell + \sum_{\ell=1}^L \Wald_\ell(P_0)\,\psi^{\blambda}_{\bW,\ell}$ and the influence-function gap $\Delta_{\bW,\omega} \equiv \phi^{\bW} - \eif_{\omega(\cdot)}$. \footnote{The linearization $\phi^{\bW}$ is the influence function~\eqref{eq:eif-general} evaluated at the implicit weight map, with $\psi^{\blambda}_{\bW,\ell}$ in place of the Riesz representers.} Then $\hat\beta^{\mathrm{GMM}}_{\bW}$ is regular for $\beta^{\mathrm{RT}}_{\bomega}(\cdot)$ at $P_0$ in the LATE model if and only if
\begin{enumerate}[label=\emph{(\roman*)},leftmargin=*,itemsep=2pt]
\item \emph{estimand matching:} $\beta^{\mathrm{GMM}}_{\bW}(P_0) = \beta^{\mathrm{RT}}_{\bomega}(P_0)$,
\item \emph{first-order tangent matching:} $\Pi_{\overline{\mathcal{G}}}\,\Delta_{\bW,\omega}=0$.
\end{enumerate}
\end{proposition}

\noindent Proposition~\ref{thm:tangent-matching} makes estimand matching necessary but not sufficient for regularity; the implicit weight map's drift must also align with the target's along the LATE tangent directions. Importantly, the class does not include CUE whose estimand departs from the Wald-share form~\eqref{eq:weighted_iv} by the variance-score remainder of~\eqref{eq:cue_decomp}. 

\section{Weighting maps by projection}\label{sec:specifications}

The previous section estimates a target once the instrument weights are specified. In many applications the researcher starts one level deeper, with a compliance-type target, while the available instruments reach only the allocations generated by their Wald decompositions. Projection is the device that turns this mismatch into a well-defined weighted-Wald target.

\subsection{The projection and the plug-in}\label{sec:projection_kkt}

For each instrument $\ell$ and type $t\in\mathcal{T}_+(P_0)$, the compliance-type weight $\alpha_{\ell,t}(P)$ of Proposition~\ref{prop:wald_decomp} sends an instrument-weight vector $\bomega\in\Delta^{L-1}$ to the type allocation
\[
  \psi_t(\bomega;P)=\sum_{\ell=1}^L\omega_\ell\,\alpha_{\ell,t}(P),\qquad t\in\mathcal{T}_+(P_0).
\]
The feasible set is what the instruments can reach,
\begin{equation}
  \Feas(P) \;\equiv\; \bigl\{(\psi_t(\bomega;P))_{t\in\mathcal{T}_+(P_0)} : \bomega \in \Delta^{L-1}\bigr\},
\end{equation}
the geometric object behind the projection. Indexing over the positive-support types is without loss: every Wald-loaded type has positive mass and thus lies in $\mathcal{T}_+(P_0)$, and by~\eqref{eq:wald_type} the loadings close, $\sum_{t\in\mathcal{T}_+(P)}\alpha_{\ell,t}(P)=1$ for each $\ell$. Under PRD the $\alpha_{\ell,t}(P)$ are nonnegative (Proposition~\ref{prop:prd}), and $\Feas(P)$ lies in the type simplex.

On $\mathcal{T}_+(P_0)$, the researcher's specification pairs a normalized target primitive $\bpsi^\star(P)$ with a \emph{metric weight} $\bkappa(P) = (\kappa_t(P))_{t\in\mathcal{T}_+(P_0)}$, $\kappa_t(P) > 0$, governing how compliance-type-specific residuals are penalized in the projection criterion. The weighted Euclidean projection is any solution of
\begin{equation}\label{eq:projection-omega-dagger}
  \bomega^\dagger(\bpsi^\star; P) \;\in\; \argmin_{\bomega \in \Delta^{L-1}}\; \sum_{t\in\mathcal{T}_+(P_0)} \kappa_t(P)\, \Bigl(\textstyle\sum_{\ell=1}^L \omega_\ell\, \alpha_{\ell,t}(P) - \psi^\star_t(P)\Bigr)^2.
\end{equation}

The projection targets the weighted-Wald scalar, equivalently a weighted average of type effects,
\[
\beta^\dagger(\bpsi^\star; P)= \bomega^\dagger(\bpsi^\star; P)'\bWald(P) = \sum_{t\in\mathcal{T}_+(P_0)} \psi_t(\bomega^\dagger(\bpsi^\star; P); P)\, \LATE_t(P),
\]
the second equality by the Wald decomposition~\eqref{eq:wald_type}.

The projection~\eqref{eq:projection-omega-dagger} is a finite-dimensional least-squares projection over the instrument simplex, a convex quadratic program whose Karush--Kuhn--Tucker (KKT) conditions characterize the solution through the active set $I(\bpsi^\star; P) = \{\ell : \omega^\dagger_\ell(\bpsi^\star; P) > 0\}$ of instruments carrying positive weight.

\begin{lemma}[Active-set stability]\label{lem:active-set-stability}
Under Assumptions~\ref{ass:strict-complementarity} and~\ref{ass:projection-dictionary}, the active set $I(\bpsi^\star; P) = I_0$ is constant on a neighborhood $\mathcal{U}_0$ of $P_0$ and the active-face first-order system is nonsingular throughout.\footnote{The positivity condition requires $\theta_t(P_0) > 0$ only for types loaded by $\bpsi^\star(P_0)$ or entering the metric through a denominator; it is weaker than full threshold-support; the conditions are those stated in Assumptions~\ref{ass:strict-complementarity} and~\ref{ass:projection-dictionary}, with the proof in Appendix~\ref{app:proof_lem_active}.}
\end{lemma}

The plug-in $\hat\beta^\dagger=\bomega^\dagger(\bpsi^\star;\hat P_n)'\widehat\bWald$ replaces $P$ by $\hat P_n$. By Lemma~\ref{lem:active-set-stability}, the weight map $P\mapsto\bomega^\dagger$ is differentiable at $P_0$; since $\hat P_n \in \mathcal{U}_0$ with probability tending to one, the sample projection has active set $I_0$ and the delta method applies on that event.

\begin{proposition}[Projected RT regularity and LAM attainment]\label{prop:lam-projected}
Assume Assumptions~\ref{ass:joint_indep}, \ref{ass:monotonicity}, \ref{ass:realized_exclusion}, \ref{ass:relevance}, \ref{ass:sampling}, \ref{ass:late-regularity}, \ref{ass:strict-complementarity}, \ref{ass:projection-dictionary}, and~\ref{ass:projection-tangent}, with Assumption~\ref{ass:prte-weight-regularity} for PRTE. Suppose the influence function~\eqref{eq:eif-general}, evaluated at $\bomega^\dagger(\bpsi^\star;\cdot)$ with inactive coordinates set to zero, has a finite $(2+\delta)$ moment, and let $I_0=I(\bpsi^\star;P_0)$. Then the plug-in $\hat\beta^\dagger$ is asymptotically linear for $\beta^\dagger(\bpsi^\star;P_0)$ with the influence function in~\eqref{eq:eif-general}, and along every LATE-compatible contiguous sequence $\{P_{n,h}\}$,
\[
\sqrt n\,\{\hat\beta^\dagger - \beta^\dagger(\bpsi^\star; P_{n,h})\} \xrightarrow{d}_{P_{n,h}} N\bigl(0, V^\dagger(P_0)\bigr),
\]
with $V^\dagger(P_0)\equiv\Var_{P_0}[\eif_{\omega^\dagger(\cdot)}]$ the variance of the projected-RT influence function at $P_0$. Under Assumption~\ref{ass:local-risk-ui}, $\hat\beta^\dagger$ attains the LAM lower bound for $\beta^\dagger(\bpsi^\star;\cdot)$:
\begin{equation}\label{eq:lam-bound-projected}
\sup_{\substack{I \subset \mathcal{G}\\ \dim I < \infty}}
\lim_{C\to\infty}
\liminf_{n\to\infty}
\sup_{\substack{h\in I\\ \norm{h}_{L^2(P_0)}\le C}}
n\,\E\!\left[\{\hat\beta^\dagger-\beta^\dagger(\bpsi^\star; P_{n,h})\}^2\right]
\;=\; V^\dagger(P_0).
\end{equation}
\end{proposition}
\noindent Proposition~\ref{prop:var-estimation} applies with $\bomega(\cdot) = \bomega^\dagger(\bpsi^\star;\cdot)$, whose weight-map Riesz representer comes from differentiating the active-face KKT map (Proposition~\ref{prop:lam-projected}): the resulting $\hat V_n$ estimates $V^\dagger(P_0)$ and carries the generated-weight correction induced by the projection map. The sandwich that freezes $\bomega^\dagger(\hat P_n)$ omits the correction and is inconsistent for $V^\dagger(P_0)$.
\subsection{Identification-gap bounds for the projected target}\label{sec:feasibility_assessment}

Under Assumption~\ref{ass:projection-dictionary}, the projection delivers the weighted-Wald scalar $\beta^\dagger(\bpsi^\star;P)=\bomega^\dagger(\bpsi^\star;P)'\bWald(P)$, a surrogate for the estimand the researcher specified. That estimand is generally not point-identified since the $L$ Wald estimands pin down only $L$ combinations of the $|\mathcal{T}_+(P_0)|$ type effects, leaving the individual $\LATE_t$ free in the directions the instruments cannot reach. The identification gap is the discrepancy between the two,
\begin{equation}\label{eq:gap-def}
{\Delta(\bpsi^\star; P)\;\equiv\;\sum_{t\in\mathcal{T}_+(P_0)}\psi^\star_t(P)\LATE_t(P)-\beta^\dagger(\bpsi^\star;P)\;=\;\sum_{t\in\mathcal{T}_+(P_0)} r_t(\bpsi^\star; P)\LATE_t(P),}
\end{equation}
with type-$t$ residual weight $r_t(\bpsi^\star;P)=\psi^\star_t(P)-\psi_t(\bomega^\dagger(\bpsi^\star;P);P)$: known weights on unidentified effects, summing to zero for normalized $\bpsi^\star(P)$. The gap vanishes when $\bpsi^\star(P)\in\Feas(P)$, where the projection matches the specification and $\beta^\dagger$ is itself the desired estimand. For counterfactual policies, these definitions extend to a finite effect dictionary $\mathcal J(P_0)$ that may augment $\mathcal T_+(P_0)$ with compliance margin outside the Wald span, with $\alpha_{\ell,t}(P)=0$ whenever no observed instrument moves margin $t$. 
\begin{proposition}[Envelope-LP]\label{prop:id-gap-distfree}
Under Assumptions~\ref{ass:joint_indep}, \ref{ass:monotonicity}, \ref{ass:realized_exclusion}, \ref{ass:relevance}, and~\ref{ass:projection-dictionary}, suppose that $\sum_{t\in\mathcal J(P_0)}\psi^\star_t(P)=1$. For coordinatewise bounds $-\infty\leq\underline M_t\leq\overline M_t\leq\infty$, let $B^{\mathrm{LP}}(\underline M,\overline M)$ contain the collections $(b_t)_{t\in\mathcal J(P_0)}$ satisfying $\sum_{t\in\mathcal J(P_0)}\alpha_{\ell,t}(P)b_t=\Wald_\ell(P)$ for every $\ell$ and $\underline M_t\leq b_t\leq\overline M_t$ for every $t\in\mathcal J(P_0)$.
If $\LATE_t(P)\in[\underline M_t,\overline M_t]$ for every $t\in\mathcal J(P_0)$, then this set is nonempty and
\begin{equation}\label{eq:gap-bound-lp}
\Delta(\bpsi^\star;P)\in
\left[
\inf_{(b_t)\in B^{\mathrm{LP}}(\underline M,\overline M)}
\sum_{t\in\mathcal J(P_0)}r_tb_t,\;
\sup_{(b_t)\in B^{\mathrm{LP}}(\underline M,\overline M)}
\sum_{t\in\mathcal J(P_0)}r_tb_t
\right].
\end{equation}
For any polyhedron $\mathcal R(P)$ containing $(\LATE_t(P))_{t\in\mathcal J(P_0)}$,~\eqref{eq:gap-bound-lp} holds with $B^{\mathrm{LP}}(\underline M,\overline M)$ replaced by $B^{\mathrm{LP}}(\underline M,\overline M)\cap\mathcal R(P)$. In particular, if $\delta_0(P)\equiv\sum_{t\notin\mathcal U}r_t\LATE_t(P)$ is separately identified under additional structure for some $\mathcal U\subseteq\mathcal J(P_0)$, one may take $\mathcal R(P)=\{(b_t)_{t\in\mathcal J(P_0)}:\sum_{t\notin\mathcal U}r_tb_t=\delta_0(P)\}$.
\end{proposition}
\noindent The constants $\underline M_t$ and $\overline M_t$ are sensitivity parameters, not identified features of the model. When $\LATE_t(P)$ have known support, that support can supply such natural bounds. For unbounded outcomes, the researcher could instead choose a plausible heuristic calibration. A plausible choice is $\max_\ell c|\Wald_\ell(P)|$: because the Wald estimands are weighted averages of the LATEs in the same units, their largest absolute magnitude provides a within-support scale benchmark, and the scale factor $c$ determines how far the effects are allowed to depart from that benchmark. 

\begin{proposition}[Monotone MTE and adjacent LATEs]\label{prop:id-gap-mte}
Under Assumptions~\ref{ass:latent_index} and~\ref{ass:local-threshold-support}, suppose that $\{t_0,\ldots,t_K\}$ correspond to ordered adjacent propensity intervals. If $\MTE(\cdot;P)$ is weakly increasing, then $
\LATE_{t_0}(P)\leq\LATE_{t_1}(P)\leq\cdots\leq\LATE_{t_K}(P).
$
Conversely, any collection $b_{t_0}\leq\cdots\leq b_{t_K}$ is rationalizable as the adjacent-interval averages of a weakly increasing MTE. The inequalities are reversed under weakly decreasing MTE.
\end{proposition}
\noindent The Envelope-LP bound in Proposition~\ref{prop:id-gap-distfree} allows the researcher to combine type-specific effect bounds with other maintained linear restrictions; the LP form draws on the partial-identification tradition of \citet{Manski1990, Manski1997} and is operationalized for IV models with shape restrictions on the underlying response functions by \citet{MogstadSantosTorgovitsky2018}. Proposition~\ref{prop:id-gap-mte} shows the complementarity between MTE monotonicity and the LATE ordering. MTE monotonicity finds direct empirical support in the declining returns to college of \citet{CarneiroHeckmanVytlacil2011}; its plausibility is design-specific.\footnote{First-stage monotonicity (Assumption~\ref{ass:monotonicity}) is a distinct, separately maintained restriction governing the existence of defiers; the recent diagnostics of \citet{Sigstad2026} and \citet{FrandsenLefgrenLeslie2023} in judge and examiner designs, and the specification test of \citet{BaiHuangMoonSantosShaikhVytlacil2024} more generally, target that first-stage restriction and should not be conflated with shape restrictions on the MTE curve.} Since adjacent-LATE order restrictions in Proposition~\ref{prop:id-gap-mte} are linear, they can be imposed jointly with the envelope restrictions in Proposition~\ref{prop:id-gap-distfree}. LP endpoint maps can be directionally rather than differentiable when their binding constraints change. \citet{FangSantos2019} develop general inference that is applicable to such maps, and \citet{HongLi2018} provide a numerical directional-delta implementation.

\subsubsection{Compliance-share and equal-weight targets (CS-ATE and EW-ATE)}\label{sec:cs_ate}\label{sec:ew_ate}

Let $\mathcal{C}_{\mathrm{ever}}(P) \equiv \bigcup_{\ell=1}^L \mathcal{C}_\ell(P)$ denote the set of ever-compliers, comprising all positive-support types that respond to at least one instrument. CS-ATE weights each compliance type by its population share among this group; EW-ATE weights all types equally regardless of prevalence. Both targets use the unweighted $\ell^2$ metric $\kappa_t \equiv 1$ and differ only in how they weight the ever-complier types. The CS-ATE and EW-ATE target weights are respectively given by
\begin{equation*}
  \psi^{\star,\mathrm{CS}}_{t}(P) \;=\; \frac{\theta_{t}(P)\, \ind\{t \in \mathcal{C}_{\mathrm{ever}}(P)\}}{\sum_{t' \in \mathcal{C}_{\mathrm{ever}}(P)} \theta_{t'}(P)}\quad\text{and}\quad
  \psi^{\star,\mathrm{EW}}_{t}(P) \;=\; \frac{\ind\{t \in \mathcal{C}_{\mathrm{ever}}(P)\}}{|\mathcal{C}_{\mathrm{ever}}(P)|},
\end{equation*}
for each type $t\in\mathcal{T}_+(P_0)$.

Under the diagonal non-overlap specialization with a one-to-one instrument-type map, both targets lie in $\Feas(P)$, and the projection is exact. The EW-ATE plug-in recovers the familiar equal-weight Wald average as the projection of a compliance-type primitive. In cumulative-threshold and other overlapping-compliance designs, feasibility is target-specific and need not hold; when it fails, the projection reverts to the KKT problem of Section~\ref{sec:projection_kkt}. Because $\bpsi^{\star,\mathrm{CS}}(P)$ depends on $P$ through $\theta_t(P)$, the CS-ATE LAM variance carries a weight-drift contribution from estimated type shares. On the other hand, for EW-ATE with $|\mathcal{C}_{\mathrm{ever}}|$ constant (the STAR application, where the exact projection gives implemented weights $\omega^\dagger_\ell=1/|\mathcal{C}_{\mathrm{ever}}|$ independent of $P$), the weights are fixed and no drift correction arises.

 \subsubsection{The policy-relevant treatment effect (PRTE) target}\label{sec:mte}
Consider a one-directional policy that shifts the joint instrument distribution from $F^Z_0$ to $F^Z_1$, raising each compliance type's treatment probability by
\[
\Delta_k^{\mathrm{PRTE}}(P)=\mathbb{P}_{F^Z_1}(D_i(\bZ)=1\mid D_i(\cdot)=t_k)-\mathbb{P}_{F^Z_0}(D_i(\bZ)=1\mid D_i(\cdot)=t_k)\ge0.
\]
Type $t_k$'s PRTE weight is its share of the total treatment shift,
\begin{equation}\label{eq:prte_integral}
\psi^{\star,\mathrm{PRTE}}_{t_k}(F^Z_0 \to F^Z_1; P) = \frac{\theta_{t_k}(P)\,\Delta_k^{\mathrm{PRTE}}(P)}{\E_{F^Z_1}[D_i] - \E_{F^Z_0}[D_i]} = \int_{p_k(P)}^{p_{k+1}(P)} \psi^{\mathrm{PRTE}}(u)\,du,
\end{equation}
equivalently the mass on type $t_k$'s threshold interval of $\psi^{\mathrm{PRTE}}$, the density of the treatment shift across the latent-resistance margin. The policy alters the assignment distribution alone: potential outcomes, the latent resistance, and the compliance types are invariant to the shift, satisfying the policy-invariance condition of \citet{HeckmanVytlacil2005}. Averaging the marginal treatment effect by this density gives the policy-relevant treatment effect $\int_0^1 \MTE(u; P)\, \psi^{\mathrm{PRTE}}(u)\,du$ \citep{HeckmanVytlacil2005, CarneiroHeckmanVytlacil2011}.

The PRTE metric is the inverse type mass, $\kappa^{\mathrm{PRTE}}_{t_k}(P)=\theta_{t_k}(P)^{-1}$. Unlike the unit metric of CS- and EW-ATE, it is not a modeling choice: it is the metric under which the type-level projection~\eqref{eq:projection-omega-dagger} agrees with the $L^2[0,1]$ projection of the policy weight $\psi^{\mathrm{PRTE}}$ onto the instrument weight functions $\{\alpha^{\MTE}_\ell\}_{\ell=1}^L$ (Lemma~\ref{lem:prte-equivalence}), with minimizer unique under Assumption~\ref{ass:prte-rank}.

From the $L$ Wald estimands, the exact PRTE is identified by a linear combination only when the target MTE weight function $\psi^{\mathrm{PRTE}}(\cdot;P)$ lies in the linear span of $\{\alpha^{\MTE}_\ell(\cdot; P)\}_{\ell=1}^L$.\footnote{\citet{MogstadSantosTorgovitsky2018} bound the PRTE via linear-programming restrictions on marginal treatment response functions; they accommodate continuous instruments through the propensity score, whereas the present analysis focuses on binary instruments.} 

\section{Applications}\label{sec:application}
\subsection{Class size and student achievement}\label{sec:star}

Tennessee's STAR experiment \citep{Krueger1999} randomly assigned kindergarteners to small, regular, or regular-with-aide classes within schools. Following the standard comparison of small versus regular classes, the aide arm is excluded. The kindergarten wave contributes 78 schools to the analysis sample, and two additional schools first appear in the retained sample in grade 1, giving 80 schools across K--3.\footnote{Schools with fewer than 10 students or fewer than 3 students per arm are dropped, as are students observed in more than one school across the panel and the few school--grade cells in which only one class type is observed; every student then belongs to one school and every cell carries a within-school comparison. The analysis sample matches the excluded observations on demographics (Table~\ref{tab:star_attrition}). Results are robust across subjects (Table~\ref{tab:star_robustness} in Appendix~\ref{app:star_robustness}).} We stack all four experimental grades (K--3) of the single STAR cohort into one sample. The outcome is the Stanford Achievement Test math score, standardized within grade; all effects are reported in grade-standardized (SD) units. Sample restrictions and robustness checks are in Appendix~\ref{app:star_robustness} (Tables~\ref{tab:star_attrition} and~\ref{tab:star_robustness}).

Each school conducts an independent within-school randomization, and every student belongs to exactly one school.\footnote{The classroom and the student are the two design-relevant groupings. The second-moment matrix $\hat\bOmega$, the Wald influence-function covariance, and all standard errors for the estimators in this section are therefore two-way cluster-robust at the classroom and student level, computed with the estimator of \citet{CameronGelbachMiller2011}; this is the clustering structure applied to the STAR panel by \citet{Mueller2013}.} The analysis isolates this design by comparing the small-class and regular-class arms school by school, using the attended class type as the treatment comparison (the pre-switch restriction in Table~\ref{tab:star_robustness} leaves the estimates nearly unchanged). 

These school-specific comparisons serve as the primitive causal objects for STAR. By partialling out school$\times$grade fixed effects, each school is reduced to a conditional single-instrument randomization problem. FWL aggregation then maps these conditional Wald estimands into a pooled coefficient, preserving their causal interpretation \citep{BlandholBonneyMogstadTorgovitsky2022}.\footnote{The analysis does not treat the school dummy variables as globally manipulable instruments. It applies the binary-instrument theory within school and then uses the FWL aggregation described in Appendix~\ref{app:fwl_bridge}.} This school$\times$grade conditioning ensures that $\bSigmaZ$, the two-way clustered $\bOmega$, and the Wald influence-function covariance (Lemma~\ref{lem:rt-diagonal-wald-cov}) are all perfectly diagonal by school. This verifies the conditional analogs of Assumptions~\ref{ass:diag_instruments} and~\ref{ass:diag_nonoverlap}, allowing us to apply the diagonal specialization of Section~\ref{sec:diagonal} school by school. The required per-school support ($p^Z_\ell \in (0,1)$), monotonicity, and non-overlap conditions all hold by construction. 

Under the diagonal specialization, each instrument $\ell$ indexes its own unique compliance type $t_\ell$. The CS-ATE compliance-type primitive is therefore $\psi^{\star,\mathrm{CS}}_{t_\ell} = \theta_\ell$, and the projection~\eqref{eq:projection-omega-dagger} is exact: $\omega^{\dagger,\mathrm{CS}}_\ell = \psi^{\star,\mathrm{CS}}_{t_\ell} = \theta_\ell$.\footnote{Because attended class type is used as the assignment comparison, the first-stage factor in the within-school analog of Corollary~\ref{cor:2sls_diag} equals one; hence the 2SLS implicit weight on school $\ell$ is $\lambda^{2\mathrm{SLS}}_\ell \propto \theta_\ell \cdot p^Z_\ell(1-p^Z_\ell)$, including an additional factor $p^Z_\ell(1-p^Z_\ell)$ comparing to CS-ATE.}

School-specific Wald estimates span a wide range (Figure~\ref{fig:m1_late_landscape}). Consequently, the $J$-statistic decisively rejects the equality of school-level Wald estimands (Table~\ref{tab:star_estimators}, $p < 0.001$). Rather than a test of instrument validity, this rejection serves as direct evidence of treatment effect heterogeneity. Table~\ref{tab:star_estimators} also reports the RT plug-in at CS-ATE and EW-ATE alongside 2SLS, TSGMM, EGMM, and CUE. Notably, the TSGMM, EGMM, and CUE estimates all fall below 2SLS.\footnote{The gap is not the canonical first-stage-estimation channel for many-instruments bias in this design: the demeaned treatment lies in the column space of $\bZ_i$, the first-stage projection is exact, and LIML, JIVE, and 2SLS coincide (Appendix~\ref{app:star_mib}). This addresses the \citet{BoundJaegerBaker1995} bias channel for the constructed within-school first stage.}

In the diagonal specialization, the residual-variance component of the EGMM weight is inversely proportional to $\sigma^2_{\epsilon,\ell}$ (Figure~\ref{fig:m3_het_penalty_b}), though the full weight also contains enrollment and propensity-strength factors. To verify that this penalty is distinct from standard weak-instrument concerns, Figure~\ref{fig:distortion_map} separates the heterogeneity penalty from the instrument-strength adjustment. The two dimensions are nearly orthogonal, confirming that the penalty is not just a relabeling of ``weak-instrument'' schools.

\begin{table}[t]
\centering
\caption{Effect of Small Class Size on Math Scores}
\label{tab:star_estimators}
\small
\setlength{\tabcolsep}{6pt}
\begin{tabular}{lcccccc}
\toprule
 & \multicolumn{4}{c}{GMM estimators} & \multicolumn{2}{c}{RT plug-in} \\
\cmidrule(lr){2-5}\cmidrule(lr){6-7}
 & 2SLS & TSGMM & EGMM & CUE & CS-ATE & EW-ATE \\
\midrule
Estimate & 0.224 & 0.205 & 0.172 & 0.176 & 0.224 & 0.185 \\
 & (0.027) & (0.028) & (0.031) & (0.045) & (0.031) & (0.029) \\
\midrule
$J$-statistic & & 135.49 & 136.24 & 136.23 & & \\
$p$-value & & $< 0.001$ & $< 0.001$ & $< 0.001$ & & \\
$N$ & \multicolumn{6}{c}{15,056} \\
Schools ($L$) & \multicolumn{6}{c}{80} \\
Classrooms; students & \multicolumn{6}{c}{903; 7,598} \\
\bottomrule
\end{tabular}
\begin{minipage}{0.9\textwidth}
\footnotesize\textit{Notes:} Small class (13--17 students) versus regular (22--25); outcome is the Stanford Achievement Test math score, standardized within grade (grade-SD units), for the single STAR cohort. Within-school randomization, school$\times$grade fixed effects residualized via FWL. Standard errors in parentheses are two-way cluster-robust at the classroom and student level \citep{CameronGelbachMiller2011}, the structure applied to STAR data by \citet{Mueller2013}. 2SLS, TSGMM, EGMM, and CUE are the GMM estimators of Section~\ref{sec:egmm_weights}, with the two-way cluster-robust second-moment matrix $\hat\bOmega$; their standard errors are misspecification-robust \citep{HallInoue2003, HansenLee2021}, consistent for the sampling variance around each estimator's probability limit under the Wald-estimand heterogeneity the $J$-statistic detects. CS-ATE and EW-ATE are the RT plug-ins of Section~\ref{sec:cs_ate}, weighting schools by enrollment share and equally; the CS-ATE standard error uses the full influence-function variance $\hat V_n$ (Proposition~\ref{prop:var-estimation}). The 2SLS and CS-ATE estimates are close but not identical (a difference of $0.0008$), reflecting the additional factor $p^Z_\ell(1-p^Z_\ell)$ in the 2SLS weights, which is close to its maximum for the large majority of schools. Inference uses standard-normal critical values.
\end{minipage}
\end{table}

\begin{figure}[htbp]
\centering
\includegraphics[width=0.65\textwidth]{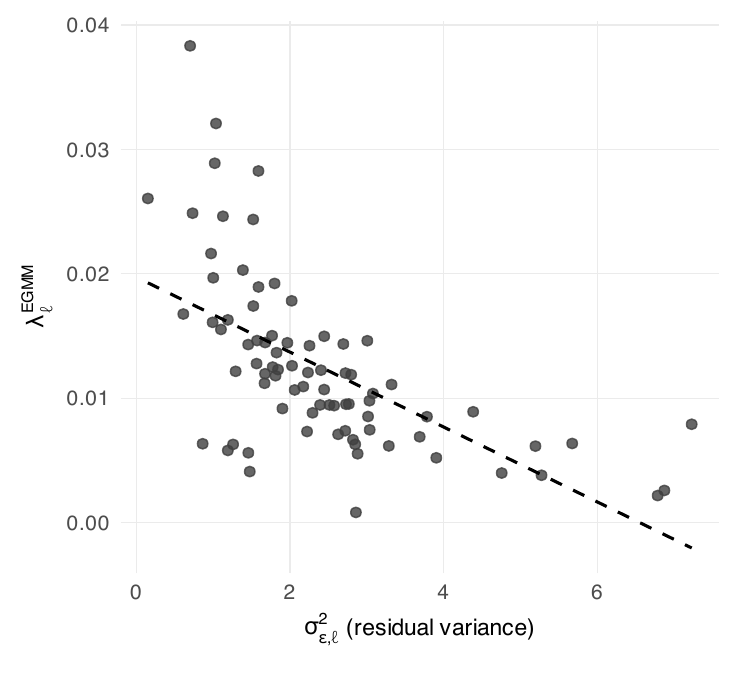}
\caption{The heterogeneity penalty in the diagonal specialization: estimated EGMM weight $\lambda^{\EGMM}_\ell$ against residual variance $\sigma^2_{\epsilon,\ell}$ for the 80 STAR schools (grade-SD units); the dashed line is the OLS fit. The EGMM weight carries an inverse residual-variance component alongside enrollment and propensity-strength factors (Corollary~\ref{cor:egmm_diag} and Appendix~\ref{app:diagonal}).}
\label{fig:m3_het_penalty_b}
\end{figure}

\begin{figure}[htbp]
\centering
\includegraphics[width=0.60\textwidth]{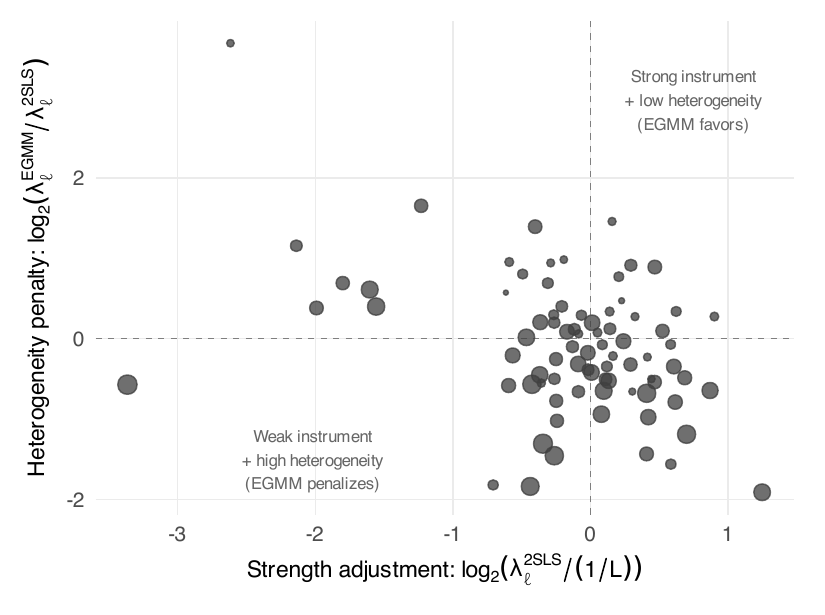}
\caption{Weight distortion map for 80 STAR schools. Horizontal axis: strength adjustment $\log_2(\lambda^{2\mathrm{SLS}}_\ell / (1/L))$; vertical axis: heterogeneity penalty $\log_2(\lambda^{\EGMM}_\ell / \lambda^{2\mathrm{SLS}}_\ell)$. Point size: $|\LATE_\ell|$.}
\label{fig:distortion_map}
\end{figure}

To evaluate a hypothetical expansion of the program, suppose the state funds a fixed number of new small-class placements and assigns a share $\omega_\ell$ of them to school $\ell$. Because assignment within schools remains random, a new placement raises a pupil's score by the school's specific effect, $\LATE_\ell$. The average gain per placement is therefore $\sum_\ell \omega_\ell \LATE_\ell$. This target is a PRTE, and it is point-identified because the experimental data fix the school effects while the policy moves only the allocation weights. 

We consider a \emph{compensatory} expansion that sends more placements to the lowest-baseline schools. Formally, letting $S_i \in \{1, \dots, L\}$ denote the school of student $i$, we assign shares $\omega_\ell \propto \exp(-\mu_{0\ell}/c)$, where $\mu_{0\ell} \equiv \E[Y_i(0) \mid S_i = \ell]$ is the regular-class mean. By varying the multiplier $1/c$, we smoothly shift this maximum-entropy allocation from equal targeting ($\omega_\ell = 1/L$ as $c \to \infty$) to full concentration on the most disadvantaged schools ($c \to 0$). Tilting the allocation toward these schools naturally raises the average gain. Figure~\ref{fig:star_prte_byc} traces this increasing trend as the share of placements reaching the disadvantaged half of schools grows from $50\%$ (no targeting) to $95\%$.\footnote{Small classes help the lowest-baseline schools most, matching the pattern \citet{Krueger1999} documents for free-lunch-eligible and Black students.}

Because both the school effects $\LATE_\ell$ and the target weights $\omega_\ell$ are estimated, a valid confidence interval must account for the statistical uncertainty in both (Proposition~\ref{prop:var-estimation}). To prevent winner's-curse bias, the point estimate is computed via repeated cross-fitting: each school's weight is derived from one random half-sample of students and its Wald from the other, averaging across splits to preserve efficiency.\footnote{A low baseline $\hat\mu_{0\ell}$ by chance draws extra weight and inflates the Wald; estimating them on independent half-samples removes this student-level alignment.} A naive standard error follows a traditional textbook GMM approach, holding the target weights fixed and capturing only the uncertainty in the school effects. This is exact only under no targeting, where weights are fixed at $1/L$.\footnote{The no-targeting point is the cross-fit EW-ATE estimate.} As targeting sharpens, the weights rely heavily on the estimated baselines $\hat\mu_{0\ell}$. The naive standard error omits this variation, causing the interval to narrow artificially (Figure~\ref{fig:star_prte_byc}). Our calibrated simulation (Appendix~\ref{app:mc_star}) confirms this represents genuine under-coverage (Figure~\ref{fig:star_prte_coverage}): while the full influence function remains near nominal, naive coverage falls sharply as the tilt increases.

\begin{figure}[htbp]
\centering
\includegraphics[width=0.6\textwidth]{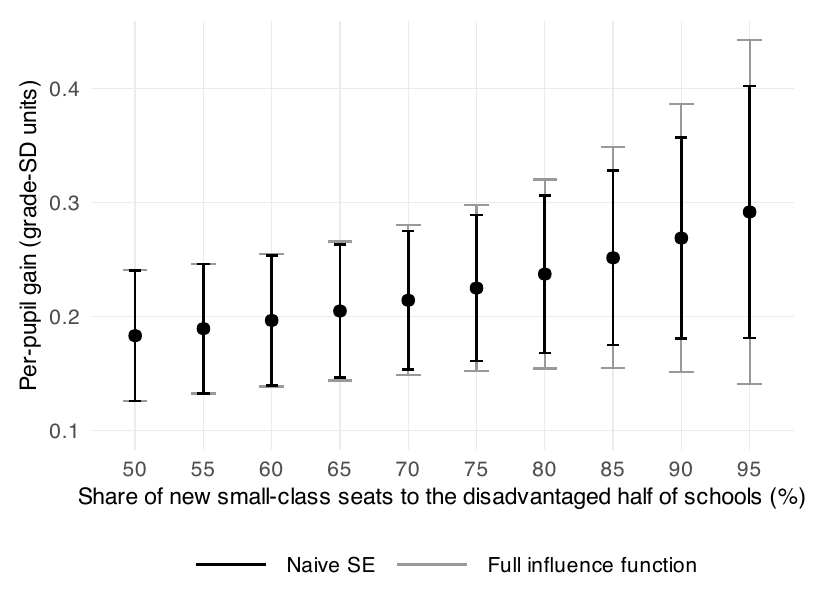}
\caption{Compensatory targeting against the share of new small-class seats reaching the disadvantaged (below-median-baseline) half of schools (the 40 lowest-baseline of 80), from a $50\%$ share (equal allocation across schools, no targeting) to $95\%$. Each point is the cross-fit compensatory PRTE, $\sum_\ell \omega_\ell(\hat\mu_{0\ell})\,\LATE_\ell$, with two $95\%$ intervals: the full influence function (grey, the reported interval) and the weights-fixed naive standard error (black).}
\label{fig:star_prte_byc}
\end{figure}

\begin{figure}[htbp]
\centering
\includegraphics[width=0.6\textwidth]{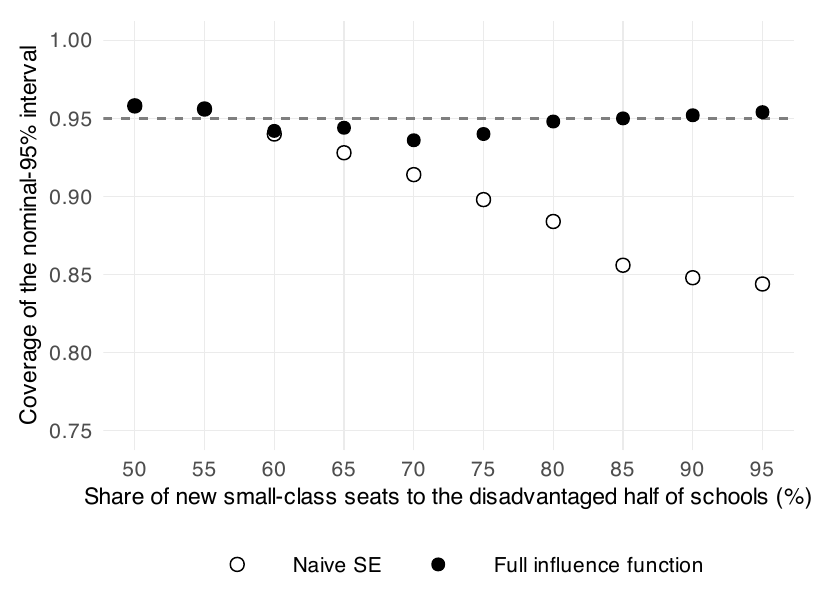}
\caption{Coverage of the nominal-$95\%$ interval for the compensatory PRTE across the share reaching the disadvantaged half of schools, from a $50\%$ share (equal allocation, no targeting) to $95\%$, on the calibrated STAR design ($B = 500$ replications). Open circles: the naive weights-fixed standard error; filled circles: the full influence function; the dashed line marks $0.95$. Both series use the cross-fit estimate.}
\label{fig:star_prte_coverage}
\end{figure}

\subsection{Patent examiner leniency and innovation}\label{sec:patent}

We apply our framework to a patent examiner leniency design, drawn from the replication data of \citet{FarreMensaHegdeLjungqvist2020}. The analysis sample consists of first-time patent applications at the United States Patent and Trademark Office from 2001 to 2009 (Table~\ref{tab:patent_groups}). The treatment $D_i$ is patent approval. The primary outcome $Y_i$ is the total citations received by all subsequent patent applications filed by the same startup, counted over the five years following each application's public disclosure. We also evaluate follow-on applications as an alternative outcome for the policy counterfactual. Examiner leniency is measured by the leave-one-out approval rate, which we use to assign applications into $Q = 7$ leniency quantile groups. This yields $L = 6$ raw cumulative instruments $Z_{\ell i} = \ind\{G_i \geq \ell + 1\}$, where $G_i$ denotes the leniency group.

To account for the assignment level, the moment implementation residualizes these instruments on art unit $\times$ year fixed effects via Frisch--Waugh--Lovell (FWL). These residualized variables implement a saturated conditional aggregation for the threshold Wald estimates.\footnote{Appendix~\ref{app:fwl_bridge} gives the FWL bridge and the resulting gamma-weighted conditional Wald interpretation.} Because the Wald estimates and the compliance-type projection are built from the exact same within-cell residualization, the projection weights and the Wald estimands they combine describe the identical set of compliers. 

Because the thresholds are nondecreasing functions of a common scalar leniency source, the cumulative instruments are positively correlated and their compliance groups overlap. Positive regression dependence (PRD) therefore holds by construction, both within each cell and unconditionally. Accordingly, all analyses use the general projection weight formulas with the full (non-diagonal) second moment matrix $\hat\bOmega$, clustered at the examiner level.\footnote{The main efficiency theory is stated under iid sampling. The patent application reports examiner-clustered analogs of the same moment and influence-function calculations; these clustered standard errors are not labeled as iid LAM standard errors. Many-instrument concerns for leniency designs are less acute here because $L = 6$ cumulative threshold instruments compress the many-examiner first stage into a small threshold design with art-unit--year controls.} The design is highly relevant: the overall first-stage $F$-statistic is 255.3, and approval rates rise monotonically across the seven groups (Table~\ref{tab:patent_groups}). This monotonic rise gives every cumulative threshold, including the upper thresholds receiving negative weights, its own strong first-stage variation.
 Finally, the per-instrument separate-threshold first stages yield nonnegative coefficients across pre-treatment subgroups (Appendix~\ref{app:patent_monotonicity_diag}), validating the monotonicity implication (Appendix~\ref{app:cumulative_monotone}).

\begin{figure}[htbp]
  \centering
  \includegraphics[width=0.65\textwidth]{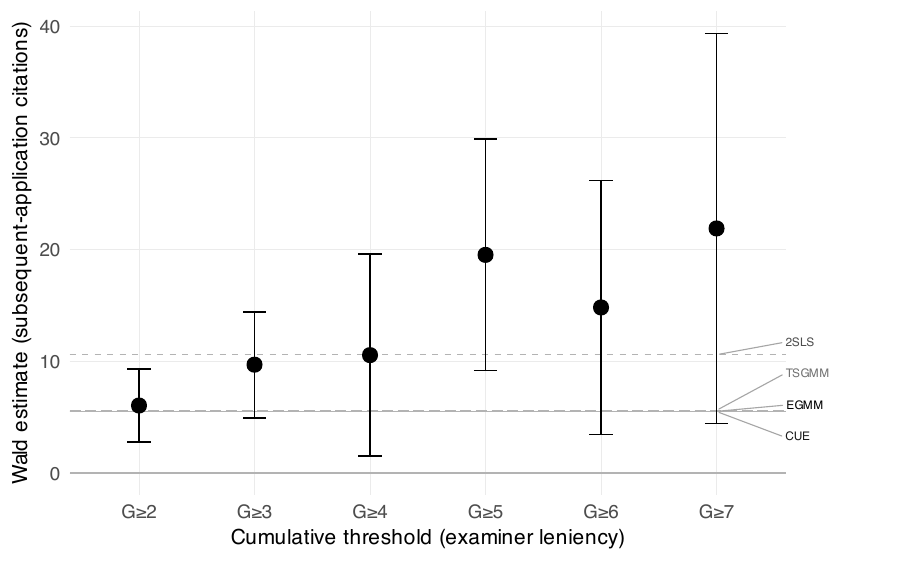}
  \caption{Threshold-specific Wald estimates (five-year citations to subsequent applications) with 95\% confidence intervals. Light horizontal reference lines mark the 2SLS, TSGMM, EGMM, and CUE estimands, labeled on the right margin. EGMM, TSGMM, and CUE lie below every individual Wald estimate; 2SLS lies within the Wald range.}
  \label{fig:patent_wald_estimates}
\end{figure}

The threshold-specific Wald estimates in Figure~\ref{fig:patent_wald_estimates} are generally larger at more lenient thresholds, with one visible nonmonotone step. The $J$-statistic decisively rejects Wald-estimand equality (Table~\ref{tab:patent_estimators}), with robustness checks across choices of $Q$ (Appendix~\ref{app:patent_robust}) supporting this pattern. Standard variance-adaptive weighting, however, profoundly distorts these results. The TSGMM, EGMM, and CUE Wald-share weights (Figure~\ref{fig:patent_weights}) concentrate heavily on the lowest threshold and assign severe negative coefficients to the upper thresholds ($G \geq 5$ and $G \geq 6$). This severe negative weighting drags the estimates below every individual Wald estimate and hence the underlying causal interpretability is not guaranteed. Because the TSGMM and EGMM estimands fall below every threshold-specific Wald estimate, and every RT target is a convex weighting of the Wald estimands, they do not match any RT target: the estimand-matching condition of Proposition~\ref{thm:tangent-matching} fails at every admissible target.

\begin{table}[t]
\centering
\caption{Effect of patent approval on citations to subsequent applications: 2SLS, TSGMM, EGMM, and CUE}
\label{tab:patent_estimators}
\begin{tabular}{lcccc}
\toprule
 & 2SLS & TSGMM & EGMM & CUE \\
\midrule
Estimate & 10.58 & 5.60 & 5.51 & 5.50 \\
 & (2.51) & (1.41) & (1.41) & (1.42) \\
\midrule
$J$-statistic & & 16.29 & 16.36 & 16.36 \\
$p$-value & & 0.0061 & 0.0059 & 0.0059 \\
$N$ & \multicolumn{4}{c}{34,434} \\
Clusters (examiners) & \multicolumn{4}{c}{5,915} \\
\bottomrule
\end{tabular}
\begin{minipage}{0.98\textwidth}
\footnotesize\textit{Notes:} Effect of patent approval on total citations to subsequent applications, counted over five years after each application's public disclosure. $L = 6$ cumulative instruments $Z_\ell = \ind\{G_i \geq \ell+1\}$ from seven examiner leniency groups; art unit $\times$ year fixed effects partialled out via FWL, standard errors clustered at the examiner level. Standard errors are misspecification-robust \citep{HallInoue2003, HansenLee2021}, consistent for the sampling variance around each estimator's probability limit under the Wald-estimand heterogeneity the $J$-statistic detects. 2SLS, TSGMM, EGMM, and CUE are the GMM estimators of Section~\ref{sec:egmm_weights}; the variance-adaptive columns place negative weights on the upper thresholds $G \geq 5$ and $G \geq 6$ (Figure~\ref{fig:patent_weights}), pulling TSGMM, EGMM, and CUE below every threshold-specific Wald estimate. The $J$-statistic is the cluster-robust Wald-estimand-equality diagnostic of Proposition~\ref{prop:jtest}.
\end{minipage}
\end{table}

\begin{figure}[htbp]
  \centering
  \includegraphics[width=0.72\textwidth]{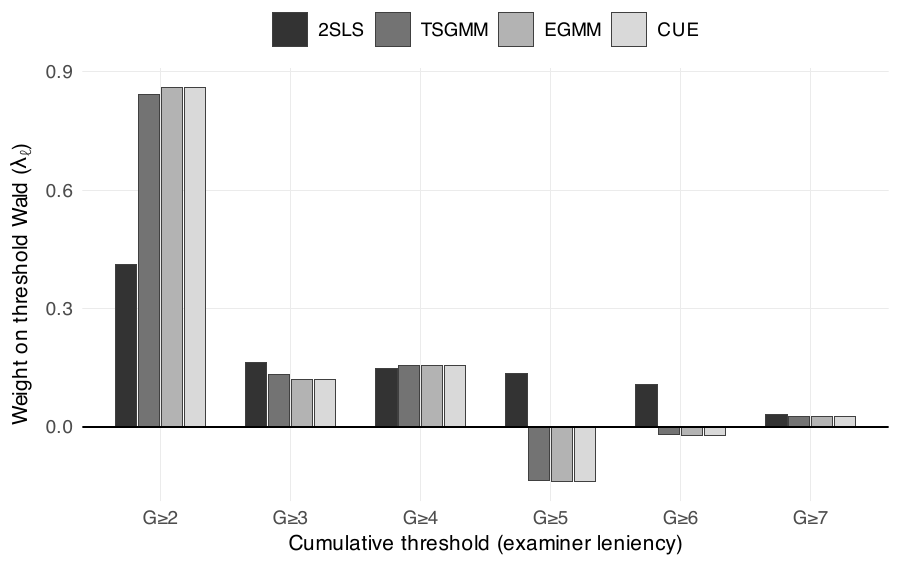}
  \caption{Implicit weights $\lambda_\ell$ on the threshold Wald estimands for 2SLS, TSGMM, EGMM, and CUE (bars by estimator; see legend). The three variance-adaptive maps concentrate weight on $G \geq 2$ (TSGMM $0.84$; EGMM and CUE $0.86$) and place negative weights on $G \geq 5$ and $G \geq 6$.}
  \label{fig:patent_weights}
\end{figure}

\begin{figure}[htbp]
  \centering
  \includegraphics[width=0.98\textwidth]{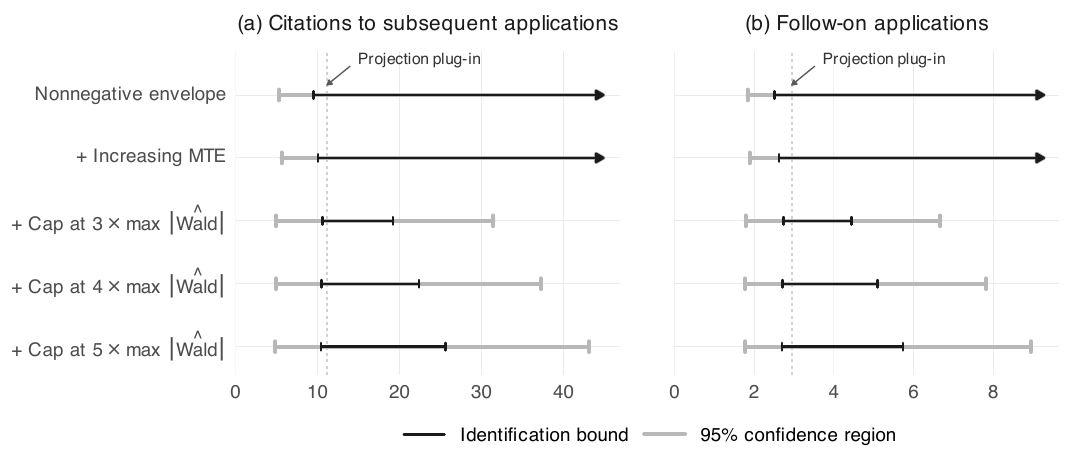}
  \caption{Projection surrogate and PRTE bounds. Black intervals are identification bounds. Gray intervals are 95\% confidence regions by multiplier bootstrap with 1,000 iterations where Appendix~\ref{app:patent_prte_bounds} describes inference. The first row imposes nonnegativity on the out-of-support $G=7\to8$ approval effects ($\underline M_t=0$), the second additionally orders all adjacent approval effects within each art-unit-by-year cell characterized by increasing MTE, and the final three additionally cap the $G=7\to8$ effects at $\overline M_t=\max_\ell c|\widehat{\Wald}_{\ell}|$, $c\in\{3,4,5\}$ (Propositions~\ref{prop:id-gap-distfree} and~\ref{prop:id-gap-mte}). Arrows indicate an unbounded upper endpoint.}
  \label{fig:patent_policy_bounds}
\end{figure}

We evaluate a staircase policy that moves every application one step toward a more lenient examiner group. For applications in groups $G=1,\ldots,6$, the counterfactual approval probability is the estimated approval probability in the next group. Group $G=7$ has no observed successor, and we form a counterfactual group $G=8$ whose approval probability equals that of group $G=7$ plus the average gain across the observed transitions, from $G=1\to2$ through $G=6\to7$.\footnote{Since not all art-unit-by-year cells contain applications at every leniency group, we maintain a common adjacent gain across cells within each technology group.}

To correctly read the plug-in as the surrogate for the policy effect, one additional outcome condition must hold. This is required because the policy counterfactual and its construction from the Wald estimates may average the effects of patent approval across art-unit-by-year cells using different weights: the former weights each cell by the complier mass it moves, whereas the latter weights each art-unit-by-year cell by its first-stage strength. On the within-support margin, we maintain that their average approval effects are equal, and report a supporting diagnostic in Appendix~\ref{app:patent_prte_conditions}. The condition does not restrict the out-of-support margin, to which every observed Wald estimand assigns zero weight. Because the policy assigns positive weight to the out-of-support margin, its target therefore lies outside the Wald span. Under this condition, the contribution to the PRTE of $G=1\to2$ through $G=6\to7$ is identified by the Wald reconstruction, leaving only the $G=7\to8$ contribution unidentified and calling for bounds.\footnote{Under this condition, in the notation of Proposition~\ref{prop:id-gap-distfree}, if $\mathcal U$ comprise only the $G=7\to8$ effects across art-unit-by-year cells, $\delta_0(P)$ is identified, and its defining equality is included in $\mathcal R(P)$ when copmuting the bounds.}

Figure~\ref{fig:patent_policy_bounds} shows the partially identified PRTE. Imposing nonnegativity on the out-of-support $G=7\to8$ approval effect preserves a positive lower bound of PRTE for both outcomes. Additionally, imposing the within-cell ordering of adjacent approval effects characterized by an increasing MTE slightly tightens the lower bound. To obtain a finite upper bound heuristically, we cap the $G=7\to8$ approval effect at a multiple of the largest observed Wald estimate, and the resulting bounds are informative.

\section{Conclusion}

When is GMM actually LATE? Any parameter-free weighting-matrix map delivers a Wald-share combination of the instrument-specific Wald estimands, with shares summing to one; positive regression dependence in turn renders every Wald estimand a convex combination of compliance-type LATEs, so nonnegative Wald shares make the estimand a nonnegatively weighted average of those LATEs. Variance-adaptive weighting offers no such guarantee: the heterogeneity penalty downweights high-dispersion instruments, off-diagonal weighting can drive Wald shares negative, and the estimand can exit the Wald range, where Proposition~\ref{thm:tangent-matching} shows the weighting is regular for no admissible target. CUE never enters the class; its criterion adds a variance-score remainder. Representativeness Targeting estimates each Wald estimand on its own moments and averages the estimates with declared nonnegative weights, attaining the local asymptotic minimax bound for the declared target.

The two applications isolate the two distortion channels. In STAR the Wald
shares stay positive and the penalty operates alone, pulling the EGMM estimate
of the small-class effect substantially below 2SLS; the $J$-test rejects
equality of the school-specific Wald estimands. The RT plug-ins report their
declared targets alongside, with the compliance-share ATE reproducing 2SLS and
the equal-weight ATE lying between the 2SLS and EGMM estimates. In the patent
leniency design the sign condition itself fails: TSGMM and EGMM place negative
Wald shares on two of the six thresholds, and the three variance-adaptive
estimates sit below every threshold-specific Wald estimate; no convex
weighting of the Wald estimands rationalizes them, even though every estimated
composite type weight happens to be positive in this application. RT instead
delivers the projected policy surrogate with identification-gap bounds.

For practice, our results operationalize the recommendation of \citet{AndrewsChen2025} that researchers state what their estimators estimate. The instrument-specific Wald estimates reveal the spread behind any pooled coefficient; the implied Wald shares, with the composite type weights they induce on the estimated compliance types, show whose treatment effects a parameter-free weighting's estimand represents; a $J$-rejection under maintained validity reports unequal Wald estimands and calls for a declared target rather than a discarded design; and RT estimates the declared target with full efficiency, or an unattainable one by projection, with bounds on the gap.

Our framework restricts attention to a binary treatment and binary instruments; multinomial treatments and continuous instruments would permit analogous compliance-type projections at the cost of additional structure \citep{AngristSantosTecchio2025, MogstadSantosTorgovitsky2018}, and the conditions for partially saturated covariate adjustment remain open \citep{BlandholBonneyMogstadTorgovitsky2022}. Within this scope, selecting a weighting matrix selects an estimand, and the selection is the researcher's to make deliberately.

\bibliography{efficient_gmm_late}

\appendix
\setcounter{lemma}{0}
\numberwithin{lemma}{section}
\setcounter{assumption}{0}
\numberwithin{assumption}{section}
\setcounter{definition}{0}
\numberwithin{definition}{section}
\setcounter{remark}{0}
\numberwithin{remark}{section}
\setcounter{equation}{0}
\numberwithin{equation}{section}
\setcounter{figure}{0}
\numberwithin{figure}{section}
\setcounter{table}{0}
\numberwithin{table}{section}

\makeatletter
\renewcommand{\theHsection}{App\thesection}
\renewcommand{\theHfigure}{App\thefigure}
\renewcommand{\theHtable}{App\thetable}
\renewcommand{\theHequation}{App\theequation}
\renewcommand{\theHlemma}{App\thelemma}
\renewcommand{\theHassumption}{App\theassumption}
\renewcommand{\theHdefinition}{App\thedefinition}
\renewcommand{\theHremark}{App\theremark}
\renewcommand{\theHcondition}{App\thecondition}
\makeatother

\begin{center}
{\Huge\textbf{Appendix}}
\end{center}
\section{Additional Assumptions and Regularity Conditions}\label{app:regularity}

\begin{assumption}[Sampling, realized equations, and regularity]\label{ass:sampling}\leavevmode
\begin{enumerate}[label=(\alph*)]
  \item \emph{(i.i.d.\ sampling.)} The observations $O_i=(Y_i,D_i,\bZ_i)$, $i = 1,\ldots,n$, are i.i.d.\ from the observed marginal of $P$ on $\mathbb{R} \times \{0,1\} \times \{0,1\}^L$.
  \item \emph{(Stable unit treatment value and realized treatment.)} $D_i = D_i(\bZ_i)$, as in Assumption~\ref{ass:realized_exclusion}, where $D_i(\cdot) : \{0,1\}^L \to \{0,1\}$ is the compliance type of Section~\ref{sec:framework}.
  \item \emph{(Exclusion restriction and observed outcome.)} $Y_i = D_i Y_i(1) + (1 - D_i) Y_i(0)$, as in Assumption~\ref{ass:realized_exclusion}, with $(Y_i(0), Y_i(1))$ the potential outcomes of Section~\ref{sec:framework}.
  \item \emph{(Finite moments.)} $\E[|Y_i(d)|^{2+\delta}] < \infty$ for $d \in \{0,1\}$ and some $\delta > 0$; by (c), $\E[|Y_i|^{2+\delta}] < \infty$.
  \item \emph{(Nondegenerate denominators.)} $\gamma_\ell(P) = \pi_\ell(P) \, p^Z_\ell(P)(1 - p^Z_\ell(P)) = \Cov_P(D_i, Z_{\ell i}) > 0$ for every $\ell = 1, \ldots, L$ (a consequence of Assumption~\ref{ass:relevance}, not an independent restriction), where $\pi_\ell(P) = \E[D_i \mid Z_{\ell i} = 1] - \E[D_i \mid Z_{\ell i} = 0]$ and $p^Z_\ell(P) = \mathbb{P}(Z_{\ell i} = 1)$. The first-stage covariance vector is $\bgamma(P) = (\gamma_1(P), \ldots, \gamma_L(P))'$.
\end{enumerate}
\end{assumption}

\begin{assumption}[Estimator regularity]\label{ass:gmm-reg}\leavevmode
\begin{enumerate}[label=(\alph*)]
  \item \emph{(Common domain.)} The EGMM fixed-point branch selection is restricted to $\mathcal B$.\\
  Under each fixed law  $P$ considered,
  \item \emph{(Interior points.)} $\beta^I(P)$, $\beta^{\TwoS}(P)$, $\beta^{\EGMM}(P)$, and $\beta^{\CUE}(P)$ lie in $\operatorname{int}(\mathcal B)$.
  \item \emph{(Nondegenerate weighting.)} $\bOmega(\beta;P)$ is positive definite for every $\beta\in\mathcal B$.
  \item \emph{(Unique pseudo-true values.)} The population CUE criterion has a unique minimizer over $\mathcal B$, and there is a unique population EGMM fixed point in $\mathcal B$.
\end{enumerate}
\end{assumption}

The cell parameter $\eta(P) = \{\mu_z, m_Y(z), p(z)\}_{z \in \mathcal{Z}}$ collects $\mu_z = \mathbb{P}(\bZ_i = z)$, $m_Y(z) = \E[Y_i \mid \bZ_i = z]$, $p(z) = \E[D_i \mid \bZ_i = z]$ over $\mathcal{Z} = \supp(\bZ_i)$. Let $\mathcal{T}_+(P_0)=\{t\in\mathcal{T}:\theta_t(P_0)>0\}$ denote the positive-mass compliance types at the baseline. For the projection specifications of Section~\ref{sec:specifications}, $\bA_{\mathrm{spec}}(P)$ denotes the type-restricted compliance-weight matrix with rows indexed by $\ell=1,\ldots,L$, columns indexed by $t\in\mathcal{T}_+(P_0)$, and entries $\alpha_{\ell,t}(P)$. In ordered-threshold specifications the columns are $\{t_0,\ldots,t_K\}$.

\begin{assumption}[Local ordered-threshold support]\label{ass:local-threshold-support}
For a threshold-indexed specification, every threshold interval entering it has positive length at $P_0$, equivalently positive type mass, and the realized propensity endpoints defining those intervals are locally separated with stable ordering, each interior endpoint attained by a unique cell of $\supp(\bZ_i)$. For the PRTE specification this condition applies to the full ordered-threshold partition: $0<p_1(P_0)<\cdots<p_K(P_0)<1$, equivalently $\theta_{t_k}(P_0)>0$ for every $k=0,\ldots,K$, with the same endpoint ordering on a neighborhood of $P_0$.
\end{assumption}

For PRTE, write $\mathcal{R}_{t_k}(P)=(p_k(P),p_{k+1}(P)]$, $k=0,\ldots,K$; the full latent-resistance domain is then $[0,1]=\bigcup_{k=0}^{K}\mathcal{R}_{t_k}(P)$ up to endpoints. Every PRTE display containing $\theta_{t_k}^{-1}$, an interval mean over $\mathcal{R}_{t_k}$, or an active threshold target is interpreted over all $k=0,\ldots,K$.

\begin{assumption}[LATE-model regularity]\label{ass:late-regularity}
For each $P \in \PLATE$, the type-conditional outcome distributions $\{F_{t, d}\}_{t \in \mathcal{T},\, d \in \{0, 1\}}$ admit densities $f_{t, d}$ with respect to a common $\sigma$-finite dominating measure on $\mathbb{R}$. The admissible perturbation directions are (i) an instrument-mass coordinate $\mu_z$ by a bounded signed measure with zero total mass, (ii) a positive-support type probability $\theta_r$ inside the relative simplex on $\mathcal{T}_+(P_0)$ by a bounded direction, and (iii) an outcome density $f_{r, d}$ for $r\in\mathcal{T}_+(P_0)$ by a bounded $L^2_0(F_{r, d})$ score direction. Each one-parameter LATE-compatible submodel $\{P_{s} : s \in (-\epsilon, \epsilon)\} \subset \PLATE$ through $P_0$ along such a direction is differentiable in quadratic mean at $s = 0$: there exists a square-integrable score $\dot\ell \in L^2_0(P_0)$ such that $\int \biggl[\sqrt{dP_{s} / dP_0} - 1 - \tfrac{s}{2}\, \dot\ell\biggr]^2 dP_0 \;=\; o(s^2)$ as  $s \to 0$.
\end{assumption}

The weight map $\omega(\cdot):\PLATE\to\Delta^{L-1}$ is Hadamard differentiable at $P_0$ along LATE-compatible local charts with ambient $L^2_0(P_0)$ extension if there is a continuous linear map $\dot\omega(\cdot;P_0):L^2_0(P_0)\to\mathbb{R}^L$ whose restriction to $\TLATE$ gives the pathwise derivative: $[\omega(P_{t_n,h_n})-\omega(P_0)]/t_n\to\dot\omega(h;P_0)$ whenever $h_n\to h$ in $L^2_0(P_0)$ along a LATE-compatible chart $P_{t_n,h_n}\in\PLATE$ with score $h_n$ and $t_n\downarrow 0$. The Riesz representers $\Riesz_\omega=(\Riesz_{\omega,1},\ldots,\Riesz_{\omega,L})'\in L^2_0(P_0)^L$ satisfy $\dot\omega_\ell(h;P_0)=\langle\Riesz_{\omega,\ell},h\rangle_{L^2(P_0)}$ and $\ind_L'\dot\omega(h;P_0)=0$. 

\begin{assumption}[Local-risk uniform integrability]\label{ass:local-risk-ui}
With $R_{n,h}\equiv n\bigl(\hat\beta^{\mathrm{RT}}_{\bomega}-\beta^{\mathrm{RT}}_{\bomega}(P_{n,h})\bigr)^2$, for every finite-dimensional $I\subset\TLATE$ and every $C<\infty$,
\[
\lim_{M\to\infty}\;
\limsup_{n\to\infty}\;
\sup_{\substack{h\in I\\ \|h\|_{L^2(P_0)}\le C}}
\E_{P_{n,h}}\!\left[
R_{n,h}\ind\{R_{n,h}>M\}
\right]=0.
\]
\end{assumption}

\begin{assumption}[Riesz representer consistency]\label{ass:riesz-consistency}
The estimated weight-map Riesz representers satisfy $\max_{\ell\leq L}\frac1n\sum_{i=1}^n\bigl(\hat\Riesz_{\omega,\ell}(O_i)-\Riesz_{\omega,\ell}(O_i;P_0)\bigr)^2=o_p(1). $
\end{assumption}

\begin{assumption}[Strict complementarity]\label{ass:strict-complementarity}
The projection Gram matrix induced by the specification's metric weights is $\bG(P)\equiv\bA_{\mathrm{spec}}(P)\,\diag(\bkappa(P))\,\bA_{\mathrm{spec}}(P)'$, with $\bG_{I_0I_0}(P)$ its submatrix on the active index set. At $P_0$, the projection solution $(\bomega^\dagger, \bnu^\dagger, \boldsymbol{m}^\dagger)(\bpsi^\star; P_0)$ satisfies $\omega^\dagger_\ell(\bpsi^\star; P_0) > 0$ for every $\ell \in I_0 \equiv  I(\bpsi^\star; P_0)$ and $m^\dagger_\ell(\bpsi^\star; P_0) > 0$ for every $\ell \in I_0^c$. Let $\mathcal{K}_{I_0}=\{v\in\mathbb{R}^{|I_0|}:\ind_{I_0}'v=0\}$. The active-face curvature condition is $v'\bG_{I_0 I_0}(P_0)v>0
$for every nonzero $v\in\mathcal{K}_{I_0}$. Equivalently, the bordered KKT matrix $\begin{psmallmatrix}\bG_{I_0I_0}(P_0)&\ind_{I_0}\\ \ind_{I_0}'&0\end{psmallmatrix}$ is nonsingular.
\end{assumption}

\begin{assumption}[Locally identified projection primitives]\label{ass:projection-dictionary}
For the projection specification, the positive-support type set is locally constant, $\mathcal{T}_+(P)=\mathcal{T}_+(P_0)$ on a neighborhood of $P_0$, and finite. For every $t\in\mathcal{T}_+(P_0)$ and every $\ell$, the instrument-type weight $\alpha_{\ell,t}(P)$, target weight $\psi^\star_t(P)$, and metric weight $\kappa_t(P)$ are point-identified observed-data functionals on a neighborhood of $P_0$. All denominators in these primitive formulas are bounded away from zero, and the three maps are Hadamard differentiable at $P_0$ with Riesz representers $\xi_{\ell,t}$, $\xi^\star_t$, and $\zeta_t$ in $L^2_0(P_0)$. The corresponding empirical plug-ins based on $\hat P_n$ obey the associated von Mises expansions with the same Riesz representers, uniformly over $\mathcal{T}_+(P_0)$.
\end{assumption}

\begin{assumption}[Projection tangent compatibility]\label{ass:projection-tangent}
For a LATE-efficiency claim about a projected target, the scalar representer obtained by inserting the active-face KKT derivative of $P\mapsto\bomega^\dagger(\bpsi^\star;P)$ into~\eqref{eq:eif-general} belongs to $\TLATEbar$.
\end{assumption}

\begin{assumption}[PRTE rank condition]\label{ass:prte-rank}
The PRTE Gram matrix under the inverse-type-mass metric,$
  \bG^{\mathrm{PRTE}}(P)
  \;=\;
  \bA_{\mathrm{spec}}(P)\,
  \diag\{\theta_{t_k}(P)^{-1}: k=0,\ldots,K\}\,
  \bA_{\mathrm{spec}}(P)',
$
satisfies
$
  v'\bG^{\mathrm{PRTE}}(P)v > 0 \quad\text{for every }v\neq0\text{ with }\ind_L'v=0.
$

\end{assumption}

\begin{assumption}[PRTE policy-weight regularity]\label{ass:prte-weight-regularity}
For the discrete-instrument policy $(F^Z_0\to F^Z_1)$, the threshold-interval mass $\psi^{\star,\mathrm{PRTE}}_{t_k}(P)$ of~\eqref{eq:prte_integral} has the closed ratio form $
\psi^{\star,\mathrm{PRTE}}_{t_k}(P)=\frac{\theta_{t_k}(P)\,\Delta^{\mathrm{PRTE}}_k(P)}{\bar\Delta(P)}, \bar\Delta(P)=\E_{F^Z_1}[D_i]-\E_{F^Z_0}[D_i]>0,
$
in which the treatment-shift numerator $\Delta^{\mathrm{PRTE}}_k(P)$ and the normalizer $\bar\Delta(P)$ are Hadamard-differentiable functionals of the realized cell propensities and cell masses $\{p(z)(P),\mu_z(P)\}_{z\in\mathcal{Z}}$ at $P_0$; the policy laws $F^Z_0,F^Z_1$ are fixed or depend on $P$ only through the baseline instrument distribution.
\end{assumption}

\clearpage

\begin{center}
\textbf{Online Appendix to ``When Is GMM Actually LATE? Weighting Matrices and Causal Interpretation in Overidentified IV''}\\
Chun Pang Chow and Hiroyuki Kasahara
\end{center}
\appendix
\numberwithin{lemma}{section}
\numberwithin{assumption}{section}
\numberwithin{condition}{section}
\numberwithin{definition}{section}
\numberwithin{remark}{section}
\numberwithin{equation}{section}
\numberwithin{figure}{section}
\numberwithin{table}{section}
\setcounter{section}{0}
\renewcommand{\thesection}{S.\arabic{section}}
\setcounter{lemma}{0}
\setcounter{assumption}{0}
\setcounter{condition}{0}
\setcounter{definition}{0}
\setcounter{remark}{0}
\setcounter{equation}{0}
\setcounter{figure}{0}
\setcounter{table}{0}

\makeatletter
\renewcommand{\theHsection}{Supp\thesection}
\renewcommand{\theHfigure}{Supp\thefigure}
\renewcommand{\theHtable}{Supp\thetable}
\renewcommand{\theHequation}{Supp\theequation}
\renewcommand{\theHlemma}{Supp\thelemma}
\renewcommand{\theHassumption}{Supp\theassumption}
\renewcommand{\theHdefinition}{Supp\thedefinition}
\renewcommand{\theHremark}{Supp\theremark}
\renewcommand{\theHcondition}{Supp\thecondition}
\makeatother

\section{Proofs for main manuscript}\label{app:proofs}

\subsection{Proof of Proposition~\ref{prop:wald_decomp}}

The reduced-form coefficient is $\rho_\ell = \E[Y_i \mid Z_\ell = 1] - \E[Y_i \mid Z_\ell = 0]$. Write $Y_i = Y_i(0) + (Y_i(1) - Y_i(0)) D_i(\bZ_i)$. Taking conditional expectations:
$\E[Y_i \mid Z_\ell = z_\ell] = \E[Y_i(0)] + \E\bigl[(Y_i(1) - Y_i(0))\, D_i(z_\ell, \bZ_{-\ell}) \mid Z_\ell = z_\ell\bigr].$
Applying joint independence (Assumption~\ref{ass:joint_indep}) cell by cell:
$\E\bigl[(Y_i(1) - Y_i(0))\, D_i(z_\ell, z_{-\ell}) \mid \bZ_i = (z_\ell, z_{-\ell})\bigr] = \E\bigl[(Y_i(1) - Y_i(0))\, D_i(z_\ell, z_{-\ell})\bigr].$
Summing over $z_{-\ell}$ with weights $\mathbb{P}(\bZ_{-\ell} = z_{-\ell} \mid Z_\ell = z_\ell)$ (the conditional probabilities $q_\ell$ and $q^0_\ell$ from the proposition statement) and conditioning on compliance type $D_i(\cdot) = t$:
$\E\bigl[(Y_i(1) - Y_i(0))\, D_i(z_\ell, z_{-\ell})\bigr] = \sum_{t \in \mathcal{T}_+(P)} \LATE_t \cdot \theta_t \cdot t(z_\ell, z_{-\ell}).$
Substituting and taking the difference across $z_\ell \in \{0,1\}$: $\rho_\ell = \sum_{t \in \mathcal{T}_+(P)} \LATE_t \cdot \theta_t \cdot \varphi_{\ell,t},$
where $\varphi_{\ell,t} = \sum_{z_{-\ell}} [t(1, z_{-\ell}) q_\ell(z_{-\ell}) - t(0, z_{-\ell}) q^0_\ell(z_{-\ell})]$. The same derivation without $(Y_i(1) - Y_i(0))$ gives $\pi_\ell = \sum_{t\in\mathcal{T}_+(P)} \theta_t \cdot \varphi_{\ell,t}$. Dividing: $\Wald_\ell = \rho_\ell/\pi_\ell = \sum_{t\in\mathcal{T}_+(P)} \LATE_t \cdot \alpha_{\ell,t}$, with $\alpha_{\ell,t} = \theta_t \varphi_{\ell,t}/\pi_\ell$ and $\sum_{t\in\mathcal{T}_+(P)} \alpha_{\ell,t} = 1$.  

\medskip

\subsection{Proof of Lemma~\ref{lem:di}}

Adding and subtracting $\sum_{z_{-\ell}} t(0, z_{-\ell})\, q_\ell(z_{-\ell})$ in the definition $\varphi_{\ell,t} = \sum_{z_{-\ell}} [t(1, z_{-\ell})\, q_\ell(z_{-\ell}) - t(0, z_{-\ell})\, q^0_\ell(z_{-\ell})]$ splits it as $\varphi_{\ell,t} = \varphi_{\ell,t}^D(P) + \varphi_{\ell,t}^I(P)$: the \emph{direct} component $\varphi_{\ell,t}^D(P) = \sum_{z_{-\ell}} [t(1, z_{-\ell}) - t(0, z_{-\ell})]\, q_\ell(z_{-\ell})$ switches $Z_\ell$ from $0$ to $1$ at the fixed conditional law $q_\ell$, and the \emph{indirect} component $\varphi_{\ell,t}^I(P) = \sum_{z_{-\ell}} t(0, z_{-\ell})\, [q_\ell(z_{-\ell}) - q^0_\ell(z_{-\ell})]$ moves that law from $q^0_\ell$ to $q_\ell$ at the fixed treatment $t(0, \cdot)$. Monotonicity (Assumption~\ref{ass:monotonicity}) gives $t(1, z_{-\ell}) \geq t(0, z_{-\ell})$ and $q_\ell(z_{-\ell}) \geq 0$, hence $\varphi_{\ell,t}^D(P) \geq 0$.

\subsection{Proof of Proposition~\ref{prop:prd}}

By Lemma~\ref{lem:di}, it suffices to show $\varphi_{\ell,t}^I(P) \geq 0$. The function $z_{-\ell} \mapsto t(0, z_{-\ell})$ is nondecreasing: monotonicity for all instruments (Assumption~\ref{ass:monotonicity}) requires $D_i(z)$ nondecreasing in each coordinate; hence, fixing $z_\ell = 0$, $t(0, z_{-\ell})$ is nondecreasing in each component of $z_{-\ell}$. Since $t(0, \cdot)$ is nondecreasing and $\bZ_{-\ell}$ is PRD on $Z_\ell$ (Assumption~\ref{ass:prd}), $\varphi_{\ell,t}^I(P) = \E[t(0, \bZ_{-\ell}) \mid Z_\ell = 1] - \E[t(0, \bZ_{-\ell}) \mid Z_\ell = 0] \geq 0$. Combined with $\varphi_{\ell,t}^D(P) \geq 0$, $\theta_t \geq 0$, and $\pi_\ell(P) > 0$ from the hypothesis of the proposition, we obtain $\alpha_{\ell,t} \geq 0$.

\subsection{Proof of Proposition~\ref{prop:weighted_iv}}

\emph{Part (a).} At any $P$ with $\bW(P)\succ0$ and $\gamma_\ell(P)\neq0$ for every $\ell$, the population criterion $\beta\mapsto\bg_\infty(\beta;P)'\bW(P)\bg_\infty(\beta;P)$, with $\bg_\infty(\beta;P)=\bg_\infty(0;P)-\beta\bgamma(P)$, is a strictly convex quadratic in $\beta$ minimized at $\beta^{\mathrm{GMM}}_{\bW}(P)=\tfrac{\bgamma(P)'\bW(P)\bg_\infty(0;P)}{\bgamma(P)'\bW(P)\bgamma(P)}$. Substituting $g_{\infty,\ell}(0;P)=\gamma_\ell(P)\Wald_\ell(P)$ yields the affine form~\eqref{eq:weighted_iv} with $\lambda_\ell(\bW(P);P)=\tfrac{\gamma_\ell(P)[\bW(P)\bgamma(P)]_\ell}{\bgamma(P)'\bW(P)\bgamma(P)}$; since $\sum_\ell\gamma_\ell(P)[\bW(P)\bgamma(P)]_\ell=\bgamma(P)'\bW(P)\bgamma(P)$, the weights sum to one.

\emph{Part (b).} Under Assumption~\ref{ass:sampling}(a),(d),(e) the law of large numbers gives $\hat\bgamma\to_p\bgamma(P_0)$ and $\widehat{\Cov}(Y_i,Z_{\ell i})\to_p \Cov_{P_0}(Y_i,Z_{\ell i})$; for the moment vector underlying $\bW(\cdot)$ (Definition~\ref{def:gmm-weight-map}), $m(\hat P_n)\to_p m(P_0)$. Assumption~\ref{ass:relevance} gives $\gamma_\ell(P_0)\neq0$, hence $\hat\gamma_\ell\neq0$ for every $\ell$ with probability tending to one. By continuity of $\tilde\bW$ at $m(P_0)$ and of $\lambda_{\min}(\cdot)$, $\bW(\hat P_n)=\tilde\bW(m(\hat P_n))\to_p\bW(P_0)$ and $\lambda_{\min}(\bW(\hat P_n))\to_p\lambda_{\min}(\bW(P_0))>0$; hence $\bW(\hat P_n)\succ0$ and $\hat\bgamma'\bW(\hat P_n)\hat\bgamma>0$ with probability tending to one. On that event the first-order condition for $Q_n(\beta)=\bg_n(\beta)'\bW(\hat P_n)\bg_n(\beta)$ together with $g_{n,\ell}(0)=\hat\gamma_\ell\widehat\Wald_\ell$ reproduces the part-(a) algebra:
$
\hat\beta^{\mathrm{GMM}}_{\bW}=\frac{\hat\bgamma'\bW(\hat P_n)\bg_n(0)}{\hat\bgamma'\bW(\hat P_n)\hat\bgamma}=\sum_{\ell=1}^L\lambda_\ell(\bW(\hat P_n);\hat P_n)\widehat\Wald_\ell$, $\sum_{\ell=1}^L\lambda_\ell(\bW(\hat P_n);\hat P_n)=1.
$\
The map $(\bgamma,\bW)\mapsto\gamma_\ell[\bW\bgamma]_\ell/(\bgamma'\bW\bgamma)$ is continuous on $\{\bW\succ0,\ \bgamma\neq0\}$, and $\widehat\Wald_\ell=\widehat{\Cov}(Y_i,Z_{\ell i})/\hat\gamma_\ell\to_p\Wald_\ell(P_0)$ at the nonzero limiting denominator, and the continuous mapping theorem gives $\lambda_\ell(\bW(\hat P_n);\hat P_n)\to_p\lambda_\ell(\bW(P_0);P_0)$ and $\hat\beta^{\mathrm{GMM}}_{\bW}\to_p\sum_\ell\lambda_\ell(\bW(P_0);P_0)\Wald_\ell(P_0)=\beta^{\mathrm{GMM}}_{\bW}(P_0)$.

\subsection{Proof of Proposition~\ref{prop:jtest}}

Throughout the proof, under each fixed law considered, after empirical demeaning each entry of $\widehat\bOmega(\beta)$ is a quadratic polynomial in $\beta$ whose coefficients obey the law of large numbers, and $\bOmega(\beta;P)$ is the corresponding population quadratic; hence $\sup_{\beta\in\mathcal B}\|\widehat\bOmega(\beta)-\bOmega(\beta;P)\|=o_p(1)$, and compactness of $\mathcal B$ with Assumption~\ref{ass:gmm-reg}(c) yields constants $0<c<C<\infty$ such that $c\leq\inf_{\beta\in\mathcal B}\lambda_{\min}\{\widehat\bOmega(\beta)\}\leq\sup_{\beta\in\mathcal B}\lambda_{\max}\{\widehat\bOmega(\beta)\}\leq C$ with probability tending to one (the eigenvalue event). The law of large numbers, Assumption~\ref{ass:sampling}(e), and Assumption~\ref{ass:gmm-reg}(b) give $\hat\beta^I\xrightarrow{p}\beta^I(P)\in\operatorname{int}(\mathcal B)$; hence $\hat\beta^I\in\mathcal B$ with probability tending to one.

\emph{Part (a): Null limit.} Under $H_0$, put $\bG_n\equiv\sqrt n\,\bg_n(\beta_0)$, $\bOmega_0\equiv\bOmega(\beta_0;P_0)$, and $\bW_*\equiv\bOmega_0^{-1}$; the central limit theorem for the demeaned sample covariance moments gives $\bG_n\Rightarrow N(\mathbf 0,\bOmega_0)$, and $\hat\bgamma\xrightarrow{p}\bgamma(P_0)\equiv\bgamma$. On the eigenvalue event, the defining CUE inequality gives
$
\frac{1}{C}\|\bg_n(\hat\beta^{\CUE})\|^2
\leq \widehat Q(\hat\beta^{\CUE})
\leq \widehat Q(\beta_0)
\leq \frac{1}{c}\|\bg_n(\beta_0)\|^2,
$
hence $\bg_n(\hat\beta^{\CUE})=O_p(n^{-1/2})$; the exact affine identity $\bg_n(\hat\beta^{\CUE})=\bg_n(\beta_0)-\hat\bgamma(\hat\beta^{\CUE}-\beta_0)$ and $\|\hat\bgamma\|\xrightarrow{p}\|\bgamma\|>0$ then give $\hat\beta^{\CUE}-\beta_0=O_p(n^{-1/2})$ without invoking the CUE first-order condition. For every fixed $H<\infty$, the affine identity and the uniform convergence of the inverse weighting matrix give $\sup_{|u|\leq H}\bigl|n\widehat Q(\beta_0+u/\sqrt n)-(\bG_n-\bgamma u)'\bW_*(\bG_n-\bgamma u)\bigr|=o_p(1)$; since $\sqrt n(\hat\beta^{\CUE}-\beta_0)=O_p(1)$ and the limiting quadratic has curvature $\bgamma'\bW_*\bgamma>0$, the argmin comparison yields $\sqrt n(\hat\beta^{\CUE}-\beta_0)=(\bgamma'\bW_*\bgamma)^{-1}\bgamma'\bW_*\bG_n+o_p(1)$. For TSGMM and EGMM, with $\bW_n$ the positive-definite matrix in the corresponding plug-in normal equation (its evaluation point lies in $\mathcal B$: $\hat\beta^I$ by the preliminary membership, $\hat\beta^{\EGMM}$ by the search restriction), the affine identity gives, on the same event and before any weighting-matrix limit is invoked, $|\hat\beta-\beta_0|=|\hat\bgamma'\bW_n\bg_n(\beta_0)|/(\hat\bgamma'\bW_n\hat\bgamma)\leq(C/c)\,\|\bg_n(\beta_0)\|/\|\hat\bgamma\|=O_p(n^{-1/2})$; uniform convergence of $\widehat\bOmega(\beta)$ and continuity of $\bOmega(\cdot;P_0)$ then send both weights to $\bW_*$, and substitution in the exact normal equations gives all three estimators the CUE expansion. Consequently, for $\mathrm{est}\in\{\TwoS,\EGMM,\CUE\}$,
$
\sqrt n\,\bg_n(\hat\beta^{\mathrm{est}})
=
\left\{\mathbf I_L-\bgamma(\bgamma'\bW_*\bgamma)^{-1}\bgamma'\bW_*\right\}\bG_n+o_p(1),
$
hence $J_n^{\mathrm{est}}\Rightarrow Y'\bM_{\bOmega}Y$ with $Y\sim N(\mathbf 0,\bOmega_0)$; since $\bM_{\bOmega}\bOmega_0\bM_{\bOmega}=\bM_{\bOmega}$ and $\operatorname{rank}(\bM_{\bOmega})=L-1$, this quadratic form is $\chi^2_{L-1}$ for all three estimators.

\emph{Part (b): Local-alternative limit.} The hypothesized limit gives $\bG_n=O_{P_{n,h}}(1)$, and contiguity transfers every part-(a) $o_p(1)$ remainder together with the preliminary events (the uniform weighting-matrix convergence, the eigenvalue event, and $\hat\beta^I\in\mathcal B$); the part-(a) criterion inequality and normal-equation bound then give $\hat\beta^{\mathrm{est}}-\beta_0=O_{P_{n,h}}(n^{-1/2})$ for $\mathrm{est}\in\{\TwoS,\EGMM,\CUE\}$, uniform convergence and continuity deliver the common limiting weight $\bW_*$, and the part-(a) expansion holds with $o_{P_{n,h}}(1)$ remainders. Substitution into the $J$-statistics yields $Y'\bM_{\bOmega}Y$ with $Y\sim N\{\dot{\bd},\bOmega_0\}$, hence $\chi^2_{L-1}(\delta^2)$ with $\delta^2=\dot{\bd}'\bM_{\bOmega}\dot{\bd}$. For the Wald-movement form, the exact identity $\bigl[\E_P\bg(\beta_0;P)\bigr]_\ell=\Cov_P(Y-\beta_0D,\,Z_\ell)=\gamma_\ell(P)\{\Wald_\ell(P)-\beta_0\}$ holds at every $P\in\PLATE$; product contiguity keeps $n$ times the squared Hellinger distance between the observed marginals of $P_{n,h}$ and $P_0$ bounded (else the product affinities $(1-h^2)^n$ vanish and the products separate entirely), those marginals therefore converge in total variation, and the bounded moments $\gamma_\ell(P_{n,h})\to\gamma_\ell(P_0)$; hence $\Wald_\ell(P_{n,h})=\beta_0+n^{-1/2}\alpha_\ell+o(n^{-1/2})$ gives $\sqrt n\,\E_{P_{n,h}}[\bg(\beta_0;P_{n,h})]\to\diag\{\gamma_\ell(P_0)\}\boldsymbol\alpha=\dot{\bd}$. Finally, $\bM_{\bOmega}\bgamma=\mathbf 0$ by direct substitution, and with $\bM_{\bOmega}\bOmega_0\bM_{\bOmega}=\bM_{\bOmega}$ and the overidentifying rank $L-1$ of the preamble, $\bM_{\bOmega}$ is positive semidefinite with kernel $\mathrm{span}(\bgamma)$; since every $\gamma_\ell(P_0)>0$, $\delta^2=0$ if and only if all $\alpha_\ell$ coincide, and $\delta^2>0$ otherwise.

\emph{Part (c): Fixed-alternative divergence.} Under $H_1:\Wald_\ell\neq\Wald_k$ for some $\ell,k$, no scalar $\beta$ zeros all $L$ population moments. The affine law of large numbers gives uniform convergence of $\bg_n(\beta)$ on $\mathcal B$; with the preliminary convergence of $\widehat\bOmega$ and Assumption~\ref{ass:gmm-reg}(c), $\widehat Q$ converges uniformly to $Q(\beta;P)\equiv\E[\bg(\beta;P)]'\bOmega(\beta;P)^{-1}\E[\bg(\beta;P)]$; compactness and uniqueness give CUE consistency by the consistency theorem for extremum estimators, the identity and TSGMM updates are continuous at their interior population solutions, and the EGMM sample update converges uniformly on $\mathcal B$ to the continuous population map, whose unique fixed point attracts any selected sample fixed point in $\mathcal B$ because the fixed-point residual is bounded away from zero off every neighborhood of it. All three estimators and their weight matrices thus converge to their pseudo-true values, and $n^{-1}J_n^{\mathrm{est}}\xrightarrow{p}c^{\mathrm{est}}_*$ with the constants of Proposition~\ref{prop:jtest}(c). Compactness, continuity, and the absence of a common moment root give $c^{\CUE}_*=\min_{\beta\in\mathcal B}Q(\beta;P)>0$; the same absence gives $\E[\bg(\beta^{\TwoS};P)]\neq\mathbf 0$, and positive definiteness of $\bOmega(\beta^I;P)$ (Assumptions~\ref{ass:gmm-reg}(b),(c)) gives $c^{\TwoS}_*>0$; and $\beta^{\EGMM}\in\mathcal B$ (Assumption~\ref{ass:gmm-reg}(b)) gives $c^{\EGMM}_*=Q(\beta^{\EGMM};P)\ge\min_{\beta\in\mathcal B}Q(\beta;P)=c^{\CUE}_*$.

\subsection{Proof of Corollary~\ref{cor:2sls_diag}}

Under Assumption~\ref{ass:diag_instruments}, $\bSigmaZ = \mathrm{diag}(p^Z_\ell(1-p^Z_\ell))$. Substituting into the 2SLS weight formula of Section~\ref{sec:2sls} gives $[\bSigmaZ^{-1}\bgamma]_\ell = \gamma_\ell / [p^Z_\ell(1-p^Z_\ell)] = \pi_\ell$, and $\lambda^{2\mathrm{SLS}}_\ell = \gamma_\ell \cdot \pi_\ell / \sum_k \gamma_k \pi_k = \pi_\ell^2 p^Z_\ell(1-p^Z_\ell) / \sum_k \pi_k^2 p^Z_k(1-p^Z_k) \geq 0$.

\subsection{Proof of Corollary~\ref{cor:egmm_diag}}\label{app:proof_cor_egmm_diag}

With $\bOmega(\beta^{\EGMM};P)$ diagonal (Lemma~\ref{lem:omega-diagonal}), Proposition~\ref{prop:weighted_iv} applied to $\bW^{\EGMM}(P) = \bOmega(\beta^{\EGMM};P)^{-1}$ yields $\lambda^{\EGMM}_\ell = \gamma_\ell^2/[\bOmega]_{\ell\ell} \cdot \bigl(\sum_k \gamma_k^2/[\bOmega]_{kk}\bigr)^{-1} = \pi_\ell^2 p^Z_\ell(1-p^Z_\ell)/\sigma^2_{\epsilon,\ell} \cdot \bigl(\sum_k \pi_k^2 p^Z_k(1-p^Z_k)/\sigma^2_{\epsilon,k}\bigr)^{-1}$.

\subsection{Proof of Corollary~\ref{cor:2step_diag}}\label{app:proof_cor_2step_diag}

The identity-weight FOC $\sum_\ell \gamma_\ell\,g_\ell(\beta;P) = 0$ at $\beta = \beta^I$ together with $g_\ell(\beta;P) = \gamma_\ell\bigl(\Wald_\ell - \beta\bigr)$ gives $\beta^I = \bgamma'\E[\bg(0;P)]/\|\bgamma\|^2 = \sum_\ell \gamma_\ell^2 \Wald_\ell / \|\bgamma\|^2$ and $\lambda^I_\ell = \gamma_\ell^2/\|\bgamma\|^2 = \pi_\ell^2 s_\ell^2 / \sum_k \pi_k^2 s_k^2$ under $\gamma_\ell = \pi_\ell s_\ell$. Applying Proposition~\ref{prop:weighted_iv}(a) to $\bW^{\TwoS}(P) = \bOmega(\beta^I(P);P)^{-1}$ in the diagonal case (Lemma~\ref{lem:omega-diagonal}) delivers~\eqref{eq:2step_diag_weights}; the closed-form denominator follows from $\bgamma'\bW^{\TwoS}\bgamma = \sum_k \pi_k^2 s_k/\sigma^2_{\epsilon,k}(\beta^I)$.

\subsection{Proof of Corollary~\ref{cor:cue_diag}}\label{app:proof_cor_cue_diag}

Under diagonality at $\beta^{\CUE}$ (Lemma~\ref{lem:omega-diagonal}) with $[\bOmega(\beta;P)]_{\ell\ell} = s_\ell \sigma^2_{\epsilon,\ell}(\beta)$, $\gamma_\ell = \pi_\ell s_\ell$, and $g_\ell(\beta;P) = \gamma_\ell\bigl(\Wald_\ell - \beta\bigr)$ (the last identity follows from $g(\beta;P) = \E[\bg(0;P)] - \beta\bgamma$ and $[\E[\bg(0;P)]]_\ell = \Cov(Y_i,Z_{\ell i}) = \gamma_\ell\Wald_\ell$),
$
  \bigl(\bOmega(\beta;P)^{-1}g(\beta;P)\bigr)_\ell \;=\; \frac{g_\ell(\beta;P)}{s_\ell\,\sigma^2_{\epsilon,\ell}(\beta)} \;=\; \frac{\pi_\ell(\Wald_\ell - \beta)}{\sigma^2_{\epsilon,\ell}(\beta)}.
$
Since $[\bOmega_\beta(\beta;P)]_{\ell\ell} = s_\ell\,{\partial \sigma^2_{\epsilon,\ell}(\beta)/\partial\beta}$ and all off-diagonal entries vanish, the variance-score correction defined in~\eqref{eq:cue_foc} is
$
  R(\beta;P) \;=\; \sum_{\ell=1}^L \bigl(\bOmega^{-1}g\bigr)_\ell^2\,[\bOmega_\beta]_{\ell\ell} \;=\; \sum_{\ell=1}^L \frac{\pi_\ell^2\,s_\ell\,(\Wald_\ell - \beta)^2\,{\partial \sigma^2_{\epsilon,\ell}(\beta)/\partial\beta}}{\sigma^2_{\epsilon,\ell}(\beta)^2}.
$
Dividing by $2\,\bgamma'\bOmega(\beta;P)^{-1}\bgamma = 2\sum_k \pi_k^2 s_k/\sigma^2_{\epsilon,k}(\beta)$ and using $(\partial \sigma^2_{\epsilon,\ell}(\beta)/\partial\beta)/\sigma^2_{\epsilon,\ell}(\beta) = \partial_\beta\log \sigma^2_{\epsilon,\ell}(\beta)$ yields~\eqref{eq:cue_diag_remainder}, with prefactor $\lambda^{\CUE}_\ell$ at $\beta = \beta^{\CUE}$. The Wald-share formula~\eqref{eq:cue_diag_weights} is Proposition~\ref{prop:weighted_iv}(a) applied to $\bW^{\CUE}(P) = \bOmega(\beta^{\CUE}(P);P)^{-1}$ in the diagonal case, with $\bgamma'\bW^{\CUE}\bgamma = \sum_k \pi_k^2 s_k/\sigma^2_{\epsilon,k}(\beta^{\CUE})$.  

\subsection{Proof of Lemma~\ref{lem:type_interval}}

Under Assumption~\ref{ass:latent_index}, $D_i = \ind\{p(\bZ_i) \geq U_i\}$ with the latent resistance normalized to $U_i \sim \mathrm{Unif}(0,1)$ independent of $\bZ_i$; hence $D_i(z) = \ind\{p(z) \geq U_i\}$. Let $p_1 < \cdots < p_K$ be the distinct ordered values of $\{p(z) : z \in \supp(\bZ)\}$, and set $p_0 = 0$ and $p_{K+1} = 1$, allowing $p_0 = p_1$ (no always-takers) or $p_K = p_{K+1}$ (no never-takers), in which case the corresponding boundary interval is empty and its type carries zero mass. For $U_i \in \mathcal{R}_{t_k} = (p_k, p_{k+1}]$, the response is $D_i(z) = \ind\{p(z) \geq p_{k+1}\}$, identical for every $z \in \supp(\bZ)$ and unchanged across resistances within the interval. The realized-support treatment-response patterns thus reduce to threshold patterns, one for each \(\mathcal{R}_{t_k}\), giving the stated bijection $t_k \leftrightarrow \mathcal{R}_{t_k}$ with $\theta_{t_k} = \mathbb{P}(U_i \in \mathcal{R}_{t_k}) = p_{k+1} - p_k$. The intervals are non-overlapping and exhaust $(0,1]$ by construction. Resistances $U_i \leq p_1$ are overcome at every $z$, and $U_i > p_K$ are overcome at no $z$.

\subsection{Proof of Proposition~\ref{prop:mte_representation}}

By Assumption~\ref{ass:latent_index}, $Y_i = Y_i(0) + (Y_i(1) - Y_i(0))\ind\{p(\bZ_i) \geq U_i\}$, with $\E[Y_i \mid \bZ_i = z] = \E[Y_i(0)] + \int_0^{p(z)} \MTE(u)\,du$. Taking conditional expectations given $Z_\ell = z_\ell$, differencing, and applying Fubini's theorem (licensed by the potential-outcome moment bound of Assumption~\ref{ass:sampling}(d), which makes $\E[Y_i(0)]$ and the MTE integral finite): $\rho_\ell = \int_0^1 \MTE(u)[\mathbb{P}(p(\bZ) \geq u \mid Z_\ell = 1) - \mathbb{P}(p(\bZ) \geq u \mid Z_\ell = 0)]\,du.$ Dividing by $\pi_\ell$ gives $\Wald_\ell = \int_0^1 \MTE(u)\,\alpha^{\MTE}_\ell(u)\,du$ with $\alpha^{\MTE}_\ell$ as in \eqref{eq:wald_mte}. By Proposition~\ref{prop:weighted_iv}(a) at the relevant weighting matrix map, $\beta^{\mathrm{GMM}}_{\bW(\cdot)}(P) = \sum_\ell \lambda_\ell(\bW(P);P)\, \Wald_\ell(P) = \int_0^1 \MTE(u)\,\psi^{\MTE}(u;\blambda(\bW(P);P))\,du$.

\subsection{Proof of Proposition~\ref{prop:prd_mte}}

Throughout, $p(z) \equiv \E[D_i(z)]$ denotes the structural propensity, well-defined for every $z \in \{0,1\}^L$ by Assumption~\ref{ass:joint_indep} and equal to $\mathbb{P}(D_i = 1 \mid \bZ_i = z)$ on $\supp(\bZ)$; the denominators are positive by Assumption~\ref{ass:relevance}. Decompose the numerator of $\alpha^{\MTE}_\ell(u)$ by adding and subtracting $\E[\ind\{p(0, Z_{-\ell}) \geq u\} \mid Z_\ell = 1]$ into a direct effect $\Delta^D_\ell(u) = \E[\ind\{p(1, Z_{-\ell}) \geq u\} - \ind\{p(0, Z_{-\ell}) \geq u\} \mid Z_\ell = 1]$ and an indirect effect $\Delta^I_\ell(u) = \E[\ind\{p(0, Z_{-\ell}) \geq u\} \mid Z_\ell = 1] - \E[\ind\{p(0, Z_{-\ell}) \geq u\} \mid Z_\ell = 0]$. Here $\Delta^D_\ell(u) \geq 0$ by monotonicity. For $\Delta^I_\ell(u)$: monotonicity for all instruments ensures $p(0, z_{-\ell})$ is nondecreasing in $z_{-\ell}$; hence $f_u(z_{-\ell}) = \ind\{p(0, z_{-\ell}) \geq u\}$ is nondecreasing. PRD gives $\E[f_u(Z_{-\ell}) \mid Z_\ell = 1] \geq \E[f_u(Z_{-\ell}) \mid Z_\ell = 0]$, i.e., $\Delta^I_\ell(u) \geq 0$. Since the denominator $\pi_\ell > 0$, $\alpha^{\MTE}_\ell(u) \geq 0$.

\subsection{Proof of Proposition~\ref{prop:rt_types}}

Substitute $\Wald_\ell = \sum_{t\in\mathcal{T}_+(P)} \alpha_{\ell,t} \, \LATE_t$ (Proposition~\ref{prop:wald_decomp}) into $\beta^{\mathrm{RT}}_{\bomega}(P) = \sum_\ell \omega_\ell \Wald_\ell(P)$ and exchange the order of summation to obtain~\eqref{eq:rt_type_decomp}. Non-negativity: PRD ensures $\alpha_{\ell,t} \geq 0$ for all $t, \ell$ (Proposition~\ref{prop:prd}), and $\omega_\ell \geq 0$ by construction; hence $\psi_t(\omega) \geq 0$. Summing: $\sum_{t\in\mathcal{T}_+(P)} \psi_t(\omega) = \sum_\ell \omega_\ell \sum_{t\in\mathcal{T}_+(P)} \alpha_{\ell,t} = \sum_\ell \omega_\ell = 1$.

\subsection{Proof of Proposition~\ref{prop:eif-decomp}}\label{app:proof_thm_pathwise}\label{app:proof_thm_lam_bound}

\noindent\emph{(i) LAM lower bound.} Fix $h_1, \ldots, h_m$ in a finite-dimensional subspace generated by bounded structural-coordinate scores: instrument-mass directions, positive-support type-probability directions inside the relative simplex on $\mathcal{T}_+(P_0)$, and bounded type-conditional outcome-density directions for positive-support types. The instrument-law and type-probability coordinates are finite-dimensional; their bounded generators span the same directions exactly. For each positive-mass outcome component $(r,d)$, $L^\infty(F_{r,d})\cap L^2_0(F_{r,d})$ is dense in $L^2_0(F_{r,d})$ by truncation and recentering. The structural-to-observed score map is a finite positive-mixture linear map, hence continuous in $L^2(P_0)$. Therefore any $L^2$ tangent-closure element can be approximated by bounded generated scores, and the lower-bound supremum passes to the limit by continuity of the $L^2$ inner product and of the projection norm. For local parameter $t\in\mathbb{R}^m$, construct $P_{n,t}$ directly in the LATE structural chart. Write the bounded coordinate directions as $a^Z_j(z)$ for the instrument law, $a^\Theta_j(r)$ for positive-support types $r\in\mathcal{T}_+(P_0)$ with $\sum_r \theta_r(P_0)a^\Theta_j(r)=0$, and $a^F_{j,r,d}(y)$ for type-outcome densities with $\int a^F_{j,r,d}\,dF_{r,d}=0$. Set $\mu_z^{(n,t)} =
\frac{\mu_z(P_0)\exp\{n^{-1/2}\sum_j t_j a^Z_j(z)\}}
{\sum_{z'}\mu_{z'}(P_0)\exp\{n^{-1/2}\sum_j t_j a^Z_j(z')\}} $,
$\theta_r^{(n,t)} =
\frac{\theta_r(P_0)\exp\{n^{-1/2}\sum_j t_j a^\Theta_j(r)\}}
{\sum_{r'\in\mathcal{T}_+(P_0)}\theta_{r'}(P_0)\exp\{n^{-1/2}\sum_j t_j a^\Theta_j(r')\}},$ where $r\in\mathcal{T}_+(P_0)$
$f_{r,d}^{(n,t)}(y) =
\frac{f_{r,d}(y)\exp\{n^{-1/2}\sum_j t_j a^F_{j,r,d}(y)\}}
{\int f_{r,d}(u)\exp\{n^{-1/2}\sum_j t_j a^F_{j,r,d}(u)\}\,du}.
$
Thus $P_{n,t}\in\PLATE$ for fixed $t$ and all large $n$. The exponential normalization gives the first-order expansions $\mu_z^{(n,t)}=\mu_z(P_0)\{1+n^{-1/2}\sum_j t_j a^Z_j(z)\}+O(n^{-1})$, $\theta_r^{(n,t)}=\theta_r(P_0)\{1+n^{-1/2}\sum_j t_j a^\Theta_j(r)\}+O(n^{-1})$, and similarly $f_{r,d}^{(n,t)}=f_{r,d}\{1+n^{-1/2}\sum_j t_j a^F_{j,r,d}\}+O(n^{-1})$ in quadratic mean. By the product-rule for QMD product experiments \citep{vanderVaart1998}, the induced observed-data model is QMD with score $\sum_j t_j h_j$ and information matrix $I_h=\E[hh']$. Therefore
$\int \biggl[\sqrt{dP_{n,t}/dP_0}-1-\frac{1}{2\sqrt n}\sum_j t_j h_j\biggr]^2dP_0=o(n^{-1}).$ The QMD-verified $m$-parameter structural submodel satisfies the LAN expansion \citep{vanderVaart1998}: $\log \frac{d(P_{n, t})^n}{dP_0^n}(O_1, \ldots, O_n) \;=\; t' \Delta_n - \tfrac12 t'\, I_h\, t + o_{P_0}(1),$
with $\Delta_n = n^{-1/2}\sum_{i=1}^n h(O_i) \xrightarrow{d} N(0, I_h)$ and $I_h = \E[h(O)h(O)']$. For each $h$ in the linear span of $\{h_1,\ldots,h_m\}$, write $t_h$ for its coordinates in this basis and set $P_{n,h}=P_{n,t_h}$.

By Lemma~\ref{lem:rt-pathwise}, $\beta^{\mathrm{RT}}_{\bomega}(P)$ is pathwise differentiable at $P_0$ with derivative $\dot\beta^{\mathrm{RT}}_{\bomega}(h; P_0) = \langle \eif_{\omega(\cdot)}, h \rangle_{L^2(P_0)}$. For any finite-dimensional $I \subset \TLATE$, the derivative restricted to $I$ is the linear functional $h \mapsto \langle \tilde\eif_I, h \rangle$ where $\tilde\eif_I$ is the orthogonal projection of $\eif_{\omega(\cdot)}$ onto $I$. The tangent-compatibility hypothesis gives $\eif_{\omega(\cdot)}\in\TLATEbar$; hence $\norm{\tilde\eif_I}_{L^2(P_0)}$ converges to $\norm{\eif_{\omega(\cdot)}}_{L^2(P_0)}$ as $I$ grows through finite-dimensional subspaces whose union is dense in $\TLATEbar$ (automatic, since $\TLATEbar$ is the closure of the structural tangent set). Theorem~\ref{thm:tc-primitive} and Corollary~\ref{cor:tc-finite-designs} give primitive sufficient conditions for the tangent-compatibility hypothesis itself.

Apply the local asymptotic minimax theorem \citep{vanderVaart1998} to the LAN sequence constructed above. Under the LAN-plus-pathwise-differentiability hypotheses verified above, the theorem yields, for the squared-error loss and any bounded generated finite-dimensional $I \subset \TLATE$,
$
\lim_{C\to\infty}\liminf_n \sup_{h \in I,\, \norm{h} \leq C} n\,\E\!\bigl[\bigl(\hat\beta_n - \beta^{\mathrm{RT}}_{\bomega}(P_{n,h})\bigr)^2\bigr] \;\geq\; \norm{\tilde\eif_I}_{L^2(P_0)}^2.
$
Taking the supremum over a nested family of finite-dimensional subspaces $I \subset \TLATE$ exhausting $\TLATEbar$ on the left, such that the per-subspace minimax bounds increase to the bound over the full tangent closure, and using the density of bounded generated subspaces and the convergence $\norm{\tilde\eif_I} \to \norm{\eif_{\omega(\cdot)}}$ established above on the right yields
$
\sup_{I \subset \TLATE,\, \dim I < \infty}
\lim_{C\to\infty}\liminf_n
\sup_{h \in I,\, \norm{h} \leq C}
n\,\E\!\bigl[\bigl(\hat\beta_n - \beta^{\mathrm{RT}}_{\bomega}(P_{n,h})\bigr)^2\bigr] 
\geq
\E[\eif_{\omega(\cdot)}(O)^2]
= V_{\omega(\cdot)}(P_0),
$
which is~\eqref{eq:lam-bound}. 

\smallskip\noindent\emph{(ii) Plug-in attainment.} The plug-in $\hat\beta^{\mathrm{RT}}_{\bomega} = \bomega(\hat P_n)' \widehat\bWald$ has the delta-method expansion $\hat\beta^{\mathrm{RT}}_{\bomega} - \beta^{\mathrm{RT}}_{\bomega}(P_0) = \bomega(\hat P_n)'\widehat\bWald - \bomega(P_0)'\bWald(P_0) = \bomega(P_0)'(\widehat\bWald - \bWald(P_0)) + (\bomega(\hat P_n) - \bomega(P_0))'\bWald(P_0) + R_n$, with cross-product remainder $R_n = (\bomega(\hat P_n) - \bomega(P_0))'(\widehat\bWald - \bWald(P_0)) = O_p(n^{-1})$ since both factors are $O_p(n^{-1/2})$ (the first by the assumed weight-map von Mises expansion, the second by the delta method for the sample covariance ratios under Assumption~\ref{ass:sampling}). The first term is $\sum_\ell \omega_\ell(P_0)(\widehat\Wald_\ell - \Wald_\ell(P_0)) = n^{-1}\sum_i \sum_\ell \omega_\ell(P_0) \phi_\ell(O_i; P_0) + o_p(n^{-1/2})$ by the same delta-method expansion, whose influence function is the Riesz representer $\phi_\ell$ of Lemma~\ref{lem:wald-pathwise}. The second term uses the same von Mises expansion:
$(\bomega(\hat P_n)-\bomega(P_0))'\bWald(P_0) = n^{-1}\sum_i \sum_\ell \Wald_\ell(P_0)\Riesz_{\omega,\ell}(O_i;P_0)+o_p(n^{-1/2}).$
For the fixed-weight map, $\Riesz_\omega \equiv \mathbf{0}$. The weight-map von Mises expansion is maintained as a hypothesis rather than derived from a particular construction. Combining,
$\sqrt n\,(\hat\beta^{\mathrm{RT}}_{\bomega} - \beta^{\mathrm{RT}}_{\bomega}(P_0)) = \frac{1}{\sqrt n}\sum_{i=1}^n \eif_{\omega(\cdot)}(O_i; P_0) + o_p(1),$
which establishes the asymptotic linearity of $\hat\beta^{\mathrm{RT}}_{\bomega}$ at $P_0$, the expansion underlying~\eqref{eq:plug-in-asylinear}. The CLT applied to $\eif_{\omega(\cdot)} \in L^2_0(P_0)$ gives $\sqrt n\,(\hat\beta^{\mathrm{RT}}_{\bomega} - \beta^{\mathrm{RT}}_{\bomega}(P_0)) \xrightarrow{d} N(0, V_{\omega(\cdot)}(P_0))$ under $P_0$. Along a contiguous submodel $\{P_{n,h}\}$, Le Cam's third lemma combined with the linearization above gives $\sqrt n\,(\hat\beta^{\mathrm{RT}}_{\bomega} - \beta^{\mathrm{RT}}_{\bomega}(P_0)) \xrightarrow{d, P_{n,h}} N(\inner{\eif_{\omega(\cdot)}}{h}_{L^2(P_0)}, V_{\omega(\cdot)}(P_0))$. The pathwise expansion gives $\sqrt n\,(\beta^{\mathrm{RT}}_{\bomega}(P_{n,h}) - \beta^{\mathrm{RT}}_{\bomega}(P_0)) \to \inner{\eif_{\omega(\cdot)}}{h}_{L^2(P_0)}$. Subtracting,
$\sqrt n\,(\hat\beta^{\mathrm{RT}}_{\bomega} - \beta^{\mathrm{RT}}_{\bomega}(P_{n,h})) \xrightarrow{d, P_{n,h}} N(0, V_{\omega(\cdot)}(P_0)),$
direction-free in $h$: regularity holds along every LATE-compatible contiguous submodel. The asymptotic linearization above has an $h$-free remainder: $\rho_n \equiv \sqrt n(\hat\beta^{\mathrm{RT}}_{\bomega}-\beta^{\mathrm{RT}}_{\bomega}(P_0))-n^{-1/2}\sum_i\eif_{\omega(\cdot)}(O_i;P_0)$ is a deterministic function of the $O_{P_0}(n^{-1/2})$ sample moments and the weight-map von~Mises remainder, not of the local direction. Fix a finite-dimensional $I\subset\TLATE$ and $C<\infty$, and take any sequence $\{h_n\}$ with $h_n\in I$ and $\|h_n\|_{L^2(P_0)}\le C$. By contiguity of $P_{n,h_n}^{n}$ to $P_0^{n}$ (Le Cam's first lemma), $\rho_n=o_{P_{n,h_n}}(1)$ along the sequence, and the functional Taylor remainder $q_n\equiv\beta^{\mathrm{RT}}_{\bomega}(P_{n,h_n})-\beta^{\mathrm{RT}}_{\bomega}(P_0)-n^{-1/2}\langle\eif_{\omega(\cdot)},h_n\rangle_{L^2(P_0)}$ satisfies $\sqrt n\,q_n\to0$: the Wald factor is smooth in $t$ along the exponential chart (moments of exponentially tilted laws under Assumption~\ref{ass:sampling}(d)), the weight factor is Hadamard differentiable along the chart by hypothesis, and Hadamard differentiability delivers the first-order expansion with remainder $o(t)$ uniformly over directions in the ball $\{h \in I : \|h\|_{L^2(P_0)} \leq C\}$, which is compact in the finite-dimensional subspace $I$; at $t = n^{-1/2}$ the remainder is therefore $o(n^{-1/2})$ along $\{h_n\}$. Hence the linearization above and Le Cam's third lemma give
$\sqrt n\,(\hat\beta^{\mathrm{RT}}_{\bomega}-\beta^{\mathrm{RT}}_{\bomega}(P_{n,h_n}))\;\xrightarrow{d,\,P_{n,h_n}}\;N(0,V_{\omega(\cdot)}(P_0)),$
a limit that does not depend on the sequence. The local-risk uniform-integrability condition (Assumption~\ref{ass:local-risk-ui}), whose supremum over $\{h\in I:\|h\|\le C\}$ sits inside the $M\to\infty$ limit, makes $\{n(\hat\beta^{\mathrm{RT}}_{\bomega}-\beta^{\mathrm{RT}}_{\bomega}(P_{n,h_n}))^2\}_n$ uniformly integrable along the sequence; with the distributional limit this gives $n\,\E[(\hat\beta^{\mathrm{RT}}_{\bomega}-\beta^{\mathrm{RT}}_{\bomega}(P_{n,h_n}))^2]\to V_{\omega(\cdot)}(P_0)$. Since this convergence holds along every such sequence with an identical limit, no sequence can keep the deviation above any $\varepsilon>0$, so $\sup_{\substack{h\in I,\ \|h\|_{L^2(P_0)}\le C}} \bigl| n\,\E[(\hat\beta^{\mathrm{RT}}_{\bomega}-\beta^{\mathrm{RT}}_{\bomega}(P_{n,h}))^2] - V_{\omega(\cdot)}(P_0) \bigr|\to0$.

\subsection{Proof of Proposition~\ref{prop:var-estimation}}\label{app:proof_var_estimation}

In the tangent-compatible case, the variance is $V_{\omega(\cdot)}(P_0) = \Var_{P_0}(\eif_{\omega(\cdot)}) = \E\bigl[(\eif_{\omega(\cdot)})^2\bigr]$ (zero mean). Expand and group by source:
$
V_{\omega(\cdot)}(P_0)
= \sum_{\ell, k} \omega_\ell\omega_k \E[\phi_\ell\phi_k]
 + 2\sum_{\ell,k}\omega_\ell\Wald_k\E[\phi_\ell\Riesz_{\omega,k}]
 + \sum_{\ell,k}\Wald_\ell\Wald_k
   \E[\Riesz_{\omega,\ell}\Riesz_{\omega,k}] 
= \bomega(P_0)'\bGamma^{\Wald}(P_0)\bomega(P_0)
 + 2\bomega(P_0)'\bUpsilon(\bomega; P_0)\bWald(P_0)
 + \bWald(P_0)'\bSigma_{\Riesz}(\bomega; P_0)\bWald(P_0),
$
with $\bUpsilon(\bomega; P_0)_{\ell k} = \Cov_{P_0}(\phi_\ell, \Riesz_{\omega,k})$ and $\bSigma_{\Riesz}(\bomega; P_0)_{\ell k} = \Cov_{P_0}(\Riesz_{\omega,\ell}, \Riesz_{\omega,k})$. The sample analog $\hat\phi_\ell(O_i)$ of the Wald influence function is the explicit polynomial $[(Y_i - \bar Y_n) - \widehat\Wald_\ell(D_i - \bar D_n)](Z_{\ell i} - p^Z_{\ell,n})/\hat\gamma_\ell$, with $\hat\phi_\ell(O_i) - \phi_\ell(O_i; P_0) = R_\ell(O_i; P_0, \hat P_n)$ a finite sum in which each term is the product of (i) between one and four of the estimation errors $(\bar Y_n - \mu_Y, \bar D_n - \mu_D, p^Z_{\ell,n} - p^Z_\ell, \widehat\Wald_\ell - \Wald_\ell(P_0), \hat\gamma_\ell - \gamma_\ell(P_0))$, (ii) a polynomial in $(Y_i, D_i, Z_{\ell i})$ with coefficients depending only on $P_0$, and (iii) either $1$ or $1/(\hat\gamma_\ell\,\gamma_\ell(P_0))$, the latter from $1/\hat\gamma_\ell - 1/\gamma_\ell = -(\hat\gamma_\ell - \gamma_\ell)/(\hat\gamma_\ell\gamma_\ell)$. Each estimation error is $O_p(n^{-1/2})$, and on the event $|\hat\gamma_\ell| \geq |\gamma_\ell(P_0)|/2$, which holds with probability tending to one because $\hat\gamma_\ell \to_p \gamma_\ell(P_0) \neq 0$ (Assumptions~\ref{ass:sampling}(e) and~\ref{ass:relevance}), the factor $1/(\hat\gamma_\ell\,\gamma_\ell(P_0))$ is bounded. Decompose
$\hat\Gamma^{\Wald}_{\ell k} \;=\; \frac{1}{n}\sum_i \phi_\ell(O_i; P_0)\phi_k(O_i; P_0) \;+\; \frac{1}{n}\sum_i \bigl[\hat\phi_\ell \hat\phi_k - \phi_\ell \phi_k\bigr](O_i),$
where the first term $\xrightarrow{p} \Gamma^{\Wald}_{\ell k}(P_0)$ by the law of large numbers under $\E[Y_i^2]<\infty$ (Assumption~\ref{ass:sampling}(d)). By the representation of $\hat\phi_\ell - \phi_\ell$ above, the second term is a finite sum of terms, each the product of estimation errors totalling $O_p(n^{-1/2})$, a factor bounded in probability on the event $|\hat\gamma_\ell| \geq |\gamma_\ell(P_0)|/2$, and a sample mean of a monomial in $(Y_i, D_i, Z_{\ell i})$ of degree at most two in $Y_i$, which is $O_p(1)$ by the law of large numbers under $\E[Y_i^2] < \infty$; hence the second term is $o_p(1)$. Hence $\hat\Gamma^{\Wald} \xrightarrow{p} \Gamma^{\Wald}(P_0)$. For the weight-drift terms, Assumption~\ref{ass:riesz-consistency} gives mean-square consistency of the estimated weight-map Riesz representers, and Cauchy--Schwarz plus the LLN for $\phi_\ell^2$ and $\Riesz_{\omega,\ell}^2$ yields $\hat\bUpsilon \xrightarrow{p} \bUpsilon(\bomega; P_0)$ and $\hat\bSigma_{\Riesz} \xrightarrow{p} \bSigma_{\Riesz}(\bomega; P_0)$. 

For the pointwise claim, fix a LATE-compatible contiguous submodel $\{P_{n,h}\}$ with score $h\in\TLATE$. Proposition~\ref{prop:eif-decomp} makes $\hat\beta^{\mathrm{RT}}_{\bomega}$ asymptotically linear with influence function $\eif_{\omega(\cdot)}$ and regular for $\beta^{\mathrm{RT}}_{\bomega}(\cdot)$; Le Cam's third lemma along $\{P_{n,h}\}$ then gives $\sqrt n\,(\hat\beta^{\mathrm{RT}}_{\bomega} - \beta^{\mathrm{RT}}_{\bomega}(P_{n,h})) \xrightarrow{d, P_{n,h}} N(0, V_\omega(P_0))$. The consistency above gives $\hat V_n \xrightarrow{p} V_\omega(P_0)$ under $P_0$, and contiguity transfers it to $P_{n,h}$; Slutsky's theorem makes the studentized pivot asymptotically standard normal, and the Wald interval covers $\beta^{\mathrm{RT}}_{\bomega}(P_{n,h})$ with asymptotic probability $1-\alpha$.

\subsection{Proof of Proposition~\ref{cor:witness}}\label{app:proof_witness_class}
Write $\bu=\bomega\oslash\bgamma$; then $\bgamma'\bu=\sum_\ell\omega_\ell=1$, and recall from Proposition~\ref{prop:weighted_iv} that $\lambda_\ell(\bW;P)=\gamma_\ell(\bW\bgamma)_\ell/(\bgamma'\bW\bgamma)$. The diagonal matrix $\bW^\dagger=\diag(\omega_\ell/\gamma_\ell^2)$ is positive definite for a locally interior target and satisfies $\bW^\dagger\bgamma=\bu$. Any member of the family has $\bW\bgamma=c\,\bW^\dagger\bgamma+\bK\bgamma=c\,\bu$, hence $\gamma_\ell(\bW\bgamma)_\ell=c\,\omega_\ell$ and $\bgamma'\bW\bgamma=c\,\bgamma'\bu=c>0$; hence $\blambda(\bW;P)=\bomega(P)$ at every $P$. A diagonal member $\diag(\bw)$ has $w_\ell\gamma_\ell=c\,\omega_\ell/\gamma_\ell$, that is $w_\ell=c\,\omega_\ell/\gamma_\ell^2$; thus $\bW^\dagger$ is the unique diagonal member up to scale.

\subsection{Proof of Proposition~\ref{thm:tangent-matching}}\label{app:proof_prop_tangent}

Write $\lambda^0_\ell=\lambda_{\bW,\ell}(P_0)$ and $\omega^0_\ell=\omega_\ell(P_0)$. Expanding both representers, the gap is
$\Delta_{\bW,\omega} = \sum_{\ell=1}^L(\lambda^0_\ell-\omega^0_\ell)\phi_\ell \;+\; \sum_{\ell=1}^L\Wald_\ell(P_0)\,[\psi^{\blambda}_{\bW,\ell}-\Riesz_{\omega,\ell}]$.

\emph{Linearization.} By Proposition~\ref{prop:weighted_iv}(b), $\hat\beta^{\mathrm{GMM}}_{\bW} = \blambda_{\bW}(\hat P_n)'\widehat\bWald$ with probability tending to one. The map $(\blambda,\bw)\mapsto\blambda'\bw$ is differentiable at $(\blambda_{\bW}(P_0),\bWald(P_0))$. Combining the expansion~\eqref{eq:W-asy-linear-main} with
$\widehat\Wald_\ell-\Wald_\ell(P_0) = n^{-1}\sum_i\phi_\ell(O_i;P_0)+o_p(n^{-1/2})$
(the delta method for the sample covariance ratios under Assumption~\ref{ass:sampling}, with influence function the Riesz representer $\phi_\ell$ of Lemma~\ref{lem:wald-pathwise} in the Supplementary Appendix), a first-order product expansion yields
$\sqrt n\,\bigl(\hat\beta^{\mathrm{GMM}}_{\bW} - \beta^{\mathrm{GMM}}_{\bW}(P_0)\bigr) \;=\; n^{-1/2}\sum_{i=1}^n \phi^{\bW}(O_i) + o_p(1)$,
with cross-product remainders $O_p(n^{-1})$ under Assumption~\ref{ass:sampling}'s $(2+\delta)$-th moments. Every component of $\phi^{\bW}$ is in $L^2_0(P_0)$.

\emph{Sufficiency.} Proposition~\ref{prop:eif-decomp} gives $\eif_{\omega(\cdot)}=\sum_{\ell}\omega^0_\ell\phi_\ell +\sum_{\ell}\Wald_\ell(P_0)\Riesz_{\omega,\ell}$; hence $\phi^{\bW}=\eif_{\omega(\cdot)}+\Delta_{\bW,\omega}$. Condition (i) aligns the centers. Along any LATE-compatible contiguous submodel $\{P_{n,h}\}$ with score $h \in \TLATE$, Le Cam's third lemma applied to the above first-order product expansion gives
$\sqrt n(\hat\beta^{\mathrm{GMM}}_{\bW} - \beta^{\mathrm{RT}}_{\bomega}(P_0)) \xrightarrow{d, P_{n,h}} N(\inner{\phi^{\bW}}{h}_{L^2(P_0)}, \Var_{P_0}(\phi^{\bW}))$, and the pathwise expansion of $\beta^{\mathrm{RT}}_{\bomega}$ gives $\sqrt n(\beta^{\mathrm{RT}}_{\bomega}(P_{n,h}) - \beta^{\mathrm{RT}}_{\bomega}(P_0)) \to \inner{\eif_{\omega(\cdot)}}{h}_{L^2(P_0)}$. Condition (ii) gives $\inner{\Delta_{\bW,\omega}}{h}_{L^2(P_0)}=0$ for every $h\in\TLATE$; subtracting yields an $h$-free limit:
$\sqrt n(\hat\beta^{\mathrm{GMM}}_{\bW} - \beta^{\mathrm{RT}}_{\bomega}(P_{n,h})) \xrightarrow{d, P_{n,h}} N(0,\Var_{P_0}(\phi^{\bW}))$. Thus $\hat\beta^{\mathrm{GMM}}_{\bW}$ is regular for $\beta^{\mathrm{RT}}_{\bomega}(\cdot)$.

\emph{Necessity.} Regularity at $P_0$ implies consistency along the baseline; hence $\beta^{\mathrm{GMM}}_{\bW}(P_0) = \beta^{\mathrm{RT}}_{\bomega}(P_0)$, equivalently $[\blambda_{\bW}(P_0) - \bomega(P_0)]'\bWald(P_0) \;=\; 0$,
which is condition (i). Along any LATE-compatible contiguous submodel $\{P_{n,h}\}$ with $h \in \TLATE$, Le Cam's third lemma applied to the above first-order product expansion and the pathwise expansion of $\beta^{\mathrm{RT}}_{\bomega}$ give, using condition (i),
$\sqrt n(\hat\beta^{\mathrm{GMM}}_{\bW} - \beta^{\mathrm{RT}}_{\bomega}(P_{n,h})) \;\xrightarrow{d, P_{n,h}}\; N\bigl(\langle\phi^{\bW} - \eif_{\omega(\cdot)}, h\rangle_{L^2(P_0)},\; \Var_{P_0}(\phi^{\bW})\bigr)$.
Regularity requires an $h$-free limit distribution; hence $\langle\Delta_{\bW,\omega}, h\rangle_{L^2(P_0)} = 0$ for every $h \in \TLATE$, and by continuity of the inner product this extends to $\overline{\mathcal{G}}$: $\Pi_{\overline{\mathcal{G}}}\Delta_{\bW,\omega}=0$.

\subsection{Proof of Lemma~\ref{lem:active-set-stability}}\label{app:proof_lem_active}

By Assumption~\ref{ass:projection-dictionary}, the entries of $\bA_{\mathrm{spec}}(P)$, $\bpsi^\star(P)$, and $\bkappa(P)$ on $\mathcal{T}_+(P_0)$ are continuous at $P_0$; write $A(P)$ for this type-restricted matrix (with entries $\alpha_{\ell,t}$) inside this proof. Hence
$G(P)=A(P)\diag(\bkappa(P))A(P)',\qquad c(P)=A(P)\diag(\bkappa(P))\bpsi^\star(P)$
are continuous at $P_0$ on the locally fixed positive-support set $\mathcal{T}_+(P_0)$. Let $\mathcal{K}_{I_0}=\{v\in\mathbb{R}^{|I_0|}:\ind_{I_0}'v=0\}$. If $|I_0|>1$, Assumption~\ref{ass:strict-complementarity} gives $b_0=\inf_{\substack{v\in\mathcal{K}_{I_0}\\ \|v\|=1}}v'G_{I_0I_0}(P_0)v>0.$
By continuity in finite dimension, a neighborhood $\mathcal{U}_0^{(1)}$ exists on which $v'G_{I_0I_0}(P)v\geq b_0/2$ for every unit $v\in\mathcal{K}_{I_0}$. If $|I_0|=1$, set $\mathcal{U}_0^{(1)}$ to any continuity neighborhood. Since $G(P)$ is a weighted Gram matrix, this tangent-positive condition is equivalent to nonsingularity of the bordered KKT matrix $M(P)=\left[\begin{smallmatrix}G_{I_0I_0}(P)&\ind_{I_0}\\ \ind_{I_0}'&0\end{smallmatrix}\right]$, and for $|I_0|=1$ this matrix is nonsingular directly. Thus $M(P)$ is invertible on $\mathcal{U}_0^{(1)}$. By the implicit function theorem, the restricted KKT solution $(\omega^\dagger_{I_0}(P), \nu^\dagger(P))$ is continuous in $P$ on $\mathcal{U}_0^{(1)}$. Define the inactive multipliers by the stationarity residual $m^\dagger_{I_0^c}(P)=G_{I_0^c,I_0}(P)\omega^\dagger_{I_0}(P)+\nu^\dagger(P)\ind_{I_0^c}-c_{I_0^c}(P)$; these are continuous as well. Let
$\epsilon_0 = \min\bigl\{\min_{\ell \in I_0}\omega^\dagger_\ell(P_0),\ \min_{\ell \in I_0^c}m^\dagger_\ell(P_0)\bigr\}$, with the convention $\min_{\ell\in\emptyset}=+\infty$. By Assumption~\ref{ass:strict-complementarity}, $\epsilon_0>0$. By continuity, there is a neighborhood $\mathcal{U}_0 \subseteq \mathcal{U}_0^{(1)}$ on which $\omega^\dagger_\ell(P) > \epsilon_0/2$ for $\ell \in I_0$ and, when $I_0^c\neq\emptyset$, $m^\dagger_\ell(P) > \epsilon_0/2$ for $\ell \in I_0^c$. Together with $\omega^\dagger_{I_0^c}(P)=0$ and $m^\dagger_{I_0}(P)=0$, these signs satisfy the full KKT system for the convex quadratic program; hence the restricted solution is the full projection solution. On $\mathcal{U}_0$ the active set is exactly $I_0$, strict complementarity holds, and the active-face bordered KKT matrix remains nonsingular.

\subsection{Proof of Proposition~\ref{prop:lam-projected}}\label{app:proof_prop_pathwise_omega}

Let $\bA_{\mathrm{spec}}(P)$, $\bpsi^\star_{\mathrm{spec}}(P)$, and $\bkappa_{\mathrm{spec}}(P)$ denote the corresponding restrictions; write $A(P)$ for $\bA_{\mathrm{spec}}(P)$ and suppress the analogous subscripts. Expanding the projection objective yields $f(\omega; P) = \tfrac12 \omega' G(P) \omega - c(P)' \omega + \tfrac12 \bpsi^\star(P)'\diag(\bkappa(P))\bpsi^\star(P)$ with $G(P) = A(P)\diag(\bkappa(P))A(P)'$ and $c(P) = A(P)\diag(\bkappa(P))\bpsi^\star(P)$, all on $\mathcal{T}_+(P_0)$. The Lagrangian with simplex multiplier $\nu$ and non-negativity multipliers $m \geq 0$ yields the KKT conditions
$
G(P) \omega^\dagger - c(P) + \nu^\dagger \ind_L - m^\dagger = 0$, 
$\ind_L' \omega^\dagger = 1$, $ \omega^\dagger \geq 0$, $m^\dagger \geq 0$, and $ m^\dagger_\ell\, \omega^\dagger_\ell = 0. $ With active set $I = I(\psi^\star; P)$, complementary slackness gives $m^\dagger_\ell = 0$ on $I$ and $\omega^\dagger_\ell = 0$ on $I^c$, so $G_{II}(P)\, \omega^\dagger_I - c_I(P) + \nu^\dagger \ind_I = 0$, and$ \ind_I' \omega^\dagger_I = 1.$
Under Assumption~\ref{ass:strict-complementarity} and Lemma~\ref{lem:active-set-stability}, the bordered matrix $M(P) = \begin{bmatrix} G_{I_0 I_0}(P) & \ind_{I_0} \\ \ind_{I_0}' & 0 \end{bmatrix} $
is invertible on $\mathcal{U}_0$; write $N(P)=M(P)^{-1}$. The KKT condition solves to $(\omega^\dagger_{I_0}(P), \nu^\dagger(P))' = N(P) (c_{I_0}(P), 1)'$ on $\mathcal{U}_0$. The pathwise derivative takes the form $\begin{bmatrix}
\dot\omega^\dagger_{I_0}(h; P_0)\\
\dot\nu^\dagger(h; P_0)
\end{bmatrix}
=
M(P_0)^{-1}
\begin{bmatrix}
\dot c_{I_0}(h; P_0)
- \dot G_{I_0 I_0}(h; P_0)\, \omega^\dagger_{I_0}(P_0)\\
0
\end{bmatrix},
$
  with $\dot\omega^\dagger_\ell(h; P_0) = 0$ for $\ell \in I_0^c$. Write
  $\dot\alpha_{\ell,t}(h;P_0)=\inner{\xi_{\ell,t}(\cdot;P_0)}{h}_{L^2(P_0)}$,
  $\dot\psi^\star_t(h;P_0)=\inner{\xi^\star_t(\cdot;P_0)}{h}_{L^2(P_0)}$, and
  $\dot\kappa_t(h;P_0)=\inner{\zeta_t(\cdot;P_0)}{h}_{L^2(P_0)}$. Then $  \dot G_{\ell,k}
  =
  \sum_{t\in\mathcal{T}_+(P_0)}
  \kappa_t(P_0)
  \bigl[
  \alpha_{k,t}(P_0)\dot\alpha_{\ell,t}(h;P_0)
  +\alpha_{\ell,t}(P_0)\dot\alpha_{k,t}(h;P_0)
  \bigr]
+
  \sum_{t\in\mathcal{T}_+(P_0)}
  \alpha_{\ell,t}(P_0)\alpha_{k,t}(P_0)\dot\kappa_t(h;P_0),
  \dot c_\ell
  =
  \sum_{t\in\mathcal{T}_+(P_0)}
  \kappa_t(P_0)\psi^\star_t(P_0)\dot\alpha_{\ell,t}(h;P_0)
  +\sum_{t\in\mathcal{T}_+(P_0)}
  \kappa_t(P_0)\alpha_{\ell,t}(P_0)\dot\psi^\star_t(h;P_0)
+
  \sum_{t\in\mathcal{T}_+(P_0)}
  \alpha_{\ell,t}(P_0)\psi^\star_t(P_0)\dot\kappa_t(h;P_0).
$
  The representer $\xi_{\ell,t}$ is supplied by Assumption~\ref{ass:projection-dictionary} for threshold types, while $\xi^\star_t$ and $\zeta_t$ represent the derivatives of the target weight and metric weight. For $P$-independent metric weights ($\bkappa \equiv \text{const}$), $\zeta_t\equiv0$; for $P$-independent targets ($\bpsi^\star \equiv \text{const}$), $\xi^\star_t\equiv0$. All three Riesz representers $\xi_{k,t}, \xi^\star_t, \zeta_t$ are computed in Lemma~\ref{lem:A_pathwise} and Appendix~\ref{app:proof_canonical_targets} for the canonical threshold specifications. The Riesz representer $\Riesz_{\omega^\dagger, \ell}$ of $\dot\omega^\dagger_\ell$ is zero for inactive instruments and, for $\ell\in I_0$, decomposes into three drift channels ($\bA$-drift, $\bpsi^\star$-drift, $\bkappa$-drift) that propagate through the bordered KKT system. With $N = M(P_0)^{-1}$, it equals $
  \Riesz_{\omega^\dagger, \ell}(\cdot; \bpsi^\star, \bkappa, P_0) \;=\; \sum_{k \in I_0}\sum_{t\in\mathcal{T}_+(P_0)} \bigl[\mathcal{X}^{(\ell)}_{k, t}\, \xi_{k, t}(\cdot; P_0) + \mathcal{E}^{(\ell)}_{k, t}\, \xi^\star_t(\cdot; P_0) + \mathcal{Z}^{(\ell)}_{k, t}\, \zeta_t(\cdot; P_0)\bigr],
$
  with coefficients $
  \mathcal{X}^{(\ell)}_{k,t}
  =
  \kappa_t(P_0)
  \bigl[
  N_{\ell,k}\psi^\star_t(P_0)
  -N_{\ell,k}\sum_{j\in I_0}A_{j,t}(P_0)\omega^\dagger_j(P_0)
  -\sum_{j\in I_0}N_{\ell,j}A_{j,t}(P_0)\omega^\dagger_k(P_0)
  \bigr],
  \mathcal{E}^{(\ell)}_{k,t}
  =\kappa_t(P_0)N_{\ell,k}\alpha_{k,t}(P_0),
  \mathcal{Z}^{(\ell)}_{k,t}
  =
  N_{\ell,k}\alpha_{k,t}(P_0)
  \bigl[
  \psi^\star_t(P_0)
  -\sum_{j\in I_0}A_{j,t}(P_0)\omega^\dagger_j(P_0)
  \bigr].
$
  The projected weight $P\mapsto\bomega^\dagger(\bpsi^\star;P)$ is the composition of the projection primitives $(\bA_{\mathrm{spec}},\bpsi^\star,\bkappa)$, the active-face Gram and linear term $(\bG_{I_0I_0},\bc_{I_0})$ they form, and the bordered-KKT inverse $N(P)=M(P)^{-1}$ of Lemma~\ref{lem:active-set-stability}; each link is continuously differentiable at $P_0$, the last by Assumption~\ref{ass:strict-complementarity}. Assumption~\ref{ass:projection-dictionary} supplies the von Mises expansions of the primitives $\alpha_{\ell,t}(P)$, $\psi^\star_t(P)$, and $\kappa_t(P)$ with Riesz representers $\xi_{\ell,t}$, $\xi^\star_t$, and $\zeta_t$; chaining them through this composition gives, for each $\ell$,
$\omega^\dagger_\ell(\bpsi^\star;\hat P_n)-\omega^\dagger_\ell(\bpsi^\star;P_0) = n^{-1}\sum_{i=1}^n \Riesz_{\omega^\dagger,\ell}(O_i;\bpsi^\star,\bkappa,P_0) +o_p(n^{-1/2}),$
   Substituting these representers into Proposition~\ref{prop:eif-decomp} gives the influence function for the projected target in the form~\eqref{eq:eif-general}. Assumption~\ref{ass:projection-tangent} makes this representer an element of the LATE tangent closure, and Assumptions~\ref{ass:joint_indep}, \ref{ass:monotonicity}, \ref{ass:realized_exclusion}, \ref{ass:relevance}, \ref{ass:sampling}, and~\ref{ass:late-regularity} supply the LATE-model conditions in Proposition~\ref{prop:eif-decomp}. Hence the projected target satisfies the differentiability (Lemma~\ref{lem:rt-pathwise}), von Mises, tangent-compatibility, and that finite-moment hypotheses, yielding asymptotic linearity and, under Assumption~\ref{ass:local-risk-ui}, LAM attainment. Finally,  for inactive coordinates, the general variance decomposition applies after restricting to $I_0$.

\subsection{Proof of Proposition~\ref{prop:id-gap-distfree}}\label{app:proof_prop_id_gap_distfree}

The Wald decomposition gives $\sum_{t\in\mathcal J(P_0)}\alpha_{\ell,t}(P)\LATE_t(P)=\Wald_\ell(P)$ for every $\ell$, and the maintained coordinate bounds place $(\LATE_t(P))_{t\in\mathcal J(P_0)}$ in $B^{\mathrm{LP}}(\underline M,\overline M)$. Its objective value is $\sum_t r_t\LATE_t(P)=\Delta(\bpsi^\star;P)$, so the infimum and supremum in~\eqref{eq:gap-bound-lp} bracket the gap. 

\subsection{Proof of Proposition~\ref{prop:id-gap-mte}}\label{app:proof_prop_id_gap_mte}

Under Assumptions~\ref{ass:latent_index} and~\ref{ass:local-threshold-support}, $\theta_{t_k}(P)=p_{k+1}(P)-p_k(P)>0$ and
$\LATE_{t_k}(P)=\theta_{t_k}(P)^{-1}\int_{p_k(P)}^{p_{k+1}(P)}\MTE(u;P)\,du$.
W.l.o.g, if $\MTE(\cdot;P)$ is weakly increasing, then for adjacent intervals $\MTE(u;P)\leq\MTE(v;P)$ whenever $u\in(p_k,p_{k+1}]$ and $v\in(p_{k+1},p_{k+2}]$. Averaging over the two intervals gives $\LATE_{t_k}(P)\leq\LATE_{t_{k+1}}(P)$. Conversely, given $b_{t_0}\leq\cdots\leq b_{t_K}$, define $m(u)=b_{t_k}$ on $(p_k,p_{k+1}]$, extend it by $b_{t_0}$ below $p_0$ and by $b_{t_K}$ above $p_{K+1}$, and set $\MTE(u;P)=m(u)$. This MTE is weakly increasing and has adjacent-interval averages $b_{t_k}$. The decreasing case is analogous.

\section{Diagonal Specialization and the Heterogeneity Penalty}\label{app:diagonal}

\subsection{Diagonality of the residual second-moment matrix}

\begin{lemma}[Diagonality of $\bOmega$ at any anchor]\label{lem:omega-diagonal}
Under Assumptions~\ref{ass:joint_indep}, \ref{ass:monotonicity}, \ref{ass:sampling}(b),(c),(d), \ref{ass:diag_instruments}, and~\ref{ass:diag_nonoverlap}, the residual second-moment matrix is diagonal at every anchor: $[\bOmega(\beta;P)]_{\ell k} = 0$ for every $\beta \in \mathbb{R}$ and every $\ell \neq k$.
\end{lemma}

\begin{proof}
Write the demeaned residual as $\tilde V_i(\beta) = Y_i(0) + (\tau_i - \beta)\,t(\bZ_i) - \E[Y_i - \beta D_i]$ with $\tau_i = Y_i(1) - Y_i(0)$ and $t = D_i(\cdot)$, by Assumption~\ref{ass:sampling}(b),(c). Fix a type $t \in \mathcal{T}_+(P)$ and a coordinate $m$ with $t \notin \mathcal{C}_m$: no $z_{-m}$ has $t(1, z_{-m}) > t(0, z_{-m})$, and Assumption~\ref{ass:monotonicity} rules out $t(1, z_{-m}) < t(0, z_{-m})$; hence $t$ is constant in coordinate $m$. Under Assumption~\ref{ass:diag_nonoverlap} each positive-mass type responds to at most one instrument; thus $t(\bZ) = h_t(Z_{j(t)})$ for a single coordinate $j(t)$, with $h_t$ constant for always- and never-takers. For $\ell \neq k$, decomposing by type, $[\bOmega(\beta)]_{\ell k} = \sum_{t\in\mathcal{T}_+(P)} \theta_t\, \E[\tilde V^2 (Z_\ell - p^Z_\ell)(Z_k - p^Z_k) \mid D_i(\cdot) = t]$. Conditional on the type, Assumption~\ref{ass:joint_indep} gives $(Y(0), \tau) \perp \bZ$ with the law of $\bZ$ unchanged, and Assumption~\ref{ass:diag_instruments} gives mutual independence of $Z_1, \ldots, Z_L$. Integrating over $\bZ$ at fixed $(Y(0), \tau)$, the integrand is a function of $Z_{j(t)}$ times $(Z_\ell - p^Z_\ell)(Z_k - p^Z_k)$; since $j(t)$ equals at most one of $\ell \neq k$, at least one centered factor is independent of the rest of the product and has mean zero; therefore the expectation vanishes. Summing over types gives the claim.
\end{proof}

\subsection{Variance decomposition and the heterogeneity penalty derivative}

The normalized residual second moment is $\sigma^2_{\epsilon,\ell} = [\bOmega]_{\ell\ell}/[p^Z_\ell(1-p^Z_\ell)]$. Under non-overlapping compliance, it decomposes as $\sigma^2_{\epsilon,\ell} = (1-p^Z_\ell)\,\pi_\ell\,\sigma^2_{\tau,\ell} + R_\ell$, where $R_\ell = \sigma^2_{\epsilon,\ell} - (1-p^Z_\ell)\pi_\ell\sigma^2_{\tau,\ell}$ collects baseline outcome variance, non-complier residual variance, LATE-deviation terms, and the within-complier baseline-effect covariance $\Cov(Y_i(0),\tau_i\mid D_i(\cdot)\in\mathcal{C}_\ell(P))$. Holding $\LATE_\ell$ and the evaluation point $\beta^*$ fixed (hence $\LATE_\ell - \beta^*$), together with the baseline-effect covariance, $R_\ell$ does not depend on $\sigma^2_{\tau,\ell}$ (this follows from the population decomposition of $\sigma^2_{\epsilon,\ell}$ into own-instrument-complier, other-instrument-complier, always-taker, and never-taker contributions, where $\sigma^2_{\tau,\ell}$ enters only through the own-complier term with coefficient $(1-p^Z_\ell)\pi_\ell$), and the heterogeneity penalty derivative follows from this decomposition alone.

Holding the residual-variance evaluation point $\beta^*$ and all other residual-variance components fixed, the partial heterogeneity-penalty derivative is $\partial \lambda^{\EGMM}_\ell / \partial \sigma^2_{\tau,\ell} = -\lambda^{\EGMM}_\ell (1 - \lambda^{\EGMM}_\ell)(1-p^Z_\ell)\pi_\ell/\sigma^2_{\epsilon,\ell}$, which is strictly negative when $L\ge2$ and all diagonal weights are positive. This is a partial derivative of the diagonal weight formula, not a total derivative of the iterated-GMM fixed point when changing $\sigma^2_{\tau,\ell}$ also changes $\beta^{\EGMM}$ or off-equilibrium residual components.

\subsection{Heterogeneity penalty under general \texorpdfstring{$\bOmega$}{Omega}}\label{app:general_omega}
\begin{proposition}[Heterogeneity penalty under general $\bOmega$]\label{prop:general_penalty}
Under Assumption~\ref{ass:relevance}, let $\bOmega \succ 0$ with $\lambda^{\EGMM}_\ell \neq 0$. Holding all other entries of $\bOmega$ fixed,
$\mathrm{sign}\!\left(\frac{\partial \lambda^{\EGMM}_\ell}{\partial [\bOmega]_{\ell\ell}}\right) = -\mathrm{sign}(\lambda^{\EGMM}_\ell)$; $|\lambda^{\EGMM}_\ell|$ strictly decreases in $[\bOmega]_{\ell\ell}$ for $L \geq 2$.
\end{proposition}

\begin{proof}
Using $\partial [\bOmega^{-1}]_{jk} / \partial [\bOmega]_{\ell\ell} = -[\bOmega^{-1}]_{j\ell}[\bOmega^{-1}]_{\ell k}$ and writing $a = \bOmega^{-1}\bgamma$, $S = \bgamma'a > 0$:
$\frac{\partial \lambda^{\EGMM}_\ell}{\partial [\bOmega]_{\ell\ell}} = \frac{\gamma_\ell a_\ell}{S^2}\bigl(a_\ell^2 - S \cdot [\bOmega^{-1}]_{\ell\ell}\bigr).$
Cauchy--Schwarz gives $a_\ell^2 = (\mathbf{e}_\ell'\bOmega^{-1}\bgamma)^2 \leq [\bOmega^{-1}]_{\ell\ell}\cdot S$, strict for $L \geq 2$; hence $a_\ell^2 - S\cdot [\bOmega^{-1}]_{\ell\ell} < 0$. The sign of the derivative is $-\mathrm{sign}(\gamma_\ell a_\ell) = -\mathrm{sign}(\lambda^{\EGMM}_\ell)$, using $\gamma_\ell > 0$ by Assumption~\ref{ass:relevance} and $\lambda^{\EGMM}_\ell = \gamma_\ell a_\ell / S$.
\end{proof}

\subsection{Diagonal Wald influence-function covariance}

\begin{lemma}[Diagonal Wald covariance under non-overlap]\label{lem:rt-diagonal-wald-cov}
Under Assumptions~\ref{ass:joint_indep}, \ref{ass:monotonicity}, \ref{ass:realized_exclusion}, \ref{ass:relevance}, and~\ref{ass:sampling}, and the diagonal specialization (Assumptions~\ref{ass:diag_instruments}--\ref{ass:diag_nonoverlap}), the Wald influence-function covariance matrix is diagonal: $\Gamma^{\Wald}_{\ell k}(P_0)=0$ for every $\ell\neq k$.
\end{lemma}

\begin{proof}
For $\ell \neq k$, $\Gamma^{\Wald}_{\ell k}$ involves $\E[\tilde\epsilon_{i,\ell}\,\tilde\epsilon_{i,k}\,(Z_{\ell i} - p^Z_\ell)(Z_{ki} - p^Z_k)]$ with the instrument-specific Wald residuals $\tilde{\epsilon}_{i,\ell} = Y_i - \Wald_\ell D_i - c_\ell$. Since, conditional on any type $t$, the residual product $\tilde\epsilon_{i,\ell}\tilde\epsilon_{i,k}$ depends on the instruments only through the type's single responding coordinate $Z_{j(t)i}$ (Assumption~\ref{ass:diag_nonoverlap}), this is the cross-instrument moment treated by the group-partition argument of Lemma~\ref{lem:omega-diagonal}, with $\tilde\epsilon_{i,\ell}\tilde\epsilon_{i,k}$ in place of the common squared residual; that argument zeroes the moment under Assumptions~\ref{ass:diag_instruments}--\ref{ass:diag_nonoverlap}.  
\end{proof}
\section{Latent Index, Threshold Types, and the PRTE Projection}\label{app:latent_prte}

\subsection{Threshold-type support and endpoint stability}

\begin{lemma}[Ordered threshold support and endpoint stability]\label{lem:threshold-types-equiv}
For a threshold-indexed specification over intervals of the realized propensity range, with $z^{(k)}$ a support cell attaining the interior endpoint $p_k(P_0)=p(z^{(k)})(P_0)$, and under Assumptions~\ref{ass:latent_index} and~\ref{ass:monotonicity}, the following are equivalent on a neighborhood of $P_0$: (a) every interval entering the specification has positive length, and each interior endpoint is attained by a unique cell $z^{(k)}$; (b) the propensity endpoints are strictly separated, with locally constant ordering and cell labels.
\end{lemma}

\begin{proof}
Under Assumption~\ref{ass:monotonicity}, the threshold functions $t_k(z) = \ind\{p(z) \geq p_{k+1}\}$ are monotone-admissible because $p(\cdot)$ is coordinate-wise non-decreasing. For the complete ordered-threshold partition, positive interval lengths mean $\theta_{t_k}=p_{k+1}-p_k>0$ for $k=0,\ldots,K$ with $p_0=0$ and $p_{K+1}=1$, hence $0<p_1<\cdots<p_K<1$. The same strict-inequality argument applies to any displayed subcollection of threshold intervals. Because each endpoint is a smooth cell-propensity functional on positive-mass cells, strict inequalities at $P_0$ persist on a neighborhood; conversely, if an adjacent endpoint equality or unresolved label tie occurs at $P_0$, no locally constant ordered labeling exists for that interval.  
\end{proof}

\subsection{The PRTE projection equivalence}

\begin{lemma}[Full-domain $L^2$-PRTE equals weighted-$\ell^2$ projection under $\kappa = \theta^{-1}$]\label{lem:prte-equivalence}
Let $w^P(\cdot) = \psi^{\mathrm{PRTE}}(\cdot\,;P)$ denote the policy-shift density of~\eqref{eq:prte_integral}; then $\int_{\mathcal{R}_{t_k}} w^P(u)\,du = \psi^{\star,\mathrm{PRTE}}_{t_k}(P)$ for every $k$. Under Assumptions~\ref{ass:joint_indep}--\ref{ass:relevance}, Assumption~\ref{ass:latent_index}, and Assumption~\ref{ass:local-threshold-support}, for any $\bomega \in \Delta^{L-1}$,
$
\int_0^1 \bigl[\psi^{\MTE}(u; \bomega; P) - w^P(u)\bigr]^2\,du
=
\sum_{k=0}^{K} \theta_{t_k}(P)^{-1}
\bigl[
\textstyle\sum_\ell \omega_\ell\, \alpha_{\ell,t_k}(P)
- \psi^{\star,\mathrm{PRTE}}_{t_k}(P)
\bigr]^2 + R(P),
$
where $R(P) = \sum_{k=0}^{K}\int_{\mathcal{R}_{t_k}} (w^P(u) - \bar w^P_{t_k}(P))^2\,du$ measures the within-interval variation of $w^P$ around its interval-mean $\bar w^P_{t_k}(P) = \theta_{t_k}(P)^{-1}\psi^{\star,\mathrm{PRTE}}_{t_k}(P)$.
\end{lemma}

\begin{proof}
Under Assumptions~\ref{ass:latent_index} and~\ref{ass:local-threshold-support}, the ordered-threshold partition $\{\mathcal{R}_{t_k}(P)=(p_k(P),p_{k+1}(P)]:k=0,\ldots,K\}$ covers $[0,1]$ up to endpoints with positive lengths $\theta_{t_k}(P)=p_{k+1}(P)-p_k(P)$ (Lemma~\ref{lem:type_interval}), and each MTE-weight function is the step function $\alpha^{\MTE}_\ell(u;P)=\sum_{k=0}^{K}\theta_{t_k}(P)^{-1}\alpha_{\ell,t_k}(P)\,\ind\{u\in\mathcal{R}_{t_k}\}$ with $\int_{\mathcal{R}_{t_k}}\alpha^{\MTE}_\ell(u;P)\,du=\alpha_{\ell,t_k}(P)$ (the step-function representation of~\eqref{eq:wald_mte} under latent index). Hence $\psi^{\MTE}(u;\bomega;P)=\sum_\ell\omega_\ell\,\alpha^{\MTE}_\ell(u;P)$ is constant on each $\mathcal{R}_{t_k}$, equal to $\theta_{t_k}(P)^{-1}\psi_{t_k}(\bomega;P)$ with $\psi_{t_k}(\bomega;P)=\sum_\ell\omega_\ell\,\alpha_{\ell,t_k}(P)$. Let $\bar w^P_{t_k}(P)=\theta_{t_k}(P)^{-1}\int_{\mathcal{R}_{t_k}}w^P(u)\,du=\theta_{t_k}(P)^{-1}\psi^{\star,\mathrm{PRTE}}_{t_k}(P)$ be the interval mean of $w^P$. On $\mathcal{R}_{t_k}$ the residual $w^P-\bar w^P_{t_k}$ integrates to zero, hence is $L^2(\mathcal{R}_{t_k})$-orthogonal to the constant $\psi^{\MTE}-\bar w^P_{t_k}$; writing $\psi^{\MTE}-w^P=(\psi^{\MTE}-\bar w^P_{t_k})-(w^P-\bar w^P_{t_k})$ and expanding,
$\int_{\mathcal{R}_{t_k}}[\psi^{\MTE}(u;\bomega;P)-w^P(u)]^2\,du =\theta_{t_k}(P)\bigl(\theta_{t_k}^{-1}\psi_{t_k}(\bomega)-\theta_{t_k}^{-1}\psi^{\star,\mathrm{PRTE}}_{t_k}\bigr)^2 +\int_{\mathcal{R}_{t_k}}(w^P-\bar w^P_{t_k})^2\,du,$
the first term being $\theta_{t_k}(P)^{-1}[\psi_{t_k}(\bomega)-\psi^{\star,\mathrm{PRTE}}_{t_k}]^2$. Summing over $k=0,\ldots,K$ yields the stated decomposition.
\end{proof}

The PRTE projection solves $\bomega^{\dagger}(\bpsi^{\star,\mathrm{PRTE}};P) \;=\; \argmin_{\bomega \in \Delta^{L-1}} \int_0^1 \bigl[\psi^{\MTE}(u; \bomega; P) - w^P(u)\bigr]^2\,du.$
By Lemma~\ref{lem:prte-equivalence}, this $L^2$ objective equals the discrete inverse-mass projection objective plus the constant $R(P)$: $ \int_0^1
  [\psi^{\MTE}(u;\bomega;P)-w^P(u)]^2du
  =
  \omega'\bG^{\mathrm{PRTE}}(P)\omega
  -2\omega'\bc^{\mathrm{PRTE}}(P)
  + C^{\mathrm{PRTE}}(P),
$
where
$
  \bc^{\mathrm{PRTE}}(P)
  =
  \bA_{\mathrm{spec}}(P)\,
  \diag\{\theta_{t_k}(P)^{-1}: k=0,\ldots,K\}\,
  \bpsi^{\star,\mathrm{PRTE}}(P),
$
and $C^{\mathrm{PRTE}}(P)$ collects terms that do not depend on $\omega$, including the within-interval remainder $R(P)$; $\bA_{\mathrm{spec}}(P)$ is the type-restricted compliance-weight matrix of Appendix~\ref{app:regularity}. Assumption~\ref{ass:local-threshold-support} gives the full threshold partition of $[0,1]$; hence the displayed integral is the full $L^2[0,1]$ objective. If $\omega_1,\omega_2\in\Delta^{L-1}$ are distinct, then $v=\omega_1-\omega_2$ satisfies $\ind_L'v=0$ and $v\neq0$. Assumption~\ref{ass:prte-rank} gives $v'\bG^{\mathrm{PRTE}}(P)v>0$; thus the objective is strictly convex on the simplex. Compactness of $\Delta^{L-1}$ gives existence and strict convexity gives uniqueness.

\section{Covariates and Design Implementation}\label{app:covariates}

\subsection{Conditional and marginal RT with covariates}

\begin{assumption}[Conditional LATE conditions]\label{ass:conditional_late}\leavevmode
\begin{enumerate}[label=(\roman*)]
  \item  {Conditional joint independence:} $(Y_i(0), Y_i(1), D_i(\cdot)) \perp\!\!\!\perp \bZ_i \mid X_i$.
  \item  {SUTVA:} $Y_i = D_i Y_i(1) + (1-D_i)Y_i(0)$ and $D_i = D_i(\bZ_i, X_i)$.
  \item  {Conditional monotonicity:} For each instrument $\ell$, all~$i$, and all configurations $\bz_{-\ell} \in \{0,1\}^{L-1}$,
  $D_i(Z_{\ell i} = 1, \bz_{-\ell}, X_i)
    \geq D_i(Z_{\ell i} = 0, \bz_{-\ell}, X_i)$.
  \item  {Conditional relevance:} For each instrument $\ell$, $\E[D_i \mid Z_{\ell i} = 1, X_i = x] \neq \E[D_i \mid Z_{\ell i} = 0, X_i = x]$ on the covariate cells used for the conditional estimand; for pooled FWL estimands it is enough that the aggregated first stage defined below is nonzero.
  \item  {Full saturation:} The first-stage specification is fully saturated in $(Z_{\ell i}, X_i)$.
  \item  {Non-overlapping conditional compliance:} For each instrument $\ell$, $\mathcal{C}_\ell(x) \cap \mathcal{C}_k(x) = \emptyset$ for all $k \neq \ell$ and all $x$.
  \item  {Conditional PRD:} For each $\ell$ and each $x$, $Z_{-\ell}$ is positively regression dependent on $Z_\ell$ conditional on $X_i = x$.
\end{enumerate}
\end{assumption}

\noindent Assumption~\ref{ass:conditional_late} is maintained jointly with the sampling and regularity conditions of Assumption~\ref{ass:sampling}, applied to the augmented observation $(Y_i, D_i, \bZ_i, X_i)$.

\noindent Full saturation (Condition~(v)) supplies the rich covariates and monotonicity-correct first stage required by \citet{BlandholBonneyMogstadTorgovitsky2022}. Together with conditional monotonicity, it makes the fitted values equal to correctly ordered conditional treatment probabilities. The same logic applies to RT because it combines conditional Wald estimands; without saturation, the estimand need not be weakly causal.

\begin{proposition}[RT with covariates]\label{prop:covariates}
Suppose $X_i$ has finite support, or is coarsened into a finite saturated partition, and maintain Conditions~(i)--(v) and~(vii) of Assumption~\ref{ass:conditional_late}.
\begin{enumerate}[label=(\alph*)]
  \item For covariate-specific target weights $\omega_\ell(x)\in\mathrm{int}(\Delta^{L-1})$, conditional RT converges to
  $\beta^{\mathrm{RT}}_{\bomega}(x)=\sum_\ell\omega_\ell(x)\Wald_\ell(x)$.
  Under Condition~(vi) and conditional instrument independence, $\Wald_\ell(x)=\LATE_\ell(x)$ for a single complier type.
  \item If $\omega_\ell$ does not vary with $x$, averaging conditional Walds with cell shares gives
  $\beta^{\mathrm{RT}}_{\bomega}=\sum_\ell\omega_\ell\overline{\Wald}_\ell$, where $\overline{\Wald}_\ell=\E_X[\Wald_\ell(X_i)]$; under the same specialization, $\overline{\Wald}_\ell=\E_X[\LATE_\ell(X_i)]$. This cell-share average is generally distinct from both the raw marginal Wald ratio and the FWL ratio; it equals the latter when $\Cov_X(\gamma_\ell(X),\Wald_\ell(X))=0$. Appendix~\ref{app:fwl_bridge} characterizes the FWL estimand used in the applications.
\end{enumerate}
\end{proposition}

\begin{proof}
Fix a positive-probability cell $x$ satisfying conditional relevance. Conditional joint independence, SUTVA, and monotonicity give the within-cell decomposition
$\Wald_\ell(x) = \frac{\Cov(Y_i,Z_{\ell i}\mid X_i=x)}{\Cov(D_i,Z_{\ell i}\mid X_i=x)} = \sum_{t\in\mathcal{T}(x)}\alpha_{\ell,t}(x)\LATE_t(x).$
Conditional PRD makes the weights nonnegative, and the within-cell law of large numbers proves part~(a). Under Condition~(vi), each complier group contains one type; conditional instrument independence removes the indirect channel, so $\Wald_\ell(x)=\LATE_\ell(x)$. For part~(b), finite support gives $n_x/n\to_p\mathbb{P}(X=x)$ jointly across cells and hence
$\sum_x \frac{n_x}{n}\sum_\ell\omega_\ell\widehat\Wald_\ell(x) \;\to_p\; \sum_\ell\omega_\ell\sum_x \mathbb{P}(X=x)\Wald_\ell(x) =\sum_\ell\omega_\ell\overline{\Wald}_\ell.$
The same non-overlap and independence specialization gives $\overline{\Wald}_\ell=\E_X[\LATE_\ell(X)]$.
\end{proof}

\subsection{FWL residualization and conditional binary instruments}\label{app:fwl_bridge}

The binary-instrument theory applies within covariate cells, not to the residualized variables. For a raw binary instrument $Z_\ell$ and saturated $X$, the FWL moment $\E[(Y-\beta D)(Z_\ell-\E[Z_\ell\mid X])]=\E_X[\Cov(Y,Z_\ell\mid X)-\beta\Cov(D,Z_\ell\mid X)]$ gives the residualized Wald ratio
$
\Wald_\ell^{\mathrm{FWL}}=\frac{\E_X[\Cov(Y,Z_\ell\mid X)]}{\gamma_\ell^{\mathrm{FWL}}}=\sum_x\frac{\mathbb{P}(X=x)\,\gamma_\ell(x)}{\gamma_\ell^{\mathrm{FWL}}}\,\Wald_\ell(x),
\gamma_\ell^{\mathrm{FWL}}=\sum_x\mathbb{P}(X=x)\,\gamma_\ell(x),
$
a first-stage-weighted average of conditional raw-binary Walds over the cells $x \in \mathcal{X}_\ell^+ = \{x : \gamma_\ell(x) \neq 0\}$ with within-cell instrument variation (cells with $\gamma_\ell(x) = 0$ contribute $\Cov(Y, Z_\ell \mid X{=}x) = 0$ to the numerator and carry zero weight), equal to $\E_X[\Wald_\ell(X)]$ exactly when $\Cov_X\bigl(\gamma_\ell(X), \Wald_\ell(X)\bigr) = 0$, for instance when $\gamma_\ell(x)$ is constant across cells or the conditional Walds are. Under the conditional-on-$X_i$ analogs of joint independence, exclusion, relevance, monotonicity, and PRD (Assumption~\ref{ass:conditional_late}(i)--(v),(vii)), the same weights aggregate the conditional compliance shares: $\alpha^{\mathrm{FWL}}_{\ell,(x,t)}=\mathbb{P}(X=x)\gamma_\ell(x)/\gamma_\ell^{\mathrm{FWL}}\cdot\alpha_{\ell,t}(x)$, and every Section~\ref{sec:specifications} projection formula applies after replacing $(\bA,\bpsi^\star,\bkappa,\bGamma^{\Wald})$ by their stacked analogs over $\mathcal{T}^X=\{(x,t):t\in\mathcal{T}(x)\}$. This is the object used by the patent cumulative-threshold instruments after art-unit--year residualization. Without full saturation, the FWL ratio confounds the heterogeneity penalty with the first-stage misspecification of \citet{BlandholBonneyMogstadTorgovitsky2022}; saturation isolates the penalty as the sole mechanism.

\subsection{Monotonicity and the separate-threshold first stage}\label{app:cumulative_monotone}

This subsection shows that nested monotonicity implies a nonnegative separate-threshold first stage in any cumulative-threshold design. Let $G_i \in \{1, \ldots, Q\}$ denote a discrete ordered index, $Z_{\ell i} = \ind\{G_i \geq \ell + 1\}$ the $\ell$th cumulative threshold for $\ell = 1, \ldots, Q - 1$, and $X_i$ a vector of pre-treatment conditioning variables. Maintain SUTVA, $D_i = D_i(G_i)$ a.s., and suppose that, conditional on $X_i = x$, $G_i$ is independent of the potential-treatment profile $\{D_i(1), \ldots, D_i(Q)\}$ and that nested monotonicity holds: $D_i(g+1) \geq D_i(g)$ a.s.\ for every $g \in \{1, \ldots, Q-1\}$. Write $m_g(x) = \E[D_i(g) \mid X_i = x]$. For a single threshold $\ell$, on cells with $0 < \mathbb{P}(Z_{\ell i} = 1 \mid X_i = x) < 1$,
$
\delta_\ell(x)
= \E[D_i \mid Z_{\ell i} = 1, X_i = x] - \E[D_i \mid Z_{\ell i} = 0, X_i = x]
= \sum_{g = \ell+1}^{Q} \sum_{h = 1}^{\ell} a_{g\ell}(x)\, b_{h\ell}(x)\, \{m_g(x) - m_h(x)\} \geq 0,
$
where $a_{g\ell}(x) = \mathbb{P}(G_i = g \mid G_i \geq \ell + 1, X_i = x)$ and $b_{h\ell}(x) = \mathbb{P}(G_i = h \mid G_i \leq \ell, X_i = x)$ are the conditional level shares above and below threshold $\ell$; each summand is nonnegative under nested monotonicity, since $g \geq \ell + 1 > \ell \geq h$ implies $m_g(x) \geq m_h(x)$. By the Frisch--Waugh--Lovell theorem, the fixed-effect coefficient on $Z_{\ell i}$ from a regression of $D_i$ on $Z_{\ell i}$ that absorbs $X_i$ is
$\beta_\ell^{\mathrm{FE}} = \frac{\E[\Var(Z_{\ell i} \mid X_i)\, \delta_\ell(X_i)]} {\E[\Var(Z_{\ell i} \mid X_i)]} \geq 0$
whenever the denominator is positive; cells without threshold variation receive zero regression weight. The same argument, applied to any pre-treatment subsample on which conditional independence and nested monotonicity are maintained, yields a nonnegative within-subsample coefficient.

The joint first stage on the full vector $\bZ_i = (Z_{1 i}, \ldots, Z_{Q-1, i})'$ admits a parallel statement, with a sign-preservation property that depends on whether $X_i$ is absorbed. Without absorbing $X_i$, the regression of $D_i$ on $\bZ_i$ identifies the unconditional adjacent increments $\E[D_i \mid G_i = \ell + 1] - \E[D_i \mid G_i = \ell]$, because the cumulative indicators span the space of $G_i$-measurable functions; under unconditional independence of $G_i$ and the potential-treatment profile (a strengthening of the conditional independence maintained above, needed because the unabsorbed regression does not condition on $X_i$), these equal $\E[m_{\ell+1}(X_i)] - \E[m_\ell(X_i)] \geq 0$ by nested monotonicity. With $X_i$ absorbed, the within-cell joint regression at $X_i = x$ identifies $m_{\ell+1}(x) - m_\ell(x) \geq 0$ for every cell at which all $Q$ leniency groups have positive mass (full within-cell support of $G_i$), and by the Frisch--Waugh--Lovell theorem the population FE coefficient vector aggregates these as
$\boldsymbol\beta^{\mathrm{FE}}_{\mathrm{joint}} = \E[\Var(\bZ_i \mid X_i)]^{-1}\, \E[\Var(\bZ_i \mid X_i)\, \boldsymbol\beta(X_i)],$
where $\boldsymbol\beta(X_i) = (m_2(X_i) - m_1(X_i), \ldots, m_Q(X_i) - m_{Q-1}(X_i))'$. Each component of $\boldsymbol\beta(X_i)$ is nonnegative under nested monotonicity, but the matrix inverse $\E[\Var(\bZ_i \mid X_i)]^{-1}$ has off-diagonal entries of either sign, and the matrix-weighted aggregation need not preserve componentwise nonnegativity. The joint first stage with $X_i$ absorbed is therefore a supportive sign check rather than a formal implication of nested monotonicity; the separate-threshold first stages established above remain the proof-based diagnostic.

\section{Pathwise Differentiability, Tangent Compatibility, and Uniform Coverage}\label{app:tc-primitive-supp}

\subsection{Influence functions of the primitives}\label{app:pathwise-diff}

\subsubsection{Wald estimands}
\begin{lemma}[IF of a population covariance]\label{lem:cov-if}
For random variables $U, V \in L^2(P_0)$ with $\E_{P_0}[U^2 V^2] < \infty$ (in particular when one of $U,V$ is bounded and the other lies in $L^2(P_0)$, as in the Wald application below), the covariance functional $\rho_{UV}(P) = \Cov_P(U, V) = \E[UV] - \E[U] \E[V]$ is pathwise differentiable at every $P_0$ with Riesz representer
$\phi_{\rho_{UV}}(O; P_0) = \bigl(U - \mu_U(P_0)\bigr)\bigl(V - \mu_V(P_0)\bigr) - \rho_{UV}(P_0),$
where $\mu_U(P_0) = \E[U]$, $\mu_V(P_0) = \E[V]$.
\end{lemma}

\begin{proof}
Along $P_{s,h}$ with $h \in L^2_0(P_0)$, $\dot\rho_{UV}(h) = \E[(UV - U\mu_V - V\mu_U) h]$ by the product rule on $\rho_{UV} = \E[UV] - \E[U]\E[V]$. Since $h$ is mean-zero, any constant added to the bracketed function leaves the inner product unchanged; the mean-zero choice is $(U-\mu_U)(V-\mu_V) - \rho_{UV}$, with mean $\rho_{UV} - \rho_{UV} = 0 \in L^2_0(P_0)$.
\end{proof}

\begin{lemma}[Pathwise differentiability of $\Wald_\ell$]\label{lem:wald-pathwise}
Under Assumptions~\ref{ass:joint_indep}--\ref{ass:relevance} and Assumption~\ref{ass:sampling}, $\Wald_\ell(P) = \Cov_P(Y_i,Z_{\ell i})/\gamma_\ell(P)$ is pathwise differentiable at every $P_0 \in \PLATE$ with Riesz representer
$
\phi_\ell(O; P_0) = \frac{\bigl[(Y - \mu_Y) - \Wald_\ell(P_0)(D - \mu_D)\bigr](Z_\ell - p^Z_\ell)}{\gamma_\ell(P_0)},
$
with $\mu_Y = \E[Y]$, $\mu_D = \E[D]$, $p^Z_\ell = \E[Z_\ell]$. The covariance entry is $\Gamma^{\Wald}_{\ell k}(P_0)
=
\frac{1}{\gamma_\ell(P_0)\gamma_k(P_0)}
\E\Bigl[
\bigl\{Y-\mu_Y-\Wald_\ell(D-\mu_D)\bigr\}
\times
\bigl\{Y-\mu_Y-\Wald_k(D-\mu_D)\bigr\}
(Z_\ell-p^Z_\ell)(Z_k-p^Z_k)
\Bigr].
$
\end{lemma}

\begin{proof}
 The bounded $Z_\ell, D\in\{0,1\}$ and $\E[Y^2]<\infty$ from Assumption~\ref{ass:sampling} verify the moment hypothesis of Lemma~\ref{lem:cov-if}; applying it to $\Cov(Y,Z_\ell)$ and $\gamma_\ell$ places $\phi_{\Cov(Y,Z_\ell)},\phi_{\gamma_\ell}\in L^2_0(P_0)$, with $\phi_{\Cov(Y,Z_\ell)} = (Y-\mu_Y)(Z_\ell - p^Z_\ell) - \Cov(Y,Z_\ell)$ and $\phi_{\gamma_\ell} = (D-\mu_D)(Z_\ell - p^Z_\ell) - \gamma_\ell$. The quotient rule for $\Wald_\ell = \Cov(Y,Z_\ell)/\gamma_\ell$ gives $\phi_\ell = (\phi_{\Cov(Y,Z_\ell)} - \Wald_\ell \phi_{\gamma_\ell})/\gamma_\ell$. The constant correction $\Cov(Y,Z_\ell) - \Wald_\ell \gamma_\ell = 0$ vanishes by definition of $\Wald_\ell$. Covariance is by taking the product $\phi_\ell\phi_k$ and applying the law of large numbers.
\end{proof}

   The FWL-residualized empirical implementations use the within-cell analog, with $Y_i$, $D_i$, and $Z_{\ell i}$ demeaned cell by cell and $\gamma^{\mathrm{FWL}}_\ell$ in the denominator, which is the influence function of the FWL functional of Appendix~\ref{app:fwl_bridge}; the display is its unconditional counterpart.

\subsubsection{Compliance-type weights}

\begin{lemma}[IF of a conditional expectation at a point of positive mass]\label{lem:cond-exp-if}
For $z \in \supp(\bZ)$ with $\mu_z(P_0) \equiv  \mathbb{P}_{0}(\bZ_i = z) > 0$, $p(z)(P) \equiv  \E[D_i \mid \bZ_i = z]$ is pathwise differentiable at $P_0$ with Riesz representer $\phi_{p(z)}(O; P_0) = \frac{\ind\{\bZ_i = z\}\,(D_i - p(z)(P_0))}{\mu_z(P_0)}.$
\end{lemma}

\begin{proof}
$p(z)(P) = \E[D \ind\{\bZ = z\}] / \mathbb{P}(\bZ = z)$. By the influence function of a mean ratio (the IF of an expectation is its centered integrand, applied to numerator and denominator and combined by the quotient rule): the IF is $(D \ind\{\bZ = z\} - p(z)(P_0)\ind\{\bZ = z\})/\mu_z(P_0)$, which simplifies to the stated formula.
\end{proof}

\begin{lemma}[Pathwise differentiability of $\alpha_{\ell,t}(P)$]\label{lem:A_pathwise}
Under Assumptions~\ref{ass:joint_indep}--\ref{ass:relevance}, Assumption~\ref{ass:sampling}, Assumption~\ref{ass:latent_index}, and Assumption~\ref{ass:local-threshold-support} for the threshold intervals entering the displayed entry, each weight $\alpha_{\ell, t_k}(P)$ is pathwise differentiable at $P_0 \in \PLATE$ with Riesz representer
\begin{equation}\label{eq:A-if-explicit}
\begin{aligned}
\xi_{\ell, t_k}(O; P_0)
&=
\frac{
\phi_{\theta_{t_k}}(O; P_0)\,\varphi_{\ell,t_k}(P_0)
 + \theta_{t_k}(P_0)\,\phi_{\varphi_{\ell,t_k}}(O; P_0)}
{\pi_\ell(P_0)} - \frac{
\theta_{t_k}(P_0)\,\varphi_{\ell,t_k}(P_0)\,\phi_{\pi_\ell}(O; P_0)}
{\pi_\ell(P_0)^2},
\end{aligned}
\end{equation}
where:
\begin{align}
\phi_{\theta_{t_k}}(O; P_0) &=
\begin{cases}
\phi_{p(z^{(1)})}, & k = 0,\\
\phi_{p(z^{(k+1)})}-\phi_{p(z^{(k)})}, & 1 \leq k \leq K-1,\\
-\phi_{p(z^{(K)})}, & k = K,
\end{cases}
(O;P_0),\label{eq:theta-if-app}\\
\phi_{\pi_\ell}(O; P_0)
&= \frac{Z_\ell\,(D - \E[D \mid Z_\ell = 1])}{\mathbb{P}_{0}(Z_\ell = 1)} - \frac{(1 - Z_\ell)\,(D - \E[D \mid Z_\ell = 0])}{\mathbb{P}_{0}(Z_\ell = 0)} \notag \\
\phi_{\varphi_{\ell,t_k}}(O; P_0)
&= \sum_{z_{-\ell}}
\Bigl[
t_k(1, z_{-\ell})\,\phi_{q_\ell(z_{-\ell})}
- t_k(0, z_{-\ell})\,\phi_{q^0_\ell(z_{-\ell})}
\Bigr](O;P_0), \label{eq:varphi-if-app}
\end{align}
The representers $\phi_{\pi_\ell}$ and~\eqref{eq:varphi-if-app} follow from Lemma~\ref{lem:cond-exp-if} and the conditional-probability representers below. Here $q_\ell(z_{-\ell};P)=\mathbb{P}(\bZ_{-\ell}=z_{-\ell}\mid Z_\ell=1)$, $q^0_\ell(z_{-\ell};P)=\mathbb{P}(\bZ_{-\ell}=z_{-\ell}\mid Z_\ell=0)$, and the conditional-probability representers are
$
\phi_{q_\ell(z_{-\ell})}(O;P_0)
=
\frac{Z_\ell\,[\ind\{\bZ_{-\ell}=z_{-\ell}\}-q_\ell(z_{-\ell};P_0)]}
     {\mathbb{P}_{0}(Z_\ell=1)}$, and $
\phi_{q^0_\ell(z_{-\ell})}(O;P_0)
=
\frac{(1-Z_\ell)\,[\ind\{\bZ_{-\ell}=z_{-\ell}\}-q^0_\ell(z_{-\ell};P_0)]}
     {\mathbb{P}_{0}(Z_\ell=0)}.
$
\end{lemma}

\begin{proof}
Apply the product/quotient rule to $\alpha_{\ell, t_k} = \theta_{t_k} \varphi_{\ell,t_k} / \pi_\ell$ using the three factor IFs~\eqref{eq:theta-if-app}--\eqref{eq:varphi-if-app}. Each factor IF is in $L^2_0(P_0)$ by Lemmas~\ref{lem:cov-if} and~\ref{lem:cond-exp-if} and by the displayed conditional-probability ratio formula; the linear combination~\eqref{eq:A-if-explicit} is in $L^2_0(P_0)$. Assumption~\ref{ass:local-threshold-support} ensures local stability of the relevant propensity ordering: by Lemma~\ref{lem:threshold-types-equiv}, the endpoints defining threshold intervals have positive minimum gaps at $P_0$, and the gaps remain positive on a neighborhood by continuity. The labelling $z^{(k)}$ is therefore well-defined locally for the endpoints. Local uniqueness of the endpoint labels (Assumption~\ref{ass:local-threshold-support}, via Lemma~\ref{lem:threshold-types-equiv}) rules out ties: each endpoint value $p_k$ is attained by a single cell $z^{(k)}$, and $\phi_{p(z^{(k)})}$ is the influence function of a single-cell conditional mean. The same local stability makes the type functions $t_k(z) = \ind\{p(z) \geq p_{k+1}\}$ locally constant in $P$ on $\supp(\bZ)$, which is why they enter~\eqref{eq:varphi-if-app} derivative-free. For interior types, the IF of $\theta_{t_k}(P) = p(z^{(k+1)})(P) - p(z^{(k)})(P)$ is the difference of two conditional-expectation IFs from Lemma~\ref{lem:cond-exp-if}; for the boundary types, $\theta_{t_0}=p(z^{(1)})-0$ and $\theta_{t_K}=1-p(z^{(K)})$, and the constants $0$ and $1$ have zero derivative. This gives~\eqref{eq:theta-if-app}.
\end{proof}

\subsubsection{Riesz representers of canonical \texorpdfstring{$(\bpsi^\star(P), \bkappa(P))$}{(psi*(P), kappa(P))} specifications}\label{app:proof_canonical_targets}

\begin{lemma}[Smoothness of canonical metric maps]\label{lem:kappa-pathwise}
Under Assumption~\ref{ass:latent_index}, Assumption~\ref{ass:local-threshold-support} for threshold intervals, and Assumption~\ref{ass:late-regularity}, each of the canonical metric maps $\bkappa^{\mathrm{CS}}(P) \equiv \mathbf{1}$, $\bkappa^{\mathrm{EW}}(P) \equiv \mathbf{1}$, and $\bkappa^{\mathrm{PRTE}}_{t_k}(P) = \theta_{t_k}(P)^{-1}$ on PRTE intervals is Hadamard-differentiable at $P_0$ tangentially to $L^2_0(P_0)$. The Riesz representer $\zeta_t(\cdot; P_0)$ of the pathwise derivative $\dot\kappa_t(h; P_0) = \inner{\zeta_t(\cdot; P_0)}{h}_{L^2(P_0)}$ lies in $L^2_0(P_0)$ for each coordinate used by the specification.
\end{lemma}

\noindent \emph{Proof.} For $\bkappa^{\mathrm{CS}}$ and $\bkappa^{\mathrm{EW}}$ the map is constant; hence $\dot\kappa_t \equiv 0$ and $\zeta_t \equiv 0 \in L^2_0(P_0)$. For $\bkappa^{\mathrm{PRTE}}$, the threshold-type probabilities are smooth functions of the ordered cell propensities, with boundary convention $\theta_{t_0}(P)=p_1(P)-0$, $\theta_{t_k}(P)=p_{k+1}(P)-p_k(P)$ for $1\leq k\leq K-1$, and $\theta_{t_K}(P)=1-p_K(P)$. Each observed propensity $p(z)(P) = \E[D_i \mid \bZ_i = z]$ is pathwise differentiable at $P_0$ with Riesz representer $\phi_{p(z)} \in L^2_0(P_0)$ when $\mu_z(P_0)>0$ (Lemma~\ref{lem:cond-exp-if}). Since the relevant propensity ordering is locally constant at $P_0$ under Assumption~\ref{ass:local-threshold-support} (Lemma~\ref{lem:threshold-types-equiv}), the Riesz representer of $\dot\theta_{t_k}(\cdot; P_0)$ is $\phi_{\theta_{t_k}}(\cdot; P_0)$ of~\eqref{eq:theta-if-app}; the artificial endpoints $0$ and $1$ are constants and have zero derivative. Hadamard differentiability strengthens this pathwise calculus. For bounded $f$ and quadratic-mean-differentiable paths $\{P_{s,h_s}\}$ with scores $h_s \to h$ in $L^2_0(P_0)$, writing $dP_{s,h_s} = (1 + \tfrac{s}{2}h_s + r_s)^2\,dP_0$ with $\|r_s\|_{L^2(P_0)} = o(s)$ bounds $s^{-1}\{\E_{P_{s,h_s}}[f] - \E_{P_0}[f]\} - \inner{f}{h_s}_{L^2(P_0)}$ in absolute value by $\|f\|_\infty\, s^{-1}\bigl(2\|r_s\| + \tfrac{s^2}{4}\|h_s\|^2 + s\|r_s\|\|h_s\| + \|r_s\|^2\bigr) \to 0$; hence $s^{-1}\{\E_{P_{s,h_s}}[f] - \E_{P_0}[f]\} \to \inner{f}{h}_{L^2(P_0)}$. Each cell mass $\mu_z(P) = \E_P[\ind\{\bZ=z\}]$ and joint moment $\E_P[D_i\ind\{\bZ=z\}]$ is of this form; the propensity $p(z)(P)$ is their ratio at the positive denominator $\mu_z(P_0) > 0$, and on the neighborhood of locally constant ordering and labels each $\theta_{t_k}(P)$ is a fixed affine function of the $p(z)(P)$; the chain rule for Hadamard-differentiable maps \citep[Theorem~20.9]{vanderVaart1998} therefore makes $\theta_{t_k}$ Hadamard-differentiable tangentially to $L^2_0(P_0)$ with representer $\phi_{\theta_{t_k}}$. For $\kappa^{\mathrm{PRTE}}_{t_k} = \theta_{t_k}^{-1}$, the quotient expansion
$
\frac{1}{\theta_{t_k}(P_{s,h_s})} - \frac{1}{\theta_{t_k}(P_0)}
= -\frac{\theta_{t_k}(P_{s,h_s}) - \theta_{t_k}(P_0)}{\theta_{t_k}(P_0)^2}
+ \frac{\{\theta_{t_k}(P_{s,h_s}) - \theta_{t_k}(P_0)\}^2}{\theta_{t_k}(P_0)^2\,\theta_{t_k}(P_{s,h_s})}
$
completes the proof: $\theta_{t_k}(P_0) > 0$ under Assumption~\ref{ass:local-threshold-support}, the Hadamard differentiability of $\theta_{t_k}$ makes $\theta_{t_k}(P_{s,h_s}) - \theta_{t_k}(P_0) = O(s)$ along every converging-score sequence; hence $\theta_{t_k}(P_{s,h_s}) \geq \theta_{t_k}(P_0)/2$ for all small $s$ and the second term is $O(s^2) = o(s)$; hence $\dot\kappa^{\mathrm{PRTE}}_{t_k}(h; P_0) = -\theta_{t_k}(P_0)^{-2}\dot\theta_{t_k}(h; P_0)$ with Riesz representer $\zeta^{\mathrm{PRTE}}_{t_k}(\cdot; P_0) = -\theta_{t_k}(P_0)^{-2}\phi_{\theta_{t_k}}(\cdot; P_0) \in L^2_0(P_0)$.

For the three canonical specifications of Section~\ref{sec:specifications}, the Riesz representers $\xi^\star_t(\cdot; P_0)$ of $\dot\psi^\star_t(\cdot; P_0)$ and $\zeta_t(\cdot; P_0)$ of $\dot\kappa_t(\cdot; P_0)$ are computed from the structural parameterization.

\emph{CS-ATE} ($\bkappa^{\mathrm{CS}} \equiv 1$, hence $\zeta_t^{\mathrm{CS}} \equiv 0$): $\psi^{\star,\mathrm{CS}}_{t_k}(P) = \theta_{t_k}(P)\ind\{t_k \in \mathcal{C}_{\mathrm{ever}}(P)\}/S_C(P)$ with $S_C(P)=\sum_{t' \in \mathcal{C}_{\mathrm{ever}}(P)}\theta_{t'}(P)$. Under local support positivity and stable propensity ordering on the ever-complier set, $\mathcal{C}_{\mathrm{ever}}(P)$ is locally constant, and the Riesz representer of $\dot\psi^{\star,\mathrm{CS}}_{t_k}(\cdot;P_0)$ is
$\xi^{\star,\mathrm{CS}}_{t_k}(\cdot; P_0) = S_C(P_0)^{-1}\bigl[\phi_{\theta_{t_k}}(\cdot; P_0) - \psi^{\star,\mathrm{CS}}_{t_k}(P_0)\textstyle\sum_{t'\in\mathcal{C}_{\mathrm{ever}}(P_0)} \phi_{\theta_{t'}}(\cdot; P_0)\bigr]$
on $\mathcal{C}_{\mathrm{ever}}(P_0)$ and zero elsewhere, where $\phi_{\theta_{t_k}}(\cdot; P_0)$ is the Riesz representer of the pathwise derivative of $\theta_{t_k}$.

\emph{EW-ATE} ($\bkappa^{\mathrm{EW}} \equiv 1$, hence $\zeta_t^{\mathrm{EW}} \equiv 0$): $\psi^{\star,\mathrm{EW}}_{t_k}(P) = \ind\{t_k \in \mathcal{C}_{\mathrm{ever}}(P)\}/|\mathcal{C}_{\mathrm{ever}}(P)|$ is locally constant at $P_0$ under local support positivity and stable propensity ordering; hence $\xi^{\star,\mathrm{EW}}_{t} \equiv 0$ as well.

\emph{PRTE} ($\bkappa^{\mathrm{PRTE}}_{t_k}(P) = \theta_{t_k}(P)^{-1}$ on threshold intervals): both $\xi^\star$ and $\zeta$ contribute. By Lemma~\ref{lem:prte-equivalence} the projection depends on the policy weight only through the interval masses $\psi^{\star,\mathrm{PRTE}}_{t_k}(P)$; hence the moving-endpoint integral in~\eqref{eq:prte_integral} is differentiated through its closed ratio form (Assumption~\ref{ass:prte-weight-regularity})
$\psi^{\star,\mathrm{PRTE}}_{t_k}(P)=\frac{\theta_{t_k}(P)\,\Delta^{\mathrm{PRTE}}_k(P)}{\bar\Delta(P)}, \qquad \bar\Delta(P)=\E_{F^Z_1}[D_i]-\E_{F^Z_0}[D_i],$
rather than through a Leibniz expansion requiring pointwise endpoint values of $w^P$. By Lemma~\ref{lem:type_interval} type $t_k$ is treated exactly when $p(\bZ)\ge p_{k+1}$, and $\bZ$ is independent of $U$; hence the type-conditional shift of~\eqref{eq:prte_integral} equals the unconditional threshold-crossing shift $\Delta^{\mathrm{PRTE}}_k(P)=\mathbb{P}_{F^Z_1}(p(\bZ)\ge p_{k+1})-\mathbb{P}_{F^Z_0}(p(\bZ)\ge p_{k+1})$, locally constant under the stable propensity ordering of Assumption~\ref{ass:local-threshold-support}. For fixed policy laws only $\theta_{t_k}(P)$ and $\bar\Delta(P)=\sum_z(F^Z_1(z)-F^Z_0(z))\,p(z)(P)$ carry derivatives, and the quotient rule gives the Riesz representer
$\xi^{\star,\mathrm{PRTE}}_{t_k}(\cdot;P_0) = \frac{\Delta^{\mathrm{PRTE}}_k(P_0)}{\bar\Delta(P_0)}\,\phi_{\theta_{t_k}}(\cdot;P_0) - \frac{\theta_{t_k}(P_0)\,\Delta^{\mathrm{PRTE}}_k(P_0)}{\bar\Delta(P_0)^2}\,\phi_{\bar\Delta}(\cdot;P_0),$
with $\phi_{\theta_{t_k}}$ from~\eqref{eq:theta-if-app}, the normalizer representer $\phi_{\bar\Delta}=\sum_z(F^Z_1(z)-F^Z_0(z))\,\phi_{p(z)}$, and the realized-propensity representers $\phi_{p(z)}$ from Lemma~\ref{lem:cond-exp-if}. For a generic within-support staircase whose baseline law $F^Z_0$ is the population instrument distribution and whose top group is fixed, the swept-margin shifts become cell masses, $\Delta^{\mathrm{PRTE}}_k(P) = \mu_{z^{(k)}}(P)$ for $k=1,\ldots,K-1$, and the normalizer $\bar\Delta(P)=\sum_{j=1}^{K-1}\theta_{t_j}(P)\,\mu_{z^{(j)}}(P)$ inherits the same dependence; the chain rule then adds the corresponding cell-mass representers to both channels of $\xi^{\star,\mathrm{PRTE}}_{t_k}$. This calculation describes the observed component of the patent staircase, not the complete target that also moves group $G=7$ to the synthetic group $G=8$. The metric representer $\zeta^{\mathrm{PRTE}}_{t_k}$ is that of Lemma~\ref{lem:kappa-pathwise}. In all three cases $\xi^\star_t, \zeta_t \in L^2_0(P_0)$ under Assumption~\ref{ass:late-regularity} and Assumption~\ref{ass:local-threshold-support} where a threshold inverse is used.

\subsubsection{The RT estimand}

\begin{lemma}[Pathwise differentiability of the RT estimand]\label{lem:rt-pathwise}
Under Assumptions~\ref{ass:relevance} and~\ref{ass:sampling}, if $\bomega(\cdot)$ is pathwise differentiable at $P_0$ with Riesz representers $\Riesz_{\omega,\ell}\in L^2_0(P_0)$, then $\beta^{\mathrm{RT}}_{\bomega}(P) = \bomega(P)'\bWald(P)$ is pathwise differentiable at $P_0$ with observed-law Riesz representer $\eif_{\omega(\cdot)}$ of~\eqref{eq:eif-general}, each summand lying in $L^2_0(P_0)$.
\end{lemma}

\begin{proof}
Write $\beta^{\mathrm{RT}}_{\bomega}(P) = \bomega(P)' \bWald(P)$ as a scalar product of two pathwise-differentiable functionals: $\bWald: \PLATE \to \mathbb{R}^L$ and $\bomega(\cdot): \PLATE \to \Delta^{L-1}$. Both are pathwise differentiable in $\mathbb{R}^L$; hence their scalar product is pathwise differentiable. Along a LATE-compatible contiguous submodel $P_{s,h}\in\PLATE$ with score $h \in \TLATE$, $\bomega(P_{s,h}) = \bomega(P_0) + s\,\dot\bomega(h; P_0) + o(s)$ and $\bWald(P_{s,h}) = \bWald(P_0) + s\,\dot\bWald(h; P_0) + o(s)$, with $\dot\Wald_\ell(h; P_0) = \inner{\phi_\ell}{h}_{L^2(P_0)}$ and $\dot\omega_\ell(h; P_0) = \inner{\Riesz_{\omega,\ell}}{h}_{L^2(P_0)}$. The same algebra gives the ambient observed-law derivative along observed-law paths for which the ambient extension is valid. The product expansion gives
$
\beta^{\mathrm{RT}}_{\bomega}(P_{s,h}) = \beta^{\mathrm{RT}}_{\bomega}(P_0) + s\bigl[\bomega(P_0)'\dot\bWald(h;P_0) + \dot\bomega(h;P_0)'\bWald(P_0)\bigr] + o(s).
$
The pathwise derivative is $\dot\beta^{\mathrm{RT}}_{\bomega}(h;P_0) = \bomega(P_0)'\dot\bWald(h;P_0) + \dot\bomega(h;P_0)'\bWald(P_0)$. By linearity of Riesz representation in $L^2_0(P_0)$ and the representations of $\dot\Wald_\ell$ and $\dot\omega_\ell$ above, the Riesz representer of $\dot\beta^{\mathrm{RT}}_{\bomega}(h; P_0)$ is $\eif_{\omega(\cdot)}(O_i; P_0) = \sum_\ell \omega_\ell(P_0) \phi_\ell(O_i) + \sum_\ell \Wald_\ell(P_0) \Riesz_{\omega,\ell}(O_i),$
which is~\eqref{eq:eif-general}. Each summand is in $L^2_0(P_0)$.
\end{proof}

\subsection{Tangent compatibility}

With the cell parameter $\eta(P) = \{\mu_z, m_Y(z), p(z)\}_{z\in\mathcal{Z}}$ defined in Appendix~\ref{app:regularity}, the canonical cell influence functions are
\begin{equation}\label{eq:cell-ifs-paper}
\begin{aligned}
\chi_z^\mu(O) &= \ind\{\bZ=z\}-\mu_z,\\
\chi_z^D(O) &= \frac{\ind\{\bZ=z\}(D-p(z))}{\mu_z},\\
\chi_z^Y(O) &= \frac{\ind\{\bZ=z\}(Y-m_Y(z))}{\mu_z}.
\end{aligned}
\end{equation}
Throughout, $\overline{\mathcal{G}}$ denotes the LATE tangent space at $P_0$ (Definition~\ref{def:late-tangent}), $\mathcal{T}_+(P_0)$ the positive-mass compliance types, $\theta_t(P_0)=P_0(D_i(\cdot)=t)$, and $f_{t,d}$ the type-conditional density of $Y(d)$ under $P_0$ (Assumption~\ref{ass:late-regularity}).

\begin{condition}[Cell-score liftability]\label{cond:cell-lift-paper}
For every $z\in\mathcal{Z}$ and every cell influence function $\chi\in\{\chi_z^\mu,\chi_z^D,\chi_z^Y\}$ in~\eqref{eq:cell-ifs-paper}, $\chi\in\overline{\mathcal{G}}$. A sufficient condition is that there exist functions $a_Z(\cdot)$ on $\mathcal{Z}$, $a_T(\cdot)$ on $\mathcal{T}_+(P_0)$, and $\{b_{t,d}(\cdot)\}_{t\in\mathcal{T}_+(P_0),\,d\in\{0,1\}}$ satisfying
$\sum_{z'\in\mathcal{Z}}\mu_{z'}\,a_Z(z')=0,\quad \sum_{t\in\mathcal{T}_+(P_0)}\theta_t\,a_T(t)=0,\quad \int b_{t,d}(y)\,f_{t,d}(y)\,dy=0,$
together with the square-integrability bound
$\sum_{z'}\mu_{z'}\,a_Z(z')^2+\sum_t\theta_t\,a_T(t)^2 +\sum_{t,d}\theta_t\int b_{t,d}(y)^2\,f_{t,d}(y)\,dy<\infty,$
and, for $P_0$-a.e.\ $(y,d,z')$,
\begin{equation}\label{eq:cell-lift-eq-paper}
\chi(y,d,z') = a_Z(z')
+\sum_{t:\,t(z')=d}
\frac{\theta_t\,f_{t,d}(y)}{\sum_{r:\,r(z')=d}\theta_r\,f_{r,d}(y)}\,
\{a_T(t)+b_{t,d}(y)\},
\end{equation}
with rows of zero denominator omitted because they lie outside the observed support.    
\end{condition}

\begin{condition}[Cell-smooth RT weight map]\label{cond:cell-smooth-paper}
On a neighborhood of $P_0$ the target weight map admits the representation
$\bomega(P)=g\bigl(\eta(P),\,r(P)\bigr), \qquad \eta(P)=\{\mu_z,m_Y(z),p(z)\}_{z\in\mathcal{Z}},$
where $g$ is continuously differentiable at $(\eta(P_0),r(P_0))$ and each additional primitive $r_j(P)$ is pathwise differentiable with Riesz representer in $\overline{\mathcal{G}}$. Fixed weights are included as the special case $g\equiv\bomega_0$, $r$ absent.
\end{condition}

\begin{theorem}[Primitive sufficient conditions for tangent compatibility]\label{thm:tc-primitive}
Assume the maintained Assumptions~\ref{ass:joint_indep}, \ref{ass:monotonicity}, \ref{ass:realized_exclusion}, \ref{ass:relevance}, \ref{ass:sampling}, and~\ref{ass:late-regularity}, the finite-moment and nonzero-denominator clauses in Proposition~\ref{prop:eif-decomp}, Conditions~\ref{cond:cell-lift-paper} and~\ref{cond:cell-smooth-paper}, and the plug-in von Mises expansion required by Proposition~\ref{prop:eif-decomp}. Then the representer in~\eqref{eq:eif-general} satisfies
$\eif_{\omega(\cdot)}(\cdot;P_0) = \sum_{\ell=1}^L \omega_\ell(P_0)\,\phi_\ell(\cdot;P_0) + \sum_{\ell=1}^L \Wald_\ell(P_0)\,\Riesz_{\omega,\ell}(\cdot;P_0) \in\overline{\mathcal{G}}$; the tangent-compatibility conclusion of Proposition~\ref{prop:eif-decomp} follows. For fixed weights $\bomega(\cdot)\equiv\bomega_0$, the second sum is absent and the conclusion follows from Condition~\ref{cond:cell-lift-paper} alone.
\end{theorem}

\begin{proof}
Write $p^Z_\ell=\E[Z_\ell]$, $\mu_Y=\E[Y]$, $\mu_D=\E[D]$, and $\beta_\ell=\Wald_\ell(P_0)$, and use the identity
$Y-\mu_Y-\beta_\ell(D-\mu_D) =\{Y-m_Y(\bZ)\}-\beta_\ell\{D-p(\bZ)\} +\{m_Y(\bZ)-\mu_Y-\beta_\ell(p(\bZ)-\mu_D)\}.$
Multiplying through by $(Z_\ell-p^Z_\ell)/\gamma_\ell$, taking expectations across the cells, and rearranging gives the cell-IF decomposition
$\phi_\ell(O;P_0) =\sum_{z\in\mathcal{Z}}\mu_z\frac{z_\ell-p^Z_\ell}{\gamma_\ell}\,\chi_z^Y(O) -\sum_{z\in\mathcal{Z}}\mu_z\frac{\beta_\ell(z_\ell-p^Z_\ell)}{\gamma_\ell}\,\chi_z^D(O) +\sum_{z\in\mathcal{Z}}c_{\ell,z}\,\chi_z^\mu(O),$
with $c_{\ell,z}=(z_\ell-p^Z_\ell)\{m_Y(z)-\mu_Y-\beta_\ell(p(z)-\mu_D)\}/\gamma_\ell$. The third sum has mean zero because $\E[(Z_\ell-p^Z_\ell)\{m_Y(\bZ)-\mu_Y-\beta_\ell(p(\bZ)-\mu_D)\}]=\Cov(Y,Z_\ell)-\beta_\ell\Cov(D,Z_\ell)=0$. Under Condition~\ref{cond:cell-lift-paper}, each cell score $\chi_z^\mu,\chi_z^D,\chi_z^Y\in\overline{\mathcal{G}}$. Because $\overline{\mathcal{G}}$ is a closed linear subspace of $L^2_0(P_0)$, the finite linear combination gives $\phi_\ell\in\overline{\mathcal{G}}$ for every $\ell$. If $\bomega(\cdot)\equiv\bomega_0$, then $\Riesz_{\omega,\ell}\equiv 0$ in~\eqref{eq:eif-general} and the representer reduces to $\sum_\ell\omega_{0,\ell}\phi_\ell$. By the cell decomposition above and closedness of $\overline{\mathcal{G}}$, this is in $\overline{\mathcal{G}}$.

Under Condition~\ref{cond:cell-smooth-paper}, the finite-dimensional delta method applied to $\bomega(P)=g(\eta(P),r(P))$ gives, for each $\ell$,
$\Riesz_{\omega,\ell}(O;P_0) =\sum_{z\in\mathcal{Z}}\bigl\{a_{\ell z}^\mu\chi_z^\mu(O)+a_{\ell z}^D\chi_z^D(O)+a_{\ell z}^Y\chi_z^Y(O)\bigr\} +\sum_j b_{\ell j}\,\xi_{r_j}(O;P_0),$
with deterministic coefficients $(a_{\ell z}^\mu,a_{\ell z}^D,a_{\ell z}^Y,b_{\ell j})$ read off the partial derivatives of $g$ at $(\eta(P_0),r(P_0))$, and $\xi_{r_j}(\cdot;P_0)\in\overline{\mathcal{G}}$ by Condition~\ref{cond:cell-smooth-paper}. Therefore $\Riesz_{\omega,\ell}\in\overline{\mathcal{G}}$. Combining with the cell decomposition above and closure of $\overline{\mathcal{G}}$ under finite linear combinations yields $\eif_{\omega(\cdot)}(\cdot;P_0)\in\overline{\mathcal{G}}$.
\end{proof}

\begin{corollary}[Projected RT and Assumption~\ref{ass:projection-tangent}]\label{cor:tc-projected}
Under the maintained assumptions of Theorem~\ref{thm:tc-primitive}, suppose Assumptions~\ref{ass:strict-complementarity} and~\ref{ass:projection-dictionary} hold, with each primitive representer $\xi_{\ell,t}$, $\xi^\star_t$, and $\zeta_t$ either expressible as a finite linear combination of the cell influence functions in~\eqref{eq:cell-ifs-paper} or otherwise liftable to a LATE structural score in the sense of~\eqref{eq:cell-lift-eq-paper}. Then the projected-weight Riesz representer $\Riesz_{\omega^\dagger,\ell}(\cdot;P_0)$, evaluated at the locally active set $I_0$ and zero elsewhere, lies in $\overline{\mathcal{G}}$. Consequently, Assumption~\ref{ass:projection-tangent} is satisfied under these primitive sufficient conditions.
\end{corollary}

\begin{proof}
Suppose $q\in L^2_0(P_0)$ admits the representation~\eqref{eq:cell-lift-eq-paper} for some $(a_Z,a_T,\{b_{t,d}\})$ satisfying the centering and square-integrability conditions of Condition~\ref{cond:cell-lift-paper}. Take the LATE-compatible submodel that tilts, along the exponential chart of Appendix~\ref{app:proof_thm_lam_bound}, the instrument law by $a_Z$, the positive-support type probabilities by $a_T$, and each positive-support type-conditional outcome marginal $f_{t,d}$ by $b_{t,d}$; the centering conditions make every tilt mean-zero and the construction perturbs only the instrument law, the type-probability simplex, and the per-arm outcome marginals; hence the chart remains inside $\PLATE$. The observed-data density in cell $(d,z)$ is the type mixture $\mu_z\sum_{t:\,t(z)=d}\theta_t\,f_{t,d}(y)$, which depends on the type-conditional outcome laws only through the per-arm marginals $f_{t,d}$. Differentiating its logarithm along the submodel gives the observed-data score
$a_Z(z)+\sum_{t:\,t(z)=d}\frac{\theta_t\,f_{t,d}(y)}{\sum_{r:\,r(z)=d}\theta_r\,f_{r,d}(y)}\bigl\{a_T(t)+b_{t,d}(y)\bigr\}=q(O),$
the right side of~\eqref{eq:cell-lift-eq-paper}; the unobserved counterfactual arm $Y(1-d)$ does not enter, because the observed law is a mixture of the per-arm marginals and so carries no dependence on the within-type copula of $(Y(0),Y(1))$. For bounded $b_{t,d}$ this exhibits $q$ as an element of $\TLATE$; for merely square-integrable $b_{t,d}$, truncate each at level $k$, recenter, and let $k\to\infty$: by the $L^2(P_0)$-continuity of the structural-to-observed score map the truncated observed scores converge to $q$, hence $q\in\overline{\mathcal{G}}$.

Under Assumption~\ref{ass:strict-complementarity} the active set $I_0$ is locally stable and the bordered KKT system characterizing $\bomega^\dagger(\bpsi^\star;P)$ has a nonsingular matrix at $(\bomega^\dagger(\bpsi^\star;P_0),\bpsi^\star(P_0),P_0)$. The implicit function theorem applied to the active-face KKT conditions, as in the proof of Lemma~\ref{lem:active-set-stability} (Appendix~\ref{app:proof_lem_active}), makes $P\mapsto\bomega^\dagger(\bpsi^\star;P)$ continuously differentiable on a neighborhood of $P_0$, with derivative given by the active-face KKT formula for $\Riesz_{\omega^\dagger,\ell}$ in the proof of Proposition~\ref{prop:lam-projected} (Appendix~\ref{app:proof_prop_pathwise_omega}). The active-face derivative is therefore a finite linear combination of the primitive representers $\xi_{\ell,t}$ for $\alpha_{\ell,t}(P)$, $\xi^\star_t$ for $\psi^\star_t(P)$, and $\zeta_t$ for $\kappa_t(P)$. Under the hypothesis of Corollary~\ref{cor:tc-projected}, each primitive representer is either covered by Condition~\ref{cond:cell-lift-paper} (in which case it is in $\overline{\mathcal{G}}$ by the cell-IF decomposition in the proof of Theorem~\ref{thm:tc-primitive}) or satisfies the lift~\eqref{eq:cell-lift-eq-paper} directly (in which case it is in $\overline{\mathcal{G}}$ by the score-lift argument above). Inactive coordinates $\ell\in I_0^c$ have zero derivative by complementary slackness. Hence $\Riesz_{\omega^\dagger,\ell}(\cdot;P_0)\in\overline{\mathcal{G}}$ for every $\ell$, and Assumption~\ref{ass:projection-tangent} follows.
\end{proof}

\begin{corollary}[Tangent compatibility in the paper's finite one-to-one designs]\label{cor:tc-finite-designs}
In the ordered-threshold design, with positive-mass threshold types $\{t_0,\ldots,t_K\}$ and a realized propensity map one-to-one on $\supp(\bZ)$, and in the single-instrument non-overlap design, with complier, always-taker, and never-taker types of positive mass, suppose that within every cell $(d,z)$ the arm-$d$ outcome densities of the types appearing in that cell have bounded likelihood ratios: each positive-weight type's arm-$d$ density and the cell's mixture density are bounded above and below by positive multiples of one another on the mixture's support. Then the cell influence functions $\chi_z^\mu, \chi_z^D, \chi_z^Y$ of~\eqref{eq:cell-ifs-paper} lie in $\overline{\mathcal{G}}$; hence Condition~\ref{cond:cell-lift-paper} holds.
\end{corollary}

\begin{proof}
The cell mass score lifts through the instrument-law coordinate ($a_Z(z')=\ind\{z'=z\}-\mu_z$, $a_T\equiv b\equiv 0$, with $\sum_{z'}\mu_{z'}a_Z(z')=\mu_z-\mu_z=0$). For $\chi_z^D$ and $\chi_z^Y$, take $a_Z \equiv 0$; both scores have zero conditional mean in every cell. Order the support cells so that $p(z^{(1)}) < \cdots < p(z^{(K)})$ (the propensity map is one-to-one; the non-overlap design is the two-cell chain, its propensities distinct because the complier mass is positive), write $g_{t,d} \equiv a_T(t) + b_{t,d}$, let $\mathcal{A}_{j,1} = \{t_0, \ldots, t_{j-1}\}$ and $\mathcal{A}_{j,0} = \{t_j, \ldots, t_K\}$ collect the types in arm $d$ of cell $z^{(j)}$, let $m_{j,d} = \sum_{r \in \mathcal{A}_{j,d}} \theta_r f_{r,d}$, and let $\kappa_{j,d}$ denote the value of $\chi$ on arm $d$ of cell $z^{(j)}$, constant in $y$ for $\chi_z^D$ and linear in $y$ for $\chi_z^Y$. In cell $z^{(j)}$ and arm $d$ the lift equation~\eqref{eq:cell-lift-eq-paper} is
$
\sum_{t \in \mathcal{A}_{j,d}} \theta_t\, f_{t,d}(y)\, g_{t,d}(y) \;=\; m_{j,d}(y)\, \kappa_{j,d}(y).
$
The base equations are the pure boundary cells, $g_{t_0,1} = \kappa_{1,1}$ and $g_{t_K,0} = \kappa_{K,0}$. Since $m_{j,1} = m_{j-1,1} + \theta_{t_{j-1}} f_{t_{j-1},1}$, subtracting the cell-$z^{(j-1)}$ from the cell-$z^{(j)}$ equation in arm $1$, and the cell-$z^{(j+1)}$ from the cell-$z^{(j)}$ equation in arm $0$, solves for the entering type:
$
g_{t_{j-1},1} = \kappa_{j,1} + \bigl(\kappa_{j,1} - \kappa_{j-1,1}\bigr) \frac{m_{j-1,1}}{\theta_{t_{j-1}} f_{t_{j-1},1}},
$
$
g_{t_j,0} = \kappa_{j,0} + \bigl(\kappa_{j,0} - \kappa_{j+1,0}\bigr) \frac{m_{j+1,0}}{\theta_{t_j} f_{t_j,0}};
$
previously solved coefficients enter each new equation only through the already-satisfied previous one; by induction every cell-arm equation holds, and each displayed ratio is the previously solved mixture over the entering type's arm density, bounded above and below by the bounded-likelihood-ratio hypothesis. For the centering, set $a_T(t) = \int g_{t,d}(y) f_{t,d}(y)\,dy$ and $b_{t,d} = g_{t,d} - a_T(t)$, so that $\int b_{t,d} f_{t,d}\,dy = 0$ by construction. The definition is arm-consistent: integrating the updates gives $\theta_{t_k}\{\int g_{t_k,1} f_{t_k,1} - \int g_{t_k,0} f_{t_k,0}\} = S_{k+1} - S_k$ with $S_j \equiv \int \kappa_{j,1}\, m_{j,1} + \int \kappa_{j,0}\, m_{j,0} = \E_{P_0}[\chi \mid \bZ = z^{(j)}] = 0$, and $\sum_t \theta_t\, a_T(t)$ telescopes to $S_K = 0$. The recursion therefore writes each $b_{t,d}$ as a finite linear combination, with bounded coefficients, of cell influence functions evaluated in the type's solving cell; the cell influence functions lie in $L^2(P_0)$, and the ratio bounds transfer square-integrability between each cell mixture and each positive-weight type's arm density, hence every $b_{t,d}$ lies in $L^2(F_{t,d})$ and $\sum_{t,d}\theta_t\|b_{t,d}\|^2_{L^2(F_{t,d})}$ is finite. The lifted structural score solving~\eqref{eq:cell-lift-eq-paper} is square-integrable and hence, by the truncation argument in the preceding proof, lies in $\overline{\mathcal{G}}$; because it equals $\chi_z^D$ (respectively $\chi_z^Y$) by construction, both cell influence functions lie in $\overline{\mathcal{G}}$, which verifies the membership requirement of Condition~\ref{cond:cell-lift-paper}.
\end{proof}

\noindent In particular, the PRTE representers $\xi^{\star,\mathrm{PRTE}}_{t_k}$ and $\zeta^{\mathrm{PRTE}}_{t_k}$ of Appendix~\ref{app:proof_canonical_targets} are finite combinations of cell-propensity and cell-mass influence functions and lie in $\overline{\mathcal{G}}$; the Riesz representer $\Riesz_{\omega^\dagger,\ell}$ is then the LATE representer.

\subsection{Uniform coverage}\label{app:uniform-coverage}

\begin{assumption}[Uniform structural LAN for a restricted LATE score class]\label{ass:late-regularity-uniform}
For every $C > 0$, fix structural coordinate score classes $\mathcal{A}_Z(C)$, $\mathcal{A}_{\Theta}(C)$, and $\mathcal{A}_{r,d}(C)$ as follows. Elements $a^Z\in\mathcal{A}_Z(C)$ satisfy $\sum_z\mu_z(P_0)a^Z(z)=0$; elements $a^\Theta\in\mathcal{A}_{\Theta}(C)$ satisfy $\sum_{r\in\mathcal{T}_+(P_0)}\theta_r(P_0)a^\Theta(r)=0$ and perturb only the relative simplex on $\mathcal{T}_+(P_0)$; elements $a^F_{r,d}\in\mathcal{A}_{r,d}(C)$ lie in $L^2_0(F_{r,d})\cap L^\infty(F_{r,d})$ for positive-support types. All these classes have a common envelope $M(C)<\infty$. Let $\mathcal{H}^{\mathrm{unif}}_C(P_0)$ be the nonempty class of observed-data scores generated from triples $\bigl(a^Z, a^\Theta, (a^F_{r,d})\bigr)$ of these structural coordinate scores by the LATE structural map, restricted to those with $\|h\|_{L^2(P_0)}\le C$. For every $h \in \mathcal{H}^{\mathrm{unif}}_C(P_0)$, there exists a LATE-compatible triangular submodel $\{P_{n,h}\}_{n\geq1}\subset\PLATE$ through $P_0$ whose structural coordinates are perturbed by the exponential chart in Appendix~\ref{app:proof_thm_lam_bound}, and the observed-data likelihood ratios satisfy, for every sequence $\{h_n\}\subset\mathcal{H}^{\mathrm{unif}}_C(P_0)$,
$\log \frac{dP_{n,h_n}^n}{dP_0^n} = \frac{1}{\sqrt n}\sum_{i=1}^n h_n(O_i) -\frac12\|h_n\|_{L^2(P_0)}^2 + o_{P_0}(1).$
\end{assumption}

\begin{assumption}[Uniform Hadamard differentiability of the weight map]\label{ass:omega-uniform-hadamard}
In addition to the baseline Hadamard differentiability of the weight map (Appendix~\ref{app:regularity}), the derivative is uniform on the restricted score classes from Assumption~\ref{ass:late-regularity-uniform}: for every $C > 0$ there exist $K_\omega(C) < \infty$ and $t_0(C)>0$ such that, for every $h \in \mathcal{H}^{\mathrm{unif}}_C(P_0)$ and every LATE-compatible path $P_{t,h}$ through $P_0$ with score $h$,
$\bigl\|\omega(P_{t,h}) - \omega(P_0) - t\, \dot\omega(h; P_0)\bigr\|_{\ell^\infty} \;\leq\; K_\omega(C)\, t^2 \quad \text{for all } |t|\leq t_0(C).$
\end{assumption}

\noindent The conditions above strengthen the pointwise coverage of Proposition~\ref{prop:var-estimation} to uniform validity over the restricted local class.

\begin{proposition}[Uniform coverage of the RT interval]\label{prop:uniform-coverage}
Under the conditions of Proposition~\ref{prop:var-estimation} and Assumptions~\ref{ass:late-regularity-uniform} and~\ref{ass:omega-uniform-hadamard}, with $V_\omega(P_0)>0$, for each fixed $C<\infty$ the RT Wald interval is uniformly valid over the restricted local class: $
  \liminf_{n \to \infty}\; \inf_{h \in \mathcal{H}^{\mathrm{unif}}_C(P_0)} \mathbb{P}_{n,h}^{n}\!\bigl(\beta^{\mathrm{RT}}_{\bomega}(P_{n,h}) \in \hat\beta^{\mathrm{RT}}_{\bomega} \pm z_{1-\alpha/2}\sqrt{\hat V_n/n}\bigr) \;\geq\; 1 - \alpha.
$
\end{proposition}

\begin{proof}
Fix $C > 0$ and write $\mathcal{H}_C = \mathcal{H}^{\mathrm{unif}}_C(P_0)$. Every $h \in \mathcal{H}_C$ obeys the $P_0$-a.s.\ envelope $\|h\|_{L^\infty(P_0)} \leq 3M(C)$: each structural block is bounded by the common envelope $M(C)$, and the type-posterior weights combining them sum to one. The envelope makes the uniform LAN of Assumption~\ref{ass:late-regularity-uniform} the standard consequence of uniform quadratic-mean differentiability of the bounded structural charts, verifiable in the finite designs of the previous subsection; the alternatives $P_{n,h}$ are built from the structural exponential chart of Appendix~\ref{app:proof_thm_lam_bound} and remain in $\PLATE$ for all large $n$. Suppose the claim fails: there are $\varepsilon > 0$, a subsequence $n_k \uparrow \infty$, and $h_{n_k} \in \mathcal{H}_C$ along which the coverage probability stays at distance at least $\varepsilon$ from $1 - \alpha$. Extending $\{h_{n_k}\}$ to a full sequence by arbitrary elements of $\mathcal{H}_C$, it is enough to show that the coverage probability converges to $1 - \alpha$ along every sequence $\{h_n\} \subset \mathcal{H}_C$.

\emph{Contiguity.} Along any subsequence there is a further subsequence, not relabeled, with $h_n \rightharpoonup h_\infty$ weakly in $L^2_0(P_0)$ and $\|h_n\|_{L^2(P_0)}^2 \to \sigma_h^2 \leq C^2$. The uniform LAN expansion and the triangular-array Lindeberg condition supplied by the envelope give $\log(dP_{n,h_n}^n/dP_0^n) \xrightarrow{d} N(-\sigma_h^2/2, \sigma_h^2)$ under $P_0$; the limiting likelihood ratio has expectation one, and Le Cam's first lemma yields $P_{n,h_n}^n \triangleleft P_0^n$ along every subsequence, hence along the original sequence: whenever $\mathbb{P}_{0}^{n}(A_n) \to 0$, also $\mathbb{P}_{n,h_n}^{n}(A_n) \to 0$.

\emph{Uniform linearization.} Under $P_0$, Proposition~\ref{prop:eif-decomp} gives $\sqrt n\,(\hat\beta^{\mathrm{RT}}_{\bomega} - \beta^{\mathrm{RT}}_{\bomega}(P_0)) = n^{-1/2}\sum_i \eif_{\omega(\cdot)}(O_i; P_0) + \rho_n$ with $\rho_n = o_{P_0}(1)$; $\rho_n$ is a deterministic polynomial in the $O_{P_0}(n^{-1/2})$ sample moments and the weight-map von Mises remainder and does not depend on $h_n$, and contiguity turns $\rho_n = o_{P_0}(1)$ into $\rho_n = o_{P_{n,h_n}}(1)$. On the target side, $\beta^{\mathrm{RT}}_{\bomega}(P_{n,h_n}) - \beta^{\mathrm{RT}}_{\bomega}(P_0) = n^{-1/2} \langle \eif_{\omega(\cdot)}, h_n \rangle_{L^2(P_0)} + q_n$, where the remainder has two uniform sources: a second-order Taylor bound for the tilted moment ratios $\Wald_\ell = \Cov(Y, Z_\ell)/\Cov(D, Z_\ell)$ under the bounded envelope, which keeps $|\Cov_{P_{n,h_n}}(D, Z_\ell)| \geq \gamma_\ell(P_0)/2$ for all large $n$ and bounds the Wald remainder by $K_{\Wald}(C)/n$ uniformly over $\mathcal{H}_C$, with $\phi_\ell$ the representer of Lemma~\ref{lem:wald-pathwise}; and the uniform Hadamard bound of Assumption~\ref{ass:omega-uniform-hadamard}, which at $t = n^{-1/2}$ bounds the $\omega$-remainder by $K_\omega(C)/n$. Through the product rule for $\bomega(P)'\bWald(P)$,
$|q_n| \;\leq\; \bigl\{\|\bomega(P_0)\|_{\ell^1} K_{\Wald}(C) + \|\bWald(P_0)\|_{\ell^1} K_\omega(C)\bigr\}/n + O(n^{-1})$; hence $\sqrt n\, q_n \to 0$ at a rate independent of the sequence. Combining the two expansions, $\sqrt n\,(\hat\beta^{\mathrm{RT}}_{\bomega} - \beta^{\mathrm{RT}}_{\bomega}(P_{n, h_n})) \;=\; \frac{1}{\sqrt n}\sum_{i=1}^n \eif_{\omega(\cdot)}(O_i; P_0) \;-\; \langle \eif_{\omega(\cdot)}, h_n \rangle_{L^2(P_0)} \;+\; \tilde\rho_n,
$ with $\tilde\rho_n = o_{P_{n, h_n}}(1)$.

\emph{Gaussian limit.} Along the weakly convergent subsequence, the Cram\'er--Wold device and the Lindeberg condition give the joint Gaussian limit of $(n^{-1/2}\sum_i \eif_{\omega(\cdot)}(O_i; P_0), \log(dP_{n,h_n}^n/dP_0^n))$ with cross-covariance $\langle \eif_{\omega(\cdot)}, h_\infty \rangle$, and Le Cam's third lemma \citep{vanderVaart1998} gives, under $P_{n,h_n}$,
$n^{-1/2}\textstyle\sum_i \eif_{\omega(\cdot)}(O_i; P_0) - \langle \eif_{\omega(\cdot)}, h_n \rangle_{L^2(P_0)} \;\xrightarrow{d}\; N(0, V_\omega(P_0)).$
The limit is subsequence-free: the weak limit $h_\infty$ enters only through the shift, which the centering removes, and $\sigma_h^2$ enters only the likelihood coordinate, which the centered quantity does not see. Since every subsequence admits a further weakly convergent subsequence with the same limit, the convergence holds along the full sequence, and apply Slutsky's theorem on the linerarization gives $\sqrt n\,(\hat\beta^{\mathrm{RT}}_{\bomega} - \beta^{\mathrm{RT}}_{\bomega}(P_{n, h_n})) \xrightarrow{d, P_{n, h_n}} N(0, V_\omega(P_0))$ along any $\{h_n\} \subset \mathcal{H}_C$.

\emph{Pivot.} The consistency $\hat V_n \xrightarrow{p, P_0} V_\omega(P_0)$ of Proposition~\ref{prop:var-estimation} transfers to $P_{n,h_n}$ by contiguity. The limiting variance is positive by hypothesis. Slutsky's theorem makes the studentized pivot $\sqrt n\,(\hat\beta^{\mathrm{RT}}_{\bomega} - \beta^{\mathrm{RT}}_{\bomega}(P_{n,h_n}))/\sqrt{\hat V_n}$ asymptotically standard normal along the sequence, and $\{|t| \leq z_{1-\alpha/2}\}$ is a continuity set of the standard normal; hence the coverage probability converges to $1 - \alpha$ along every sequence in $\mathcal{H}_C$, a contradiction.
\end{proof}

\section{RT and the GMM Weighting-Matrix Class}\label{app:witness-class-supp}

\subsection{The complete family of RT-implementing maps}
\begin{proposition}[RT-implementing weighting matrix maps]\label{prop:witness-class}
Maintain Assumption~\ref{ass:relevance} and a locally interior target: $\omega_\ell(P) > 0$ for every $\ell$ on a neighborhood $\mathcal{U}_0$ of $P_0$, with $\gamma_\ell(P) \neq 0$ throughout (shrinking $\mathcal{U}_0$ if needed). The implicit weights $\blambda(\bW; P) = \bgamma(P) \odot \bv / \bigl(\bgamma(P)'\bv\bigr)$ depend on $\bW$ only through the \emph{effective weight vector} $\bv = \bW(P)\,\bgamma(P)$. A map $\bW(\cdot)$ with $\bW(P) \succ 0$ on $\mathcal{U}_0$ yields the target weights, $\blambda(\bW(P); P) = \bomega(P)$ for every $P$, if and only if
$\bW(P)\,\bgamma(P) \;\propto\; \bomega(P) \oslash \bgamma(P).$
\end{proposition}

\begin{proof}
From Proposition~\ref{prop:weighted_iv}, $\lambda_\ell(\bW;P)=\gamma_\ell v_\ell/(\bgamma'\bv)$ with $\bv=\bW(P)\bgamma(P)$; hence the implicit weights depend on $\bW$ only through $\bv$, and $\bgamma'\bv=\bgamma'\bW\bgamma>0$ by positive definiteness. Fix $P\in\mathcal{U}_0$. If $\blambda(\bW(P);P)=\bomega(P)$, then $\gamma_\ell v_\ell=(\bgamma'\bv)\,\omega_\ell$ for every $\ell$, and dividing by $\gamma_\ell\neq0$ gives $\bv=(\bgamma'\bv)\,\bomega\oslash\bgamma$, a positive multiple of $\bomega\oslash\bgamma$. Conversely, if $\bv=c \,\bomega\oslash\bgamma$ for some $c \neq0$, then $\bgamma'\bv=c \sum_\ell\omega_\ell=c >0$ and $\lambda_\ell=\gamma_\ell(c\,\omega_\ell/\gamma_\ell)/c=\omega_\ell$ for every $\ell$. The equivalence holds pointwise on $\mathcal{U}_0$, which is the claim.
\end{proof}

\FloatBarrier
\section{Additional STAR Results}\label{app:star_robustness}

Tables~\ref{tab:star_attrition} and~\ref{tab:star_robustness} report attrition and robustness.

\begin{table}[t]
\centering
\caption{Attrition Analysis: Analysis Sample vs.\ Excluded Observations}
\label{tab:star_attrition}
\begin{tabular}{lccccc}
\toprule
 & In sample & Excluded & Difference & SE & $p$-value \\
\midrule
Female & 0.488 & 0.478 & 0.010 & (0.013) & 0.427 \\
African American & 0.331 & 0.348 & -0.017 & (0.037) & 0.652 \\
White & 0.663 & 0.646 & 0.016 & (0.037) & 0.662 \\
\midrule
Treatment share & 0.457 & 0.529 & & & \\
$N$ & 15,056 & 2,150 & & & \\
\bottomrule
\end{tabular}
\begin{minipage}{0.9\textwidth}
\footnotesize\textit{Notes:} Analysis sample: the 15,056 student$\times$grade observations entering the estimation: small or regular arm, non-missing math score, school with at least 10 students and 3 per arm, student observed in a single school, and a school$\times$grade cell with both arms. Excluded: all other small- or regular-arm observations. Standard errors for the difference, clustered at the school level, in parentheses.
\end{minipage}
\end{table}
\begin{table}[t]
\centering
\caption{Robustness Across Outcomes and Samples}
\label{tab:star_robustness}
\footnotesize\setlength{\tabcolsep}{2pt}
\begin{tabular}{lccccccccc}
\toprule
 & $N$ & $L$ & 2SLS & TSGMM & EGMM & CUE & CS-ATE & $J$ & $p$ \\
\midrule
Math & 15,056 & 80 & 0.224 & 0.205 & 0.172 & 0.176 & 0.224 & 136.2 & $< 0.001$ \\
Reading & 14,876 & 80 & 0.238 & 0.211 & 0.173 & 0.173 & 0.239 & 125.3 & $< 0.001$ \\
Math, pre-switch & 14,357 & 80 & 0.226 & 0.202 & 0.168 & 0.166 & 0.225 & 134.7 & $< 0.001$ \\
\bottomrule
\end{tabular}
\begin{minipage}{0.97\textwidth}
\footnotesize\textit{Notes:} Each row is a separate specification on the single STAR cohort. Outcome: Stanford Achievement Test scaled score for the indicated subject, standardized within grade. Small class (13--17 students) versus regular class (22--25 students). All specifications use within-school$\times$grade FWL and school-specific assignment comparisons, with the second-moment matrix $\hat\bOmega$ two-way clustered at the classroom and student level.
\end{minipage}
\end{table}
 
\subsection{Many-instruments robustness}\label{app:star_mib}

Table~\ref{tab:star_mib} verifies that the 2SLS--EGMM gap is not a many-instruments artifact.

\begin{table}[t]
\centering
\caption{Many-Instruments Robustness}
\label{tab:star_mib}
\begin{tabular}{lccccc}
\toprule
 & 2SLS & LIML & JIVE & LOO-IV & EGMM \\
\midrule
Estimate & 0.224 & 0.224 & 0.224 & 0.224 & 0.172 \\
 & (0.027) & (0.027) & (0.027) & (0.037) & (0.031) \\
\midrule
$L/N$ & \multicolumn{5}{c}{0.005 \;(80 / 15,056)} \\
$\hat\kappa_{\text{LIML}}$ (LIML eigenvalue) & & 1.0307 & & & \\
$D_{\text{dm}} \in \text{col}(Z)$ & \multicolumn{5}{c}{Yes (first stage exact)} \\
\bottomrule
\end{tabular}
\begin{minipage}{0.92\textwidth}
\footnotesize\textit{Notes:} Effect of attending a small class (13--17 students) versus a regular class (22--25 students) on the grade-standardized Stanford math score, single STAR cohort. $L = 80$ within-school assignment comparisons ($Z_s = \tilde{D} \cdot \mathbf{1}\{\text{school} = s\}$), $N = 15,056$ student$\times$grade observations. All specifications use within-school$\times$grade FWL; standard errors are two-way cluster-robust at the classroom and student level (jackknife SE for LOO-IV).
\end{minipage}
\end{table}
 
\subsection{Calibrated simulation}\label{app:mc_star}

The compensatory-PRTE estimator of Section~\ref{sec:star} is validated on a stacked-STAR DGP calibrated to the estimation sample ($L = 80$ schools) and drawn at twice its size each replication. The cross-fit point retains a bias from the classroom and cohort shocks common to both half-samples and from the convex weight function applied to the noisily estimated baselines; the doubled size keeps it small relative to the interval widths, and the coverage comparison isolates the standard errors.
 Per school $\ell$, from the estimation sample: enrollment share $\theta_\ell$, assignment probability $p_\ell$, regular-class baseline $\mu_{0\ell} = \E[Y(0) \mid S = \ell]$, school effect $\LATE_\ell$, and the two arm variances. The within-school residual variance splits into a persistent student component, a classroom component, a school$\times$grade$\times$year component, and idiosyncratic noise.

In each replication: First, resample student records with replacement to twice the realized size (two independent cohort-year blocks); enrollment shares vary across replications. Then, draw $Y_i = \mu_{0\ell} + \LATE_\ell D_i + \alpha_{s(i)} + \gamma_{c(i)} + \delta_{\ell g} + e_i$ with the calibrated component variances: $\alpha$ persistent student (one draw per student), $\gamma$ classroom, $\delta$ school$\times$grade$\times$year (common to both arms), $e$ idiosyncratic. Next, Residualize $Y, D$ on school$\times$grade-year cells (FWL) and set $Z_{\ell, i} = D^{dm}_i \ind\{S_i = \ell\}$, giving $\bSigmaZ$ diagonal by school. Then, Cross-fit the compensatory PRTE: each school's baseline and Wald come from independent half-samples of its students ($40$ splits averaged). Form the naive standard error (weights fixed, Wald channel only) and the full-influence-function standard error, both two-way cluster-robust at the classroom and student level. Finally, score the nominal-$95\%$ intervals against the population target $\sum_\ell \omega_\ell(\mu_{0\ell})\,\LATE_\ell$.

\section{Additional Patent Results}\label{app:patent_supp}

\subsection{Patent examiner design: supplementary tables}

Table~\ref{tab:patent_groups} summarizes the seven leniency groups.

\begin{table}[t]
\centering
\caption{Examiner Leniency Group Summary Statistics}
\label{tab:patent_groups}
\begin{tabular}{lcccccc}
\toprule
 & $N$ & Examiners & Lenience & Approved & Citations & Follow-on \\
\midrule
G1 (Strict) & 4,920 & 1,683 & 0.098 & 0.263 & 1.90 & 1.13 \\
G2 & 4,930 & 1,792 & 0.347 & 0.511 & 3.84 & 2.01 \\
G3 & 4,908 & 1,329 & 0.514 & 0.613 & 4.55 & 2.16 \\
G4 & 4,919 & 1,120 & 0.615 & 0.692 & 3.88 & 2.13 \\
G5 & 4,920 & 1,073 & 0.690 & 0.758 & 7.87 & 2.63 \\
G6 & 4,918 & 973 & 0.756 & 0.820 & 6.65 & 2.74 \\
G7 (Lenient) & 4,919 & 944 & 0.845 & 0.861 & 14.35 & 4.16 \\
\midrule
Total & 34,434 & 5,915 & & 0.646 & 6.15 & 2.42 \\
\bottomrule
\end{tabular}
\begin{minipage}{0.92\textwidth}
\footnotesize\textit{Notes:} Examiners grouped into seven quantiles by leave-one-out approval rate. Group-row examiner counts are distinct examiners appearing in that application-level leniency group. Citations: total citations received by all subsequent patent applications filed by the same startup, counted over the five years following each application's public disclosure. Follow-on: total patent applications filed after the first-action decision on the startup's first application.
\end{minipage}
\end{table}
 
\subsection{First-stage monotonicity diagnostics}\label{app:patent_monotonicity_diag}

Table~\ref{tab:patent_binned_monotonicity} reports first-stage monotonicity diagnostics for the seven-bin patent cumulative design with art-unit-by-year cells as the conditioning vector $X_i$.

\begin{table}[t]
\centering
\caption{First-stage monotonicity diagnostics for the patent cumulative-leniency design}
\label{tab:patent_binned_monotonicity}
\small
\begin{tabular}{l r c c}
\toprule
Subgroup & $N$ & Single-threshold min $\hat\beta_\ell$ & Joint cumulative min $\hat\beta_\ell$ \\
\midrule
All & 34,434 & $0.132$ (0.009) & $0.036$ (0.009) \\
\midrule
Early cohort (2001--04) & 20,464 & $0.122$ (0.010) & $0.036$ (0.011) \\
Late cohort (2005--09) & 13,970 & $0.151$ (0.016) & $0.036$ (0.017) \\
\midrule
TC16 Biotech & 2,468 & $0.110$ (0.057) & $0.021$ (0.076) \\
TC17 Chemical & 3,659 & $0.094$ (0.034) & $0.003$ (0.023) \\
TC21--28 Comp/EE/IT & 10,947 & $0.142$ (0.012) & $0.036$ (0.021) \\
TC36 Business methods & 9,165 & $0.152$ (0.018) & $0.023$ (0.024) \\
TC37 Mechanical & 8,195 & $0.115$ (0.019) & $0.008$ (0.017) \\
\bottomrule
\end{tabular}
\begin{minipage}{0.95\textwidth}
\footnotesize\textit{Notes:} First-stage regressions of patent approval on the cumulative-threshold leniency instruments $Z_\ell = \ind\{G \geq \ell + 1\}$ for $\ell = 1, \ldots, 6$, with art-unit$\times$year fixed effects. \emph{Single-threshold}: six separate regressions, one threshold at a time; ``min $\hat\beta_\ell$'' is the smallest of the six (the proof-based diagnostic of Appendix~\ref{app:cumulative_monotone}). \emph{Joint cumulative}: one regression on all six thresholds simultaneously; the partial coefficients aggregate cell-level adjacent increments by FWL (supplemental evidence of Appendix~\ref{app:cumulative_monotone}). The parenthetical entry is the examiner-clustered standard error.
\end{minipage}
\end{table}
 
\subsection{Robustness across specifications}\label{app:patent_robust}
 Figure~\ref{fig:patent_jtest_robustness} re-partitions examiners into $Q$ groups and recomputes the cluster-robust $J$-statistic for each $Q$.

\begin{figure}[!h]
  \centering
  \includegraphics[width=0.85\textwidth]{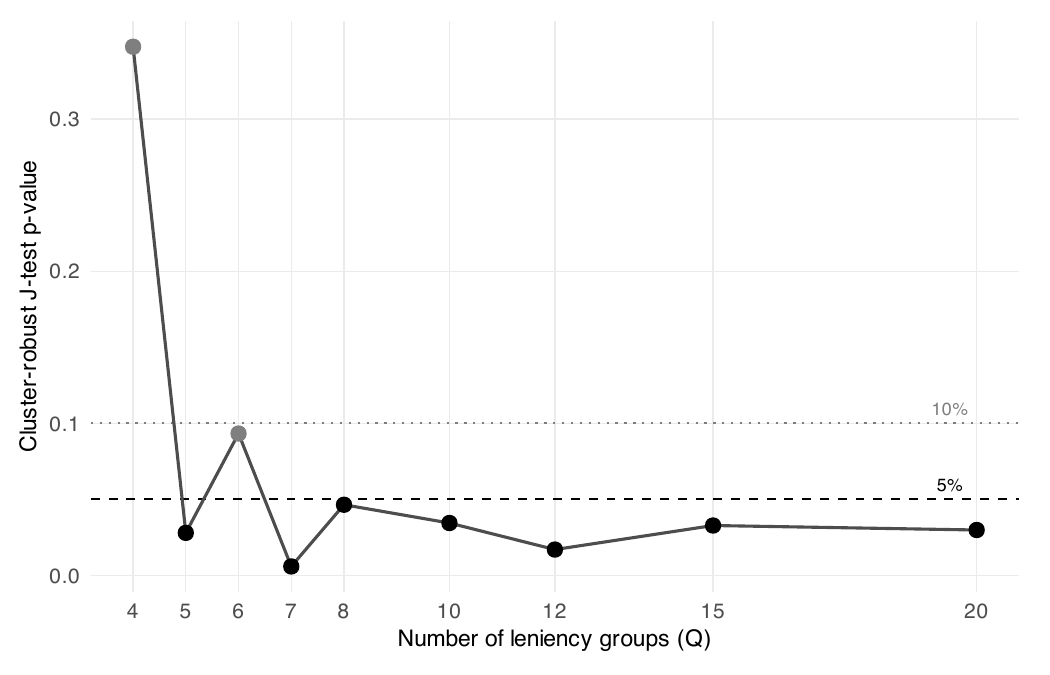}
  \caption{Cluster-robust $J$-diagnostic $p$-value against the number of leniency groups $Q$; the dashed and dotted lines mark the $5\%$ and $10\%$ levels. }
  \label{fig:patent_jtest_robustness}
\end{figure}

\subsection{Policy and Wald reconstruction weights}\label{app:patent_prte_conditions}

The policy and its Wald reconstruction load the same six observed leniency transitions but can weight art-unit-by-year cells differently. We therefore maintain the following condition. To state the condition, let $x$ index an art-unit-by-application-year cell, let $\mu_x=\Pr(X=x)$ and $s_{xg}=\Pr(G=g\mid X=x)$, and let $\mathsf T(x)$ denote the broad technology group containing cell $x$. We estimate $D_i=\alpha^D_{x(i)}+\sum_{g=2}^{7}\delta_{\mathsf T(x(i)),g}\ind\{G_i=g\}+v_i$, normalizing $\delta_{\mathsf T(x),1}=0$ for every $x$. We index each observed transition by its destination group: $\theta_{xj}\equiv\theta_{\mathsf T(x),j}=\delta_{\mathsf T(x),j}-\delta_{\mathsf T(x),j-1}$, $j=2,\ldots,7$. The detailed decomposition below maintains this pooled first-stage model for $\E[D\mid G,X]$ on the observed cell--group support. All raw estimated adjacent gains are positive. The normalized observed staircase weight on cell--transition pair $(x,j)$ is $h^{\mathrm{obs}}_{xj}=\mu_xs_{x,j-1}\theta_{xj}/\{\sum_{x'}\sum_{r=2}^{7}\mu_{x'}s_{x',r-1}\theta_{x'r}\}$, and its aggregate transition weight is $\psi^{\mathrm{obs}}_j=\sum_x h^{\mathrm{obs}}_{xj}$, $j=2,\ldots,7$. Let $b_{xj}$ denote the average effect of approval for the cell-$x$ applicants induced into approval by the transition from group $j-1$ to group $j$. For threshold instrument $Z_\ell$, define its detailed Wald weight as $w_{\ell xj}=\mu_x\theta_{xj}C_{\ell,j-1,x}/\{\sum_{x'}\sum_{r=2}^{7}\mu_{x'}\theta_{x'r}C_{\ell,r-1,x'}\}$, where $C_{\ell k,x}=\Cov(Z_\ell,Z_k\mid X=x)$. Under the maintained IV model and pooled first-stage restriction, $\Wald_\ell=\sum_x\sum_{j=2}^{7}w_{\ell xj}b_{xj}$ and $A_{\ell j}=\sum_xw_{\ell xj}$. When $A$ has full rank, $\boldsymbol\rho^{\mathrm{obs}}=A^{-T}\bpsi^{\mathrm{obs}}$ is the unique linear combination of the six Wald estimands whose aggregate weights match the policy weights on every observed margin. In particular, define $g^{\mathrm{obs}}_{xj}=\sum_{\ell=1}^6\rho^{\mathrm{obs}}_\ell w_{\ell xj}$; then $\sum_xg^{\mathrm{obs}}_{xj}=\sum_xh^{\mathrm{obs}}_{xj}=\psi^{\mathrm{obs}}_j$. The exact maintained condition is
\begin{equation}
\underbrace{\sum_x\sum_{j=2}^{7}h^{\mathrm{obs}}_{xj}b_{xj}}_{\text{observed policy component}}
=
\underbrace{{\boldsymbol\rho^{\mathrm{obs}}}'\bWald}_{\text{Wald reconstruction}}
\quad\Longleftrightarrow\quad
\sum_x\sum_{j=2}^{7}(h^{\mathrm{obs}}_{xj}-g^{\mathrm{obs}}_{xj})b_{xj}=0.
\label{eq:patent_weighting_agreement}
\end{equation}
Equation~\eqref{eq:patent_weighting_agreement} permits arbitrary heterogeneity across cells and margins provided that it is orthogonal, on average, to the difference between the two weighting schemes.

The condition applies only to the six observed transitions. For the new group-$7$ to group-$8$ transition, the extrapolation sets $\theta_{x8}\equiv\frac{1}{6}\sum_{j=2}^{7}\theta_{xj}$. The staircase itself therefore fixes the frontier share $\eta=\{\sum_x\mu_xs_{x7}\theta_{x8}\}/\{\sum_x\mu_x[\sum_{j=2}^{7}s_{x,j-1}\theta_{xj}+s_{x7}\theta_{x8}]\}$ and its cell weights $\nu_{x8}=\mu_xs_{x7}\theta_{x8}/\{\sum_{x'}\mu_{x'}s_{x'7}\theta_{x'8}\}$. We assess~\eqref{eq:patent_weighting_agreement} with the diagnostic that asks whether cells weighted differently by the staircase policy and its reconstruction have systematically different approval effects. Using all policy cells, let $\widehat m_x^{\mathrm{obs}}=\widehat\mu_x^{-1}\sum_{j=2}^{7}\widehat h^{\mathrm{obs}}_{xj}$ denote the sample policy density, let $\widehat c_{\ell x}=\widehat{\Cov}(Z_\ell,D\mid X=x)$, and let $\widehat{\bar c}_\ell=\sum_x\widehat\mu_x\widehat c_{\ell x}$. The sample reconstruction density is $\widehat g_x^{\mathrm{obs}}=\sum_{\ell=1}^L(\widehat\rho^{\mathrm{obs}}_\ell/\widehat{\bar c}_\ell)\widehat c_{\ell x}$, and its difference from the policy density is $\widehat r_x^{\mathrm{obs}}=\widehat g_x^{\mathrm{obs}}-\widehat m_x^{\mathrm{obs}}$. For each outcome, we estimate the regression $Y_i
=
\alpha^Y_{x(i)}
+q_3(\widehat p_i)
+D_i\widetilde q_3(\widehat p_i)
+\gamma_{\mathrm{weight}}D_i(\widehat r_{x(i)}^{\mathrm{obs}}-\overline{\widehat r}^{\mathrm{obs}})
+\varepsilon_i,
$
where $\alpha^Y_{x(i)}$ are outcome-regression cell fixed effects, $\widehat p_i$ is the fitted first-stage index, $\overline{\widehat r}^{\mathrm{obs}}$ is the sample mean of $\widehat r_{x(i)}^{\mathrm{obs}}$, and $q_3$ and $\widetilde q_3$ are cubics. The flexible index terms allow both outcome levels and approval effects to vary with fitted leniency. The coefficient $\gamma_{\mathrm{weight}}$ asks whether the cell-weight difference is associated with larger or smaller approval effects. Table~\ref{tab:patent_prte_conditions} reports the diagnostic for both outcomes.

\begin{table}[!b]
\centering
\caption{Staircase policy--reconstruction weighting check}
\label{tab:patent_prte_conditions}
\small
\setlength{\tabcolsep}{3pt}
\begin{tabular}{lcc}
\toprule
 & Citations to subsequent applications & Follow-on applications \\
\midrule
Bootstrap $p$-value & 0.195 & 0.365 \\
\bottomrule
\end{tabular}
\begin{minipage}{0.97\textwidth}
\footnotesize\textit{Notes:} Patent sample, art-unit-by-year cells. Each entry tests $H_0\colon\gamma_{\mathrm{weight}}=0$. The $p$-values use an examiner-level pairs cluster bootstrap ($B=200$ for each outcome).
\end{minipage}
\end{table}
 
\subsection{Computation of identification-gap bounds and confidence regions}\label{app:patent_prte_bounds}

Let $\mathcal X_8=\{x:s_{x7}\theta_{x8}>0\}$ denote the cells with positive $7\to8$ policy flow. Every observed Wald allocation assigns zero weight to this new transition, whereas $\eta>0$, so the complete target lies outside the Wald span. Under~\eqref{eq:patent_weighting_agreement}, the identification gap is $\Delta=\{(1-\eta)\boldsymbol\rho^{\mathrm{obs}}-\bomega^\dagger\}'\bWald+\eta\sum_{x\in\mathcal X_8}\nu_{x8}b_{x8}$, where the first term is identified. The nonnegative envelope restricts only $b_{x8}\geq0$ on $\mathcal X_8$ and gives the one-sided bound $\Delta\in[\{(1-\eta)\boldsymbol\rho^{\mathrm{obs}}-\bomega^\dagger\}'\bWald,\infty)$. The increasing-MTE specification adds $b_{x2}\leq\cdots\leq b_{x7}$ in every cell and $b_{x7}\leq b_{x8}$ on $\mathcal X_8$. The sensitivity rows further impose $b_{x8}\leq c\max_\ell|\widehat\Wald_{\ell}|$.

Under the full-rank condition above and Assumption~\ref{ass:strict-complementarity}, the lower endpoint under the nonnegative envelope, $\delta_0=\{(1-\eta)\boldsymbol\rho^{\mathrm{obs}}-\bomega^\dagger\}'\bWald$, is smooth and supplies an internal influence-function check. A non-singleton optimal face can make an LP value map directionally, rather than fully, Hadamard differentiable. We instead apply the numerical directional delta method of \citet{HongLi2018} to the complete endpoint map, subject to the conditions below. Let $\vartheta$ collect the joint cell-by-group probabilities, the technology-specific adjacent first-stage gains, and the six outcome Wald moments. For each Gaussian examiner-cluster multiplier direction $G^*$, the same multipliers move every component of $\vartheta$, preserving their covariance. With small $\epsilon_n$, a positive probability coordinate follows $\widehat p_j(\epsilon_n)=\widehat p_j\{1+\epsilon_nG^*_{p,j}/(2\widehat p_j)\}^2/\sum_k\widehat p_k\{1+\epsilon_nG^*_{p,k}/(2\widehat p_k)\}^2$, a positive first-stage coordinate follows $\widehat\theta_j(\epsilon_n)=\widehat\theta_j\{1+\epsilon_nG^*_{\theta,j}/(2\widehat\theta_j)\}^2$, and $\widehat W_{\ell,Y}(\epsilon_n)=\widehat W_{\ell,Y}+\epsilon_nG^*_{W,\ell,Y}$.

For every draw, we reconstruct all detailed Wald weights, policy flows, the weighting-agreement equality, and $M_Y$ from these perturbed primitives, recompute the convex projection, and solve both endpoints of every applicable LP. Thus the endpoint direction is $\{\Phi(\widehat\vartheta(\epsilon_n))-\Phi(\widehat\vartheta)\}/\epsilon_n$. The step obeys $\epsilon_n\to0$ and $\sqrt n\epsilon_n\to\infty$. For inference, we maintain that the complete endpoint map $\Phi$ is Hadamard directionally differentiable at $\vartheta_0$, tangentially to the support of the examiner-cluster Gaussian limit, and that the multiplier law consistently estimates that limit. We also maintain a fixed finite dictionary, locally stable cell, leniency-group, technology-profile, and frontier-support labels, local feasibility and finiteness of each reported endpoint, and that every primitive coordinate retained as positive is locally bounded away from zero. Writing $R_{\widehat\vartheta}(\epsilon h)$ for the displayed joint path, these conditions give $\sup_{h\in K}\| \{R_{\widehat\vartheta}(\epsilon_nh)-\widehat\vartheta\}/\epsilon_n-h\|\to_p0$ for every compact tangent set $K$, so the path is first-order equivalent to the additive perturbation in \citet{HongLi2018}. The numerical directional delta method therefore consistently estimates the pointwise joint limit law. One-sided regions use the 95th percentile of the lower-endpoint direction; finite regions use the 95th percentile of the maximum of the standardized lower- and upper-endpoint errors.

\end{document}